\documentclass[12pt,draftclsnofoot,onecolumn]{IEEEtran}
\usepackage{stfloats}
\usepackage{times}
\usepackage[final]{graphicx}
\usepackage{amsmath,amsfonts}
\usepackage{amssymb,amsbsy}

\usepackage{epsf,graphics,epsfig,subfigure}
\usepackage{cite}
\usepackage{multirow}

\newcommand\ignore[1]{}
\newcommand{\argmax}{\operatornamewithlimits{arg\ max}}

\newcommand{\hsp}{\hspace{0.1in} }
\newcommand{\hspp}{\hspace{0.05in} }
\newcommand{\hsppp}{\hspace{0.02in} }

\newcommand{\snr}{ {\sf SNR}  }

\newcommand{\iid} {  {\sf{iid}} }

\newcommand{\bggi} { {\bf h}_{ {\sf iid}, \hsppp i} }

\newcommand{\bI} {  {\mathbf{I}} }

\newtheorem{lem}{Lemma}

\newtheorem{defn}{Definition}
\newtheorem{prp}{Proposition}
\newtheorem{thm}{Theorem}

\newtheorem{conj}{Conjecture}

\newcommand{\bha} { {  \mathbf{h}  } }

\newsavebox{\savepar}

\begin{document}
\title{Statistical Beamforming on the Grassmann Manifold for the
Two-User Broadcast Channel}
\author{\large {\hspace{0.1in}}
Vasanthan Raghavan$^{\star}$, Stephen V.\ Hanly, Venugopal V.\ Veeravalli
\thanks{V.~Raghavan is with the Department of Mathematics, University of
Southern California, Los Angeles, CA 90089, USA. He was with the
Coordinated Science Laboratory, University of Illinois at Urbana-Champaign,
Urbana, IL 61801, USA and the Department of Electrical and Electronic
Engineering, The University of Melbourne, Parkville, VIC 3052, Australia
when parts of this work was done. S.~V.~Hanly is with the Department of
Electrical and Computer Engineering, National University of Singapore,
Singapore 117576. V.~V.~Veeravalli is with the Coordinated Science Laboratory
and the Department of Electrical and Computer Engineering, University of
Illinois at Urbana-Champaign, Urbana, IL 61801, USA.
Email: {\tt{vasanthan\_raghavan@ieee.org, elehsv@nus.edu.sg, vvv@illinois.edu.}}
$^{\star}$Corresponding author.}
\thanks{
This work has been supported in part by the ARC through grant DP-0984862,
the NUS through grant WBS \#R-263-000-572-133 and the NSF through grant
CNS-0831670 at the University of Illinois. This paper was presented in part
at the IEEE International Symposium on Information Theory, Austin, TX, 2010.}}

\topmargin=0.2in
\maketitle
\baselineskip 18pt
\setlength{\textheight}{23.4cm}

{\vspace{-0.2in}}
\begin{abstract}
\noindent A Rayleigh fading spatially correlated broadcast setting
with $M = 2$ antennas at the transmitter and two-users (each with a
single antenna) is considered. It is assumed that the users have perfect
channel information about their links whereas the transmitter has only
statistical information of each user's link (covariance matrix of the
vector channel). A low-complexity linear beamforming strategy that
allocates equal power and one spatial eigen-mode to each user is employed
at the transmitter. Beamforming vectors on the Grassmann manifold that
depend only on statistical information are to be designed at the transmitter
to maximize the ergodic sum-rate delivered to the two users. Towards this
goal, the beamforming vectors are first fixed and a closed-form expression
is obtained for the ergodic sum-rate in terms of the covariance matrices
of the links. This expression is non-convex in the beamforming vectors
ensuring that the classical Lagrange multiplier technique is not
applicable. Despite this difficulty, the optimal solution to this problem
is shown to be the solution to the maximization of an appropriately-defined
average signal-to-interference and noise ratio (${\sf SINR}$) metric for
each user. This solution is the dominant generalized eigenvector of a pair
of positive-definite matrices where the first matrix is the covariance
matrix of the forward link and the second is an appropriately-designed
``effective'' interference covariance matrix. In this sense, our work is a
generalization of optimal signalling along the dominant eigen-mode of the
transmit covariance matrix in the single-user case. Finally, the ergodic
sum-rate for the general broadcast setting with $M$ antennas at the
transmitter and $M$-users (each with a single antenna) is obtained in
terms of the covariance matrices of the links and the beamforming vectors.
\end{abstract}


\ignore{
The eventual goal of our work is on understanding the performance of
limited feedback schemes on MIMO downlink channels. Towards that eventual
goal, we take a small step in this work by considering a MISO downlink
with two users.

 4) We propose limited feedback schemes and our philosophy is bottom
to top, in comparison with RVQ schemes in vogue which are motivated by
perfect CSIT benchmarks (a top to bottom philosophy), 5) Our scheme is
based on answering the question whether a feedback scheme should focus on
maximizing the numerator in the sum rate vs minimizing the denominator
}

\section{Introduction}
\label{sec1}
The last fifteen years of research in wireless communications has seen the
emergence of multi-antenna signalling as a viable option to realize high
data-rates at practically acceptable reliability levels.
While initial work on multi-antenna design was primarily motivated by the
single-user
paradigm~\cite{foschini_blast,visotsky,goldsmith_review,bolsckei,venu_capacity},
more recent attention has been on the theory and practice of multi-user
multi-antenna communications~\cite{haardt1,gesbert_mumimo,sharif1,caire_review}.
The focus of this paper is on a broadcast setting that typically models a
cellular downlink. We study the multiple-input single-output (MISO) broadcast
problem where a central transmitter with $M$ antennas communicates with $M$
users in the cell, each having a single antenna. Under the assumption of
perfect channel state information (CSI) at both the transmitter and the user
ends, significant progress has been made over the last few years on
understanding optimal signalling that achieves the
sum-capacity~\cite{caire_shamai,pramodv,jindal1,jindal2,wei_yu,boche1} as well
as the capacity region~\cite{weingarten} of the multi-antenna broadcast
channel. The capacity-achieving {\em dirty-paper coding} scheme~\cite{costa}
pre-nulls interference from simultaneous transmissions by other users
to a specific user and hence results in a multiplexing gain of $M$.

Nevertheless, the high implementation complexity associated with dirty-paper
coding~\cite{dpc_complexity} makes it less attractive in standardization
efforts for practical systems. The consequent search for low-complexity
signalling alternatives that are within a fixed power-offset\footnote{Two
schemes are within a fixed power-offset if the difference in power level
necessary to achieve a fixed rate with the two schemes stays bounded
independent of the rate.} of the dirty-paper coding scheme has resulted in
an array of candidate linear (as well as non-linear) precoding
techniques~\cite{tomlinson_harashima,joham,boccardi,swindlehurst2,yoo_goldsmith,haardt,wiesel,coord_bf,chae2}.
In particular, a linear beamforming scheme that is developed as a
generalization of the single-user beamforming scheme has attracted
significant attention in the literature. Specifically, a scheme where the
transmitter allocates one eigen-mode to each user and shares the power
budget equally among all the users is the focus of this work.

If perfect CSI is available at both the ends, instantaneous nulls can be
created in the interference sub-space of each user (or interference can
be zeroforced) and thus this scheme remains order-optimal with respect to
the dirty-paper coding scheme~\cite{yoo_goldsmith}. However, the practical
utility of the linear beamforming scheme is dependent on how gracefully
its performance degrades with the quality of CSI at the transmitter. This
is because while reasonably accurate CSI can be obtained at the user end via
pilot-based training schemes, CSI at the transmitter requires either channel
reciprocity or reverse link feedback, both of which put an overwhelming
burden on the operating cost~\cite{caire_review}. In the extreme (and
pessimistic) setting of no CSI at the transmitter, the multiplexing gain
reduces to $1$ (that is, it is lost completely relative to the perfect CSI
case).

In practice, the channel evolves fairly slowly on a statistical scale and it
is possible to learn the spatial statistics\footnote{With a Rayleigh (or a
Ricean) fading model for the MISO channel, the complete statistical
information of the link is captured by the covariance matrix (or the mean
vector and the covariance matrix) of the vector channel.} of the individual
links at the transmitter with minimal cost. With only statistical information
at the transmitter, the interference cannot be nulled out completely and a
low-complexity decoder architecture that treats interference as noise is
often preferred. Initial works assume an identity covariance matrix for
all the users corresponding to an {\em independent and identically
distributed} (i.i.d.) fading process in the spatial
domain~\cite{haardt1,gesbert_mumimo,sharif1,caire_review}. However, this
model cannot be justified in practical systems that are often deployed in
environments where the scattering is localized in certain spatial directions
or where antennas are not spaced wide apart due to infrastructural
constraints~\cite{canonical_jayesh}.

While signalling design for the single-user setting under a very general
spatial correlation model is now
well-understood~\cite{foschini_blast,visotsky,goldsmith_review,bolsckei,venu_capacity},
the broadcast case where the channel statistics vary across users and
different users experience different covariance matrices has not received
much attention. In particular,~\cite{tareq} studies the problem where all
the users share a common non-i.i.d.\ transmit covariance matrix and captures
the impact of this common covariance matrix on the achievable rates.
In~\cite{kountouris}, the authors show that second-order spatial statistics
can be exploited to schedule users that enjoy better channel quality and hence
improve the overall performance of an opportunistic beamforming scheme. In the
same spirit, it is shown in~\cite{hammarwall1} and~\cite{hammarwall2} that
second-order moments of the channel in combination with instantaneous norm (or
weighted-norm) feedback is sufficient to extract almost all of the multi-user
diversity gain in a broadcast setting. Spatial correlation is exploited to
reduce the feedback overhead of a limited feedback codebook design
in~\cite{trivellato,vasanth_ita10,clerckx,ruben}.

\noindent {\bf \em Summary of Main Contributions:} With this background,
the main focus of this paper is to fill some of the gaps in understanding
the information-theoretic limits of broadcast channels with low-complexity
signalling schemes (such as linear beamforming) under practical assumptions
on CSI and decoder architecture. We study the simplest non-trivial version of
this problem corresponding to the two-user ($M = 2$) case.
We design
optimal beamforming vectors on the Grassmann manifold\footnote{Informally,
${\cal G}(M,1)$ denotes the space of all $M$-dimensional unit-norm beamforming
vectors modulo the phase of the first element of the vector. A more formal
definition is provided in Sec.~\ref{sec2} (Def.~\ref{def1}).} ${\cal G}(2,1)$
to maximize the ergodic sum-rate achievable with the linear beamforming scheme.

The first step to this goal is the computation of the ergodic sum-rate in
closed-form. For this, we develop insight into the structure of the density
function of the weighted-norm of beamforming vectors isotropically distributed
on ${\cal G}(2,1)$. Exploiting this knowledge, we derive an explicit
expression for the ergodic sum-rate in terms of the covariance matrices
of the users and the beamforming vectors.
This expression can be rewritten in terms of a certain generalized ``distance''
measure between the beamforming vectors. As a result of this complicated
non-linear dependence, the sum-rate is non-convex in the beamforming
vectors thus precluding the use of the classical Lagrangian approach to
convex optimization. Instead, a first-principles based technique is
developed where the beamforming vectors are decomposed along an appropriately
chosen (in general, non-orthogonal) basis. Exploiting this decomposition
structure, we obtain an upper bound for the ergodic sum-rate, which we show
is tight for a specific choice of beamforming vectors (see
Theorems~\ref{thm2} and~\ref{thm1}). This optimal choice is the dominant
generalized eigenvector\footnote{A generalized eigenvector generalizes the
notion of an eigenvector to a pair of matrices. A more technical definition
is provided in Sec.~\ref{sec4} (Def.~\ref{def2}). The dominant eigenvector
is the eigenvector corresponding to the dominant eigenvalue. Under the
assumption that the eigenvector is unit-norm, it is unique on ${\cal G}(M,1)$.}
of a pair of covariance matrices, with one of them being the covariance
matrix of the forward link and the other an appropriately-designed
``effective'' interference covariance matrix. The generalized eigenvector
structure is the solution to maximizing an appropriately-defined average
signal-to-interference and noise ratio (${\sf SINR}$) metric for each user
and thus generalizes our intuition from the single-user
case~\cite{foschini_blast,visotsky,goldsmith_review,bolsckei,venu_capacity}.
Table~I in the Conclusions section (Sec.~\ref{sec6}) summarizes
the structure of the optimal beamforming vectors under different
signal-to-noise ratio (${\sf SNR}$) assumptions.

While a generalized eigenvector solution has been obtained in the perfect CSI
case for the multiple-input multiple-output (MIMO) broadcast
problem~\cite{wiesel2,coord_bf} and the MIMO interference channel problem in
the low-interference regime~\cite{vannapu2}, to the best of our knowledge,
its appearance in the statistical setting is a first. A closely-related work
of ours~\cite{vasanth_isit10_intf} reports the optimality of the generalized
eigenvector solution for the statistical beamformer design in the MISO
interference channel setting with two antennas. We also extend our intuition
to the weighted ergodic sum-rate maximization problem~\cite{kobayashi} and
conjecture on the structure of the optimal beamforming vectors. Numerical
results justify our conjecture and the intuition behind it. Finally,
closed-form expression for the ergodic sum-rate in terms of the covariance
matrices of the links and the beamforming vectors are obtained in the general
$M$-user case.

\noindent {\bf \em Organization:} This paper is organized as follows.
With Section~\ref{sec2} explaining the background of the problem,
ergodic rate expressions in terms of the covariance matrices of the links
and beamforming vectors are obtained in Section~\ref{sec3} for the $M = 2$
case. The non-convex optimization problem of ergodic sum-rate maximization
is the main focus of Section~\ref{sec4} with the low- and the high-$\snr$
extremes providing insight for the development in the
intermediate-$\snr$ regime. The focus shifts to weighted ergodic sum-rate
maximization in Section~\ref{sec5}. In addition, sum-rate expressions are
generalized to the general $M$-user case and concluding remarks are
provided in Section~\ref{sec6}. Most of the proofs$/$details are relegated
to the Appendices.

\noindent {\bf \em Notation:} We use upper- and lower-case bold symbols
for matrices and vectors, respectively. The notations ${\bf \Lambda}$
and ${\bf U}$ are usually reserved for eigenvalue and eigenvector matrices
whereas ${\bf I}$ is reserved for the identity matrix
(of appropriate dimensionality). The $i$-th diagonal element of
${\bf \Lambda}$ is denoted by ${\bf \Lambda}_i$ while the $i$-th element
of a vector ${\bf x}$ is denoted by ${\bf x}(i)$. At times, we also use
$\lambda_1, \lambda_2, \cdots$ to denote the eigenvalues of a Hermitian
matrix, and these eigenvalues are often arranged in decreasing order as
$\lambda_1 \geq \lambda_2 \geq \cdots$. The Hermitian
transpose and inverse operations of a matrix are denoted by $( \cdot)^H$
and $(\cdot )^{-1}$ while the trace operator is denoted by ${\sf Tr}(\cdot)$.
The two-norm of a vector is denoted by the symbol $\| \cdot \|$. The
operator $E[\cdot]$ stands for expectation while the density function of
a random variable is denoted by the symbol ${p}(\cdot)$. The symbols
${\mathbb C}$ and ${\mathbb R}^+$ stand for complex and positive real
fields, respectively.
$X \sim {\cal CN}(\mu,\sigma^2)$ indicates that $X$ is a complex Gaussian
random variable with mean $\mu$ and variance $\sigma^2$.

\section{System Setup}
\label{sec2}
We consider a broadcast setting that models a MISO cellular downlink
with $M$ antennas at the transmitter and $M$ users, each with a single
antenna. We denote the $M \times 1$ vector channel between the transmitter
and user $i$ as $\bha_i, \hsppp i = 1, \cdots, M$. While different
multi-user communication strategies can be
considered~\cite{tomlinson_harashima,joham,boccardi,swindlehurst2,yoo_goldsmith,haardt,wiesel,coord_bf,chae2},
as motivated in Sec.~\ref{sec1}, the focus here is on a linear
beamforming scheme where the information-bearing signal $s_i$ meant for
user $i$ is beamformed from the transmitter with the $M \times 1$
unit-norm vector ${\bf w}_i$. We assume that $s_i$ is unit energy and the
transmitter divides its power budget\footnote{The practically motivated
power-control problem where different powers could be allocated to the
different users is a related problem, but it is not studied here.} of
$\rho$ equally across all the users. Equal power allocation is popular
in current-generation cellular standards where low-complexity schemes are
preferred. The received symbol $y_i$ at user $i$ is written as
\begin{eqnarray}
y_ i & = & \sqrt{ \frac{ \rho}{M} } \cdot
\bha_i^H \left( \sum_{i=1}^M {\bf w}_i s_i \right) + n_i, \hspp
i = 1, \cdots , M  
\end{eqnarray}
where $\rho$ is the transmit power and $n_i$ denotes the
${\mathcal{CN}}(0,1)$ complex Gaussian noise added at the receiver.

Initial works on the broadcast problem assume that $\bha_i$ is ergodic
and it evolves over time and frequency in an i.i.d.\ fashion, and is
spatially i.i.d. While the above assumption can be justified in the time
and frequency axes with a frame-based and multi-carrier signalling approach
(common in current-generation systems), it cannot be justified along the
spatial axis. This is because the channel variation in the spatial (antenna)
domain cannot be i.i.d.\ unless the antennas at the transmitter end are
spaced wide apart and the scattering environment connecting the transmitter
with the users is sufficiently rich~\cite{canonical_jayesh}. With this
motivation, the main emphasis of this work is on understanding the impact
of the users' spatial statistics on the performance of a linear beamforming scheme.

We assume a Rayleigh fading\footnote{While more general fading models
such as Ricean or Nakagami-$m$ models can be considered, this paper
focuses on the Rayleigh model alone.} (zero mean complex Gaussian) model
for the channel, which implies that the complete spatial statistics are
described by the second-order moments of $\{ \bha_i \}$. For the MISO model,
the channel $\bha_i$ of user $i$ can be written as
\begin{eqnarray}
\bha_i & = & {\bf \Sigma}_{i}^{1/2} \hsppp \bggi 
\label{bhai}
\end{eqnarray}
where $\bggi$ is an $M \times 1$ vector with i.i.d.\ ${\mathcal{ CN}}(0,1)$
entries and ${\bf{\Sigma}}_{i} \triangleq E \left[ \bha_i \bha_i^H \right]$
is the transmit covariance matrix corresponding to user $i$. Note
that~(\ref{bhai}) is the most general statistical model for $\bha_i$
under the MISO assumption.
With ${\bf \Sigma}_i = \bI$ for all users,~(\ref{bhai}) reduces to the
i.i.d.\ downlink model well-studied in the
literature~\cite{haardt1,gesbert_mumimo,sharif1,caire_review}.

Under the assumption of Gaussian inputs $\{s_i \}$, the instantaneous
information-theoretic\footnote{All logarithms are to base $e$ and all
rate quantities are assumed to be in nats$/$s$/$Hz in this work.} rate,
$R_i$, achievable by user $i$ with the linear beamforming scheme using a
mismatched\footnote{Here, the decoding rule is different from the optimal
decoding rule due to the presence of multi-user interference.}
decoder~\cite{lapidoth_mismatched} is given by
\begin{eqnarray}
\label{mission5}
R_i & = & \log \left( 1 + \frac{ \frac{\rho}{M} \cdot
| {\bf h}_i^H {\bf w}_i |^2  }
{ 1 + \frac{ \rho}{M} \cdot \sum_{j \neq i} | {\bf h}_i^H {\bf w}_j  |^2 }
\right) \\
& = & \underbrace{ \log \left( 1 + \frac{\rho}{M} \cdot
\sum_{j=1}^M | {\bf h}_i^H {\bf w}_j |^2 \right)}_{I_{i, \hsppp 1} }
- \underbrace{ \log \left( 1 + \frac{\rho}{M} \cdot
 \sum_{j \neq i} | {\bf h}_i^H {\bf w}_j |^2 \right)
}_{ I_{i,\hsppp 2} }.
\label{mission3}
\end{eqnarray}
With the spatial correlation model assumed in~(\ref{bhai}), we can write
$I_{i,\hsppp 1}$ as
\begin{eqnarray}
I_{i, \hsppp 1} &  = & \log \left( 1 + \frac{\rho}{M} \cdot
\bggi^H \hsppp {\bf \Sigma}_{i}^{1/2} \left( \sum_{j=1}^M {\bf w}_j
{\bf w}_j^H \right) {\bf \Sigma}_{i}^{1/2} \hsppp \bggi \right) \\
& = & \log \left( 1 +  \frac{\rho}{M} \cdot \bggi^H \hsppp {\bf V}_i
\hsppp {\bf \Lambda}_i \hsppp {\bf V}_i^H \bggi \right),
\label{citing1}
\end{eqnarray}
where we have used the following eigen-decomposition
in~(\ref{citing1}): 
\begin{eqnarray}
{\bf V}_i \hsppp {\bf \Lambda}_i \hsppp {\bf V}_i^H & = &
{\bf \Sigma}_{i}^{1/2} \left( \sum_{j=1}^M {\bf w}_j {\bf w}_j^H \right)
{\bf \Sigma}_{i}^{1/2} \label{eqna7}
\\
\label{mission1}
{\bf \Lambda}_i & = & {\sf diag} \big( [ {\bf \Lambda}_{i, \hsppp 1},
\hspp \cdots, \hspp {\bf \Lambda}_{i, \hsppp M}] \big), \hspp \hspp
{\bf \Lambda}_{i,\hsppp 1} \geq \cdots \geq {\bf \Lambda}_{i, \hsppp M}
\geq 0.
\end{eqnarray}
Similarly, we can write $I_{i, \hsppp 2}$ as
\begin{eqnarray}
I_{i, \hsppp 2} & = & \log \left( 1 +  \frac{\rho}{M} \cdot \bggi^H
\hsppp \widetilde{ {\bf V}} _i
\hsppp \widetilde{ {\bf \Lambda}}_ i \hsppp \widetilde{ {\bf V}}_i^H
\bggi \right)
\\ \label{eqna9}
\widetilde{ {\bf V}} _i \hsppp
\widetilde{ {\bf \Lambda}}_ i \hsppp \widetilde{ {\bf V}}_i^H
& = & {\bf \Sigma}_{i}^{1/2} \left( \sum_{j \neq i}
{\bf w}_j {\bf w}_j^H \right) {\bf \Sigma}_{i}^{1/2} \\
\label{mission2}
\widetilde{{ \bf \Lambda}}_i & = & {\sf diag} \big([
\widetilde{ {\bf \Lambda}} _{i, \hsppp 1},  \hspp \cdots, \hspp
\widetilde{ {\bf \Lambda}}_{i, \hsppp M}] \big),
\hspp \hspp \widetilde{{\bf \Lambda}}_{i,\hsppp 1} \geq \cdots
\geq \widetilde{ {\bf \Lambda}} _{i, \hsppp M} \geq 0.
\end{eqnarray}

The goal of this work is to maximize the throughput conveyed from the
transmitter
to the users by the choice of beamforming vectors. Specifically, the metric
of interest is the ergodic sum-rate, ${\cal R}_{\sf sum}$, achievable with
the linear beamforming scheme:
\begin{eqnarray}
{\cal R}_{\sf sum} \triangleq \sum_{i=1}^M E \left[ R_i \right].
\end{eqnarray}
For this, note that the achievable rate in~(\ref{mission5}) is
invariant to transformations of the form ${\bf w}_i \mapsto e^{j \theta}
{\bf w}_i$ for any $\theta$. Coupled with the unit-norm assumption for
${\bf w}_i$, the space over which optimization is performed is precisely
defined as follows.
\begin{defn} [{\bf \em Stiefel and Grassmann Manifolds~\cite{arias_smith}}]
\label{def1}
The uni-dimensional complex Stiefel manifold ${\sf St}(M,1)$ refers
to the unit-radius complex sphere in $M$-dimensions and is defined as
\begin{eqnarray}
{\sf St}(M,1)
= \left\{ {\bf x} \in {\mathbb C}^M : \|{\bf x}\| = 1\right\}.
\end{eqnarray}
The uni-dimensional complex Grassmann manifold ${\cal G}(M,1)$ consists
of the set of one-dimensional subspaces of ${\sf St}(M,1)$. Here, a
transformation of the form ${\bf x} \mapsto e^{j \theta} {\bf x}$ (for any
$\theta$) is treated as invariant by considering all vectors of the form
$e^{j \theta} {\bf x}$ (for some $\theta$) to belong to the one-dimensional
sub-space spanned by ${\bf x}$.
\endproof
\end{defn}
The optimization objective is then to understand the structure of the
beamforming vectors, $\{ {\bf w}_{i, \hsppp {\sf opt}} \}$, that maximize
${\cal R}_{\sf sum}$:
\begin{eqnarray}
{\bf w}_{i, \hsppp {\sf opt} }  
= \argmax \limits_{ {\bf w}_i \hsppp \in \hsppp {\cal G}(M,1) } {\cal R}_{\sf sum},
\hspp i = 1, \cdots , M.
\label{prob_defn}
\end{eqnarray}
In~(\ref{prob_defn}), the candidate beamforming vectors,
$ \left\{ {\bf w}_i \right\}$, depend only on the long-term statistics
of the channel, which (as noted before) in the MISO setting is the set
of all transmit covariance matrices, $\left\{ {\bf \Sigma}_i \right\}$.

Towards the goal of computing ${\cal R}_{\sf sum}$, we decompose
${\bf h}_{\iid, \hsppp i }$ into its magnitude and directional components
as ${\bf h}_{\iid, \hsppp i } = \| {\bf h}_{\iid, \hsppp i} \| \cdot
\widehat{ {\bf h} }_{\iid, \hsppp i}$. It is well-known~\cite{foschini_blast}
that $\| {\bf h}_{\iid, \hsppp i} \|^2$ can be written as
\begin{eqnarray}
\label{decomp1}
\| {\bf h}_{\iid, \hsppp i} \|^2 = \frac{1}{2} \sum_{j=1}^{2M} \chi_j^2
\end{eqnarray}
where $\chi_j^2$ is a standard (real) chi-squared random variable and
$\widehat{ {\bf h} }_{\iid, \hsppp i}$ is a unit-norm vector that
is isotropically distributed~\cite{arias_smith,hassibi1} on ${\cal G}(M,1)$. Thus,
we can rewrite $I_{i, \hsppp 1}$ and $I_{i, \hsppp 2}$ as
\begin{eqnarray}
I_{i, \hsppp 1} & = &
\log \left( 1 +  \frac{\rho}{M} \cdot \| {\bf h}_{ \iid, \hsppp i } \|^2
\cdot \widehat{ {\bf h} }_{\iid, \hsppp i}^H \hsppp {\bf V}_i
\hsppp {\bf \Lambda}_i \hsppp {\bf V}_i^H
\widehat{ {\bf h} }_{\iid, \hsppp i} \right) \\
I_{i, \hsppp 2} & = &
\log \left( 1 +  \frac{\rho}{M} \cdot \| {\bf h}_{ \iid, \hsppp i } \|^2
\cdot \widehat{ {\bf h} }_{\iid, \hsppp i}^H
\hsppp \widetilde{ {\bf V}} _i
\hsppp \widetilde{ {\bf \Lambda}}_ i \hsppp \widetilde{ {\bf V}}_i^H
\widehat{ {\bf h} }_{\iid, \hsppp i} \right).
\end{eqnarray}
Since the magnitude and directional information of an i.i.d.\
random vector are independent~\cite{hassibi1},
$E \left[I_{i,\hsppp 1} \right]$ and $E \left[I_{i,\hsppp 2} \right]$
can be further written as
\begin{eqnarray}
E\left[I_{i, \hsppp 1} \right] & = & E
_{ \| {\bf h}_{\iid, \hsppp i} \| } \left[
E_{ \widehat{\bf h}_{\iid, \hsppp i} }
\left[  \log \left( 1 + \frac{\rho}{M} \cdot
\| {\bf h}_{ \iid, \hsppp i } \|^2 \cdot
\widehat{ {\bf h} }_{\iid, \hsppp i}^H \hsppp
 {\bf \Lambda}_ i  \hsppp
\widehat{ {\bf h} }_{\iid, \hsppp i} \right)
\right] \right]
\label{eq_ii1}\\
E\left[I_{i, \hsppp 2} \right] & = & E
_{ \| {\bf h}_{\iid, \hsppp i} \| } \left[
E_{ \widehat{\bf h}_{\iid, \hsppp i} }
\left[  \log \left( 1 + \frac{\rho}{M} \cdot
\| {\bf h}_{ \iid, \hsppp i } \|^2 \cdot
\widehat{ {\bf h} }_{\iid, \hsppp i}^H \hsppp
\widetilde{ {\bf \Lambda}}_ i  \hsppp \widehat{ {\bf h} }_{\iid, \hsppp i} \right)
\right] \right],
\label{eq_ii2}
\end{eqnarray}
where we have also used the fact that a fixed\footnote{Note that the
unitary transformation is independent of the channel realization when
the beamforming vectors are dependent only on the long-term statistics of the
channel.} unitary transformation of an isotropically distributed vector
on ${\cal G}(M,1)$ does not alter its distribution.

\section{Rate Characterization: Two-User Case}
\label{sec3}
We now restrict attention to the special case of two-users ($M = 2$) and
focus on computing the ergodic information-theoretic rates given in~(\ref{eq_ii1})
and~(\ref{eq_ii2}) in closed-form. The following theorem computes the
ergodic rates as a function of the covariance matrices of the two
links (${\bf \Sigma}_1$ and ${\bf \Sigma}_2$), and the choice of beamforming
vectors (${\bf w}_1$ and ${\bf w}_2$).
\begin{thm}
\label{prop_basic_rate}
The ergodic information-theoretic rate achievable at user $i$ (where
$i = 1,2$) with linear beamforming in the two-user case is given as
\begin{eqnarray}
E \left[R_{i} \right] =
E \left[  I_{i, \hsppp 1} \right] - E \left[ I_{i, \hsppp 2}  \right]
=  \frac{  {\bf \Lambda}_{i, \hsppp 1} \hsppp
h \left(\frac {\rho {\bf \Lambda}_{i, \hsppp 1}}{2}  \right)
- {\bf \Lambda}_{i, \hsppp 2} \hsppp
h \left(\frac {\rho {\bf \Lambda}_{i, \hsppp 2}}{2}  \right)}
{  {\bf \Lambda}_{i, \hsppp 1} - {\bf \Lambda}_{i, \hsppp 2} }
 - h \left(\frac {\rho \widetilde{ {\bf \Lambda}} _{i, \hsppp 1}  }{2}  \right)
\end{eqnarray}
where $h(\bullet)$ is a monotonically increasing function defined as
\begin{eqnarray}
h(x) \triangleq \exp \left( \frac{1}{x} \right) \hsppp E_1 \left( \frac{1}{x}
\right), \hspp x \in (0, \infty)
\label{hx}
\end{eqnarray}
with $E_1(x) = \int_{x}^{\infty} \frac{e^{-t}}{t} dt$ denoting the
Exponential integral~\cite{abramowitz}. The corresponding eigenvalues (cf.~(\ref{eqna7})
and~(\ref{eqna9})) can be written in terms of ${\bf \Sigma}_i$ and the
beamforming vectors as follows:
\begin{eqnarray}
{\bf \Lambda}_{i, \hsppp 1} & = &
\frac{ A_i + B_i + \sqrt{ (A_i - B_i)^2 + 4C_i^2 }  }{2}  \\
{\bf \Lambda}_{i, \hsppp 2} & = &
\frac{ A_i + B_i - \sqrt{ (A_i - B_i)^2 + 4C_i^2 }  }{2} \\
\widetilde{ {\bf \Lambda}} _{i, \hsppp 1} & = & B_i,
\end{eqnarray}
where $A_i = {\bf w}_i^H {\bf \Sigma}_i {\bf w}_i$, $B_i = {\bf w}_j^H
{\bf \Sigma}_i {\bf w}_j$ and $C_i = |{\bf w}_i^H {\bf \Sigma}_i
{\bf w}_j|$ with $j \neq i$ and $\{ i,j \} = 1,2$.
\end{thm}
\begin{proof}
Since $E \left[ R_i \right] = E \left[ I_{i, \hsppp 1} \right] -
E \left[ I_{i, \hsppp 2} \right]$, we start by computing
$E\left[I_{i, \hsppp 1} \right]$. From~(\ref{eq_ii1}), we have
\begin{eqnarray}
E \left[I_{i, \hsppp 1} \right] =
E_{ {\bf X} } \left[ \int_{  y = {\bf \Lambda}_{i,\hsppp 2} }
^{{\bf \Lambda}_{i, \hsppp 1} }
\log \left( 1 + \frac{\rho}{2} \cdot x y  \right) { p}_i(y) dy \right] 
\end{eqnarray}
where ${\bf X}$ stands for the random variable ${\bf X} = 
\| {\bf h}_{ \iid, \hsppp i } \|^2$, $x$ is a realization of
${\bf X}$ and ${p}_i(y)$ denotes the density function of
\begin{eqnarray}
\widehat{ {\bf h} }_{\iid, \hsppp i}^H \hsppp {\bf \Lambda}_i \hsppp
\widehat{ {\bf h} }_{\iid, \hsppp i}
= \sum_{j=1}^2 {\bf \Lambda}_{i,\hsppp j} \hspp \Big| \widehat{
{\bf h} }_{\iid, \hsppp i}(j) \Big|^2,
\end{eqnarray}
evaluated at $y$ with ${\bf \Lambda}_{i,\hsppp 2} \leq y \leq
{\bf \Lambda}_{i,\hsppp 1}$. That is, a closed-form computation of
$E \left[ I_{i, \hsppp 1} \right]$ requires the density function of
weighted-norm of vectors isotropically distributed on ${\cal G}(2,1)$.
In Lemma~\ref{lemma_density2} of Appendix~\ref{app_cdf}, we show that
\begin{eqnarray}
{p}_i(y) = \frac{1}{ {\bf \Lambda}_{i, \hsppp 1} -
{\bf \Lambda}_{i, \hsppp 2} }, \hspp {\bf \Lambda}_{i, \hsppp 2} \leq
y \leq {\bf \Lambda}_{i, \hsppp 1}.
\end{eqnarray}
Using this information along with the chi-squared structure of
$\| {\bf h}_{ \iid, \hsppp i } \|^2$ (see~(\ref{decomp1})), we have
\begin{eqnarray}
\label{basic_eqn1}
E \left[ I_{i, \hsppp 1} \right] = \frac{1}{ {\bf \Lambda}_{i, \hsppp 1} -
{\bf \Lambda}_{i, \hsppp 2} } \cdot
\int_{x = 0}^{\infty}x e^{-x}
\int_{ y = {\bf \Lambda}_{i,\hsppp 2}}^{ {\bf \Lambda}_{i, \hsppp 1} }
\log \left( 1 + \frac{\rho}{2} x y   \right) dy \hsppp dx.
\end{eqnarray}
Integrating out the $y$ variable, we have
\begin{eqnarray}
E \left[ I_{i, \hsppp 1} \right] & = &
\frac{1}{ \frac{\rho}{2} \left(
{\bf \Lambda}_{i, \hsppp 1} - {\bf \Lambda}_{i, \hsppp 2} \right) } \cdot
\int_{x = 0}^{\infty} \left(1 + \frac{\rho}{2} {\bf \Lambda}_{i, \hsppp 1}
\hsppp x \right) \cdot
\log \left(1 + \frac{\rho}{2} {\bf \Lambda}_{i, \hsppp 1} \hsppp x \right)
\cdot e^{-x} dx  \nonumber \\
& & - \frac{1}{ \frac{\rho}{2} \left(
{\bf \Lambda}_{i, \hsppp 1} - {\bf \Lambda}_{i, \hsppp 2} \right) } \cdot
\int_{x = 0}^{\infty} \left(1 + \frac{\rho}{2} {\bf \Lambda}_{i, \hsppp 2}
\hsppp x \right) \cdot
\log \left(1 + \frac{\rho}{2} {\bf \Lambda}_{i, \hsppp 2} \hsppp x \right)
\cdot e^{-x} dx - 1.
\end{eqnarray}
Following a routine computation using the list of integral table
formula~\cite[$4.337(2)$, p.\ 572]{gradshteyn}, we have the expression for
$E \left[ I_{i, \hsppp 1} \right]$. Particularizing this expression to
the case of $E \left[ I_{i, \hsppp 2} \right]$ in~(\ref{eq_ii2}) with
$\widetilde{ {\bf \Lambda}} _{i, \hsppp 2} = 0$ results in the rate
expression as in the statement of the theorem. To
complete the proposition, an elementary computation of the eigenvalues
of the associated $2 \times 2$ matrices in~(\ref{eqna7})
and~(\ref{eqna9}) results in their characterization.
\end{proof}

\ignore{
\begin{figure}[htb!]
\begin{center}
\begin{tabular}{cc}
\includegraphics[height=2.5in,width=3in]{bcast_figs/bc/fig_hfunc.eps} &
\includegraphics[height=2.5in,width=3in]{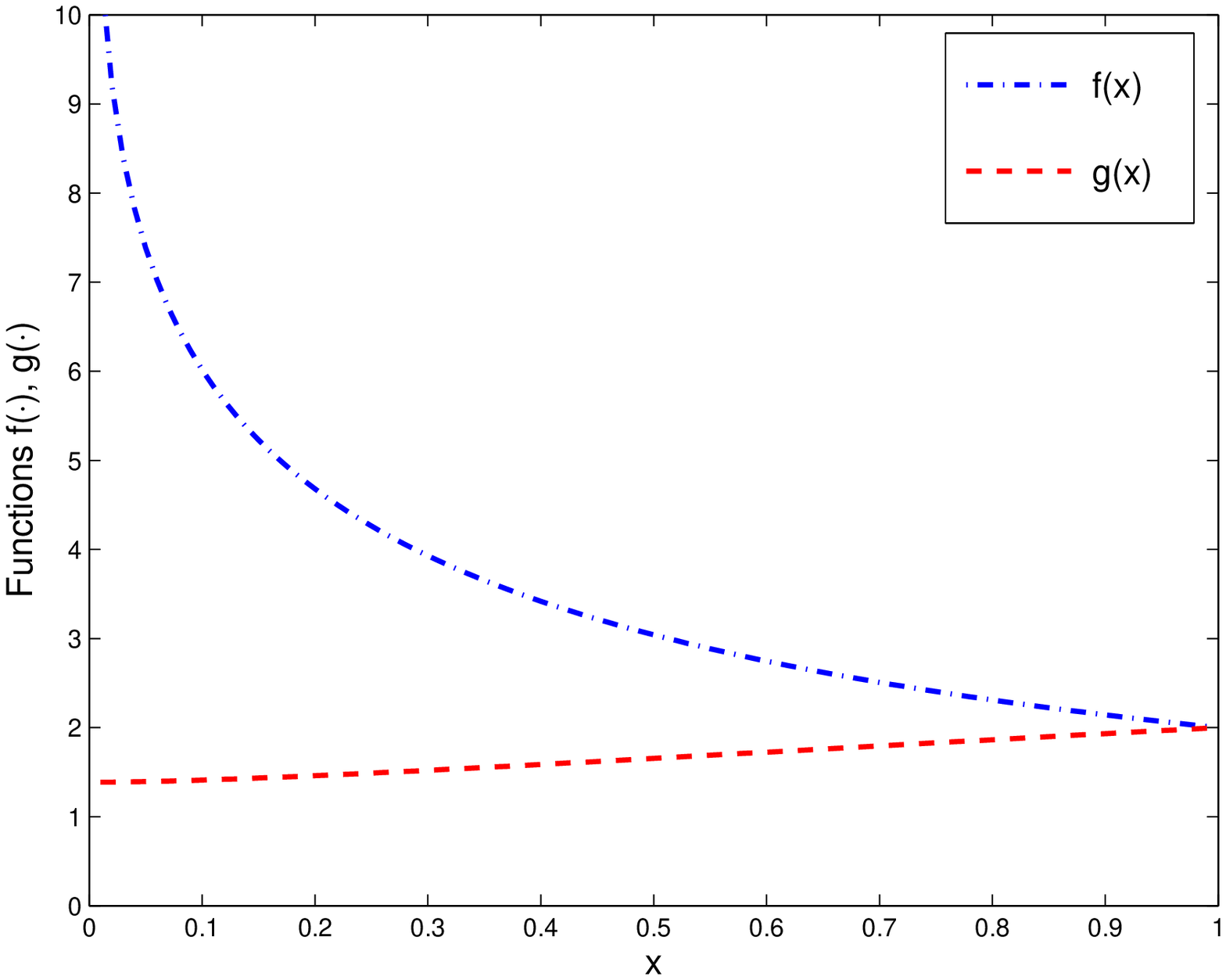}
\\ (a) & (b)
\end{tabular}
\caption{\label{fig2} (a) The behavior of $h(x)$, (b) The behavior of
$f(x)$ and $g(x)$.}
\end{center}
\end{figure}
}
The increasing nature of $h(\bullet)$, defined in~(\ref{hx}), is illustrated
in Fig.~\ref{fig2}. Towards the goal of obtaining physical intuition on
the structure of the optimal beamforming vectors, it is of interest to
obtain the limiting form of the ergodic rates in the low- and the
high-$\snr$ extremes.
\begin{figure}[htb!]
\centering
\begin{tabular}{c}
\includegraphics[height=3in,width=3.8in]{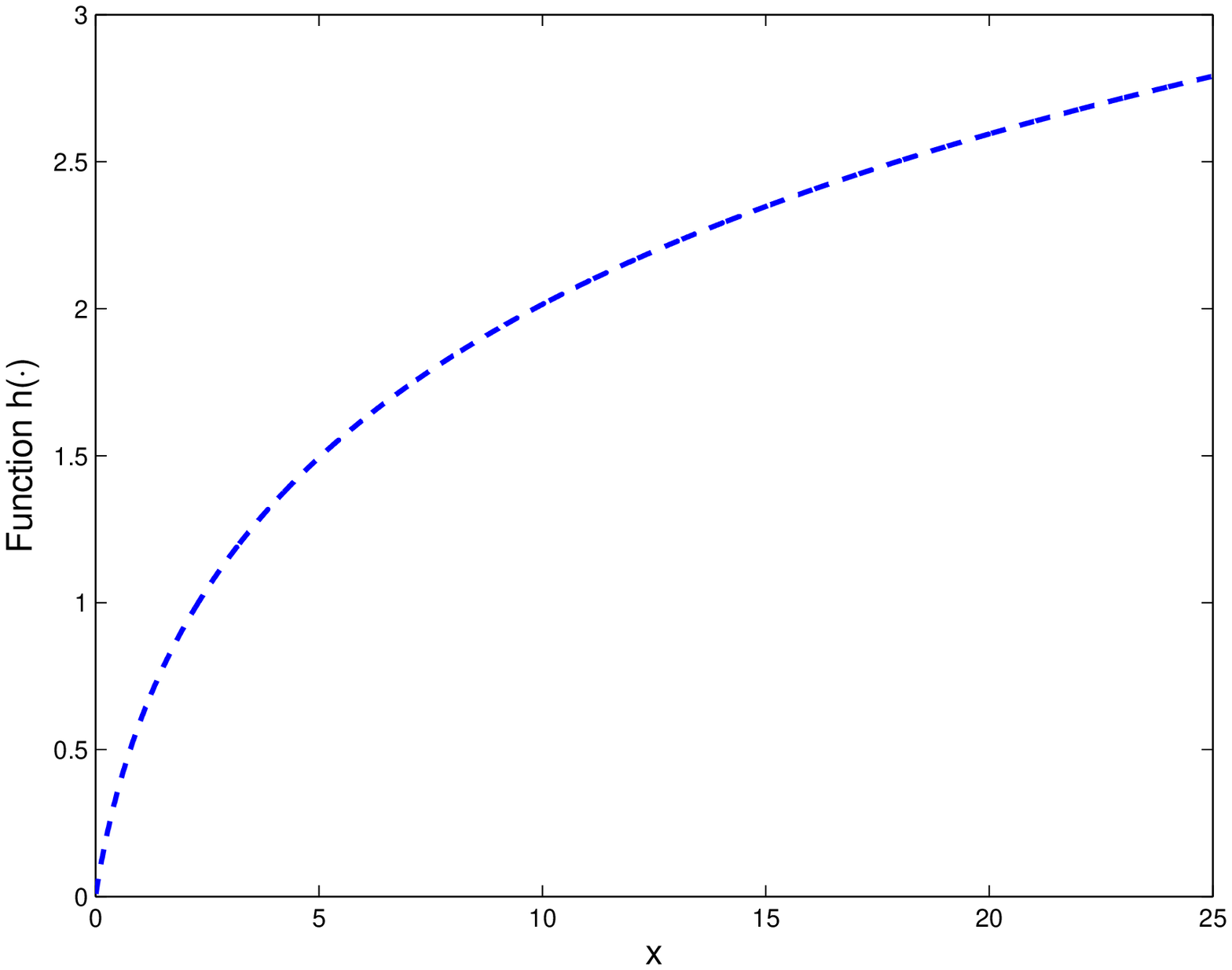}
\end{tabular}
\caption{The behavior of $h(x)$ for $x$ satisfying $0 < x \leq 25$.
\label{fig2}}
\end{figure}

\subsection{Low-$\snr$ Extreme}
\begin{prp}
The ergodic rate $E \left[ R_i \right]$ can be bounded as
\begin{eqnarray}
1 - \rho \hsppp {\sf C}_{\sf low} \leq
\frac{ E \left[R_i \right]}{ \frac{\rho}{2}
\left( {\bf \Lambda}_{i, \hsppp 1} + {\bf \Lambda}_{i, \hsppp 2} - B_i
\right) } \leq 1 + \rho \hsppp {\sf C}_{\sf up}
\label{lowsnr}
\end{eqnarray}
for some positive constants ${\sf C}_{\sf up}$ and ${\sf C}_{\sf low}$
(not provided here for the sake of brevity) that depend only on the
eigenvalues ${\bf \Lambda}_{i, \hsppp 1}, {\bf \Lambda}_{i, \hsppp 2}$
and $\widetilde{ {\bf \Lambda}} _{i, \hsppp 1}$. Thus, as $\rho \rightarrow 0$,
we have
\begin{eqnarray}
\frac{E \left[R_i \right]}{\rho} & \stackrel{\rho \rightarrow 0}{\rightarrow} &
\frac{1}{2} \left( {\bf \Lambda}_{i, \hsppp 1} +
{\bf \Lambda}_{i, \hsppp 2} - B_i
\right) \\
& = & \frac{A_i}{2} = \frac{ {\bf w}_i^H {\bf \Sigma}_i {\bf w}_i
}{2} . 
\label{lowsnr_simp}
\end{eqnarray}
\label{prop_lowsnr}
\end{prp}
{\vspace{-0.2in}}
\begin{proof}
We need the following bounds on the Exponential
integral~\cite[$5.1.20$, p.\ 229]{abramowitz}: 
\begin{eqnarray}
\frac{x}{1 + 2 x} \leq \frac{1}{2} \log \big( 1 + 2x \big)
\leq h(x) \leq \log \big( 1 + x \big) \leq x
{\vspace{0.05in}}
\end{eqnarray}
where the extremal inequalities are established by using the fact that
\begin{eqnarray}
\frac{x}{x+1} \leq \log(1 + x) \leq x.
\end{eqnarray}
Using these bounds, it is straightforward to see that the relationship
in~(\ref{lowsnr}) holds. Note that both the upper and lower bounds
converge to the same value as $\rho \rightarrow 0$, which results in the
simplification in~(\ref{lowsnr_simp}).
\end{proof}
In the low-$\snr$ extreme, the system is noise-limited, hence the
linear scaling of $E \left[R_i  \right]$ with $\rho$.

\subsection{High-$\snr$ Extreme}
\begin{prp}
As $\rho \rightarrow \infty$, we have
\begin{eqnarray}
E \left[ R_i \right] & \stackrel{ \rho \rightarrow \infty }
{ \rightarrow} &
\frac{  {\bf \Lambda}_{i,\hsppp 1}
\log \left(  {\bf \Lambda}_{i, \hsppp 1}  \right)
- {\bf \Lambda}_{i,\hsppp 2}
\log \left( {\bf \Lambda}_{i, \hsppp 2}  \right)  }
{ {\bf \Lambda}_{i, \hsppp 1} - {\bf \Lambda}_{i, \hsppp 2}  } -
\log \left( B_i \right)
\label{highsnr} \\
& = &
\frac{ A_i + B_i  }{
2 \hsppp \sqrt{ \left( A_i - B_i \right)^2 + 4C_i^2}  } \cdot
\log \left(  \frac{ A_i + B_i + \sqrt{ \left( A_i - B_i \right)^2 + 4C_i^2}  }
{A_i + B_i - \sqrt{ \left( A_i - B_i \right)^2 + 4C_i^2} } \right)
\nonumber \\
& & {\hspace{2.6in}}
+ \frac{1}{2} \hsppp
\log \left( \frac{ A_i B_i - C_i^2} {B_i^2}  \right)
\label{highsnr_simp}
\end{eqnarray}
with $A_i, B_i$ and $C_i$ as in~Theorem~\ref{prop_basic_rate}.
\label{prp_highsnr}
\end{prp}
\begin{proof}
The following asymptotic expansion~\cite[$5.1.11$, p.\ 229]{abramowitz}
of the Exponential integral is useful in obtaining the limiting form of
$E \left[R_i \right]$ as $\rho \rightarrow \infty$:
\begin{eqnarray}
E_1(x) & = & \log \left( \frac{1}{x} \right)
+ \sum_{k = 1}^{\infty} \frac{ (-1)^{k+1} x^k }{k \cdot k!} -\gamma
\\
& \stackrel{x \rightarrow 0}{ \rightarrow } & \log \left( \frac{1}{x} \right)
+ x - \gamma
\end{eqnarray}
where $\gamma \approx 0.577$ is the Euler-Mascheroni constant. Using the
limiting value of $E_1(x)$ to approximate $E \left[ R_i \right]$,
we have the expression in~(\ref{highsnr}). Expanding ${\bf \Lambda}_{i,\hsppp 1}$
and ${\bf \Lambda}_{i,\hsppp 2}$ in terms of $A_i, B_i$ and $C_i$, we have
the expression in~(\ref{highsnr_simp}).
\end{proof}
Unlike the low-${\sf SNR}$ extreme, $E \left[R_i \right]$ is not a
function of $\rho$ here. The dominating impact of interference (due to
the fixed nature of the beamforming vectors that are not adapted to the
channel realizations) and the consequent boundedness of $E \left[R_i
\right]$ in~(\ref{highsnr}) as $\rho$ increases should not be surprising.

\section{Sum-Rate Optimization: Two-User Case}
\label{sec4}
We are now interested in understanding the structure of the optimal
choice of beamforming vectors $\left( {\bf w}_{1, \hsppp {\sf opt}},
{\bf w}_{2, \hsppp {\sf opt} } \right)$ that maximize ${\cal R}_{\sf sum}$
as a function of ${\bf \Sigma}_1$, ${\bf \Sigma}_2$ and $\rho$. This
problem is difficult, in general. To obtain insight, we first consider
the low- and the high-$\snr$ extremes before studying the intermediate-$\snr$
regime.

For simplicity, let us assume an eigen-decomposition for ${\bf \Sigma}_1$
and ${\bf \Sigma}_2$ of the form
\begin{eqnarray}
\label{pdef}
{\bf \Sigma}_1 & = & {\bf U} \hspp {\sf diag}([\lambda_1({\bf \Sigma}_1),
\hsppp \lambda_2 ({\bf\Sigma}_1) ])
\hspp {\bf U}^H,\\ 
{\bf \Sigma}_2 & = & \widetilde{\bf U} \hspp
{\sf diag}([\lambda_1({\bf \Sigma}_2), \hsppp
\lambda_2({\bf \Sigma}_2)]) \hspp \widetilde{\bf U}^H,
\end{eqnarray}
where ${\bf U} = \left[ {\bf u}_1 ({\bf \Sigma}_1), \hsppp
{\bf u}_2({\bf \Sigma}_1)  \right]$, ${\widetilde{ {\bf U}}} =
\left[ {\bf u}_1({\bf \Sigma}_2), \hsppp {\bf u}_2({\bf \Sigma}_2)
\right]$, and $\lambda_1({\bf \Sigma}_i) \geq \lambda_2({\bf \Sigma}_i),
\hspp i = 1,2$. In particular, we assume that both ${\bf \Sigma}_1$ and
${\bf \Sigma}_2$ are positive-definite, that is, $\lambda_2( {\bf \Sigma}_i)
> 0$.

\subsection{Low-$\snr$ Extreme}
\begin{prp}
\label{prop_snr_low}
In the low-$\snr$ regime, from Prop.~\ref{prop_lowsnr} we see that
the maximization of $E \left[ R_i \right]$ involves optimizing over
${\bf w}_i$ alone. Thus, we have
\begin{eqnarray}
{\bf w}_{i, \hsppp {\sf opt}} = \argmax \limits_{ {\bf w}_i } {\cal R}_{\sf sum}
= \argmax \limits_{ {\bf w}_i } {\bf w}_i^H {\bf \Sigma}_i {\bf w}_i =
e^{j \nu_i} \hsppp
{\bf u}_1( {\bf \Sigma}_i ) 
\end{eqnarray}
for some choice of $\nu_i \in [0, 2\pi), \hsppp i = 1,2$. The resulting
ergodic sum-rate satisfies
\begin{eqnarray}
\lim_{ \rho \rightarrow 0}
\frac{ {\cal R}_{ \sf sum } } { \rho }
= \frac{1}{2} \cdot \Big[
\lambda_{1}( {\bf \Sigma}_1) + \lambda_{1} ( {\bf \Sigma}_2) \Big].
\end{eqnarray}
\endproof
\end{prp}
In the low-$\snr$ extreme, the optimal solution is such that the
transmitter signals
to a given user along the dominant statistical eigen-mode of that user's
channel and ignores the other user's channel completely. This is a solution
motivated by the single-user viewpoint where the optimality of signalling
along the dominant statistical eigen-mode of the forward channel is
well-known~\cite{foschini_blast,visotsky,goldsmith_review,bolsckei,venu_capacity}.
This solution is not surprising since in the noise-limited regime, the
broadcast channel is well-approximated by separate single-user models
connecting the transmitter to each receiver.

\subsection{High-$\snr$ Extreme}
\label{sec4b}
Define ${\bf \Sigma}$ (and its corresponding eigen-decomposition) as
\begin{eqnarray}
{\bf \Sigma} \triangleq {\bf \Sigma}_2^{-\frac{1}{2}} \hspp
{\bf \Sigma}_1 \hspp {\bf \Sigma}_2^{-\frac{1}{2}} =
{\bf V} \hspp {\sf diag}\left([ \eta_1 \hspp \eta_2  ]\right) \hspp
{\bf V}^H
\label{defn1}
\end{eqnarray}
where ${\bf V} = \left[ {\bf v}_1 \hspp {\bf v}_2 \right]$ and $\eta_1
\geq \eta_2$. Note that ${\bf \Sigma}$ is positive-definite ($\eta_2 > 0$)
since both ${\bf \Sigma}_1$ and ${\bf \Sigma}_2$ are positive-definite.
The main result of this section is as follows.
\begin{thm}
\label{thm2}
In the high-$\snr$ extreme, the ergodic sum-rate is maximized by the
following choice of beamforming vectors:
\begin{eqnarray}
{\bf w}_{1, \hsppp {\sf opt}} = e^{j \nu_1} \cdot \frac{
{\bf \Sigma}_2^{-\frac{1}{2} } \hsppp {\bf v}_1 } { \|
{\bf \Sigma}_2^{-\frac{1}{2} } \hsppp {\bf v}_1 \| },
\hspp \hspp {\bf w}_{2, \hsppp {\sf opt}} =
e^{j \nu_2} \cdot\frac{
{\bf \Sigma}_2^{-\frac{1}{2} } \hsppp {\bf v}_2 } { \|
{\bf \Sigma}_2^{-\frac{1}{2} } \hsppp {\bf v}_2 \| }
\label{bformer}
\end{eqnarray}
for some choice of $\nu_i \in [0, 2\pi), \hspp i = 1,2$. The optimal
ergodic sum-rate satisfies
\begin{eqnarray}
\lim_{ \rho \rightarrow \infty} {\cal R}_{\sf sum} =
\frac{\kappa_{1} \hsppp 
\log \left( \kappa_{1} \right)}{\kappa_{1}-1} +
\frac{\log \left(\kappa_{2} \right)}
{\kappa_{2}-1}
\label{rhs3}
\end{eqnarray}
where
\begin{eqnarray}
\kappa_{1}
\triangleq \frac{ \eta_1 \tau_2}{ \eta_2 \tau_1}, &  &
\kappa_{2} \triangleq
\frac{\tau_2}{\tau_1},
\label{k12gen}
\end{eqnarray}
\begin{eqnarray}
\tau_1 = {\bf v}_1^H \hsppp {\bf \Sigma}_2^{-1} \hsppp {\bf v}_1,
\hsp\hsp
\tau_2 = {\bf v}_2^H \hsppp {\bf \Sigma}_2^{-1} \hsppp {\bf v}_2,
\hsp \hsp
\tau_3 = {\bf v}_1^H \hsppp {\bf \Sigma}_2^{-1} \hsppp {\bf v}_2.
\label{defn2}
\end{eqnarray}
\end{thm}
{\vspace{0.2in}}
\begin{proof}
The first step in our proof is to rewrite the high-$\snr$ rate
expression in a form that permits further analysis. This is done in
Appendices~\ref{app_rewrite} and~\ref{app_prop2}.
With the definition of $\{ {\bf  v}_1, {\bf v}_2 \}$ as in~(\ref{defn1}),
since ${\bf \Sigma}_2$ is full-rank, we can decompose ${\bf w}_1$ and
${\bf w}_2$ as
\begin{eqnarray}
{\bf w}_1 & = & \frac{
\alpha \hsppp {\bf \Sigma}_2^{-\frac{1}{2} } \hsppp {\bf v}_1 +
\beta \hsppp {\bf \Sigma}_2^{-\frac{1}{2} } \hsppp {\bf v}_2 }
{ \| \alpha \hsppp {\bf \Sigma}_2^{-\frac{1}{2} } \hsppp {\bf v}_1 +
\beta \hsppp {\bf \Sigma}_2^{-\frac{1}{2} } \hsppp {\bf v}_2 \| }
\label{new_eqn1}
\\
{\bf w}_2 & = &
\frac{ \gamma \hsppp {\bf \Sigma}_2^{-\frac{1}{2} } \hsppp {\bf v}_1 +
\delta \hsppp {\bf \Sigma}_2^{-\frac{1}{2} } \hsppp  {\bf v}_2}
{ \| \gamma \hsppp {\bf \Sigma}_2^{-\frac{1}{2} } \hsppp {\bf v}_1 +
\delta\hsppp {\bf \Sigma}_2^{-\frac{1}{2} } \hsppp  {\bf v}_2 \| }
\label{new_eqn2}
\end{eqnarray}
for some choice of $\{ \alpha, \beta, \gamma, \delta \}$ with $\alpha =
|\alpha| e^{j \theta_{\alpha}}$ (similarly, for other quantities)
satisfying $| \alpha |^2 + |\beta |^2 = |\gamma|^2 + | \delta |^2 = 1$.
In Appendix~\ref{app_gen}, we show that the ergodic sum-rate optimization over
the six-dimensional parameter space $\{ |\alpha|, |\gamma|,
\theta_{\alpha}, \theta_{\beta}, \theta_{\gamma}, \theta_{\delta} \}$
results in the choice as in the statement of the theorem.
\end{proof}
Many remarks are in order at this stage.

\noindent {\bf \em Remarks:}
\begin{enumerate}
\item
Recall the definition of a generalized eigenvector:
\begin{defn} [{\bf \em Generalized eigenvector~\cite{stewart}}]
\label{def2}
A generalized eigenvector ${\bf x}$ (with the corresponding generalized
eigenvalue $\sigma$) of a pair of matrices $\left( {\bf A}, \hsppp
{\bf B} \right)$ satisfies the relationship
\begin{eqnarray}
{\bf A} \hsppp {\bf x} = \sigma \hsppp {\bf B} \hsppp {\bf x}.
\end{eqnarray}
In the special case where ${\bf B}$ is invertible, a generalized eigenvector
of the pair $\left( {\bf A}, \hsppp {\bf B} \right)$ is also an eigenvector
of ${\bf B}^{-1} {\bf A}$. If ${\bf A}$ and ${\bf B}$ are also
positive-definite, then all the generalized eigenvalues are positive.
While a unit-norm generalized eigenvector (or an eigenvector) is not unique
on ${\sf St}(M,1)$, it is unique on ${\cal G}(M,1)$.
\endproof
\end{defn}
We decompose $\{ {\bf w}_1, {\bf w}_2 \}$
in~(\ref{new_eqn1})-(\ref{new_eqn2}) along the basis\footnote{Note
that ${\bf \Sigma}_2$ is a full-rank matrix and hence, the vectors
${\bf \Sigma}_2^{-\frac{1}{2} } \hsppp {\bf v}_1$ and
${\bf \Sigma}_2^{-\frac{1}{2} } \hsppp {\bf v}_2$ form a non-orthogonal
basis (in general), whereas $\big\{ {\bf v}_1, {\bf v}_2 \big\}$ is
orthonormal.} $\big\{ {\bf \Sigma}_2^{-\frac{1}{2} } \hsppp {\bf v}_1,
{\bf \Sigma}_2^{-\frac{1}{2} } \hsppp {\bf v}_2 \big\}$ instead of the
more routine basis $\big\{ {\bf v}_1, {\bf v}_2 \big\}$. The reason for
this peculiar choice is as follows. It turns out that ${\bf \Sigma}_2^
{-\frac{1}{2} } \hsppp {\bf v}_1$ and ${\bf \Sigma}_2^ {-\frac{1}{2} }
\hsppp {\bf v}_2$ are the dominant generalized eigenvectors (corresponding
to the largest generalized eigenvalue) of the pairs
$\left( {\bf \Sigma}_1, \hsppp {\bf \Sigma}_2 \right)$ and
$\left( {\bf \Sigma}_2, \hsppp {\bf \Sigma}_1 \right)$, respectively.
For this claim, 
we use~(\ref{defn1}) to note that
\begin{eqnarray}
{\bf \Sigma}_2^{-1} \hsppp {\bf \Sigma}_1 & = &
{\bf \Sigma}_2^{ - \frac{1}{2} } \hsppp
\big( {\bf \Sigma}_2^{ - \frac{1}{2} } {\bf \Sigma}_1
{\bf \Sigma}_2^{ - \frac{1}{2} } \big) \hsppp
{\bf \Sigma}_2^{\frac{1}{2} }
= {\bf M} \hsppp {\bf D} \hsppp {\bf M}^{-1} \\
{\bf \Sigma}_1^{-1} \hsppp {\bf \Sigma}_2 & = &
\big({\bf \Sigma}_2^{-1} \hsppp {\bf \Sigma}_1  \big)^{-1} =
{\bf M} \hsppp {\bf D}^{-1} \hsppp {\bf M}^{-1}
\end{eqnarray}
where ${\bf M} = {\bf \Sigma}_2^{-\frac{1}{2} } \hsppp {\bf V}$ and
${\bf D} = {\sf diag}\big( [ \eta_1 \hspp \eta_2  ] \big)$. This means
that we can write
\begin{eqnarray}
{\bf w}_{1, \hsppp {\sf opt}} & = & e^{j \nu_1} \cdot {\bf u}_1
\big( {\bf \Sigma}_2^{-1} {\bf \Sigma}_1 \big) \\
{\bf w}_{2,\hsppp {\sf opt}} & = & e^{j \nu_2} \cdot {\bf u}_2
\big( {\bf \Sigma}_2^{-1} {\bf \Sigma}_1 \big)
\end{eqnarray}
where ${\bf u}_1(\bullet)$ and ${\bf u}_2(\bullet)$ are the dominant
and non-dominant eigenvectors, respectively.
%
Using the generalized eigenvector structure, it is easy to see that
\begin{eqnarray}
{\bf u}_1 \big( {\bf \Sigma}_2^{-1} {\bf \Sigma}_1 \big) & = &
{\bf u}_2 \big( {\bf \Sigma}_1^{-1} {\bf \Sigma}_2 \big) \\
{\bf u}_2 \big( {\bf \Sigma}_2^{-1} {\bf \Sigma}_1 \big) & = &
{\bf u}_1 \big( {\bf \Sigma}_1^{-1} {\bf \Sigma}_2 \big)
\end{eqnarray}
and thus
\begin{eqnarray}
{\bf w}_{1, \hsppp {\sf opt} } & = & e^{j \nu_1} \cdot {\bf u}_1
\big( {\bf \Sigma}_2^{-1} {\bf \Sigma}_1 \big) \\
{\bf w}_{2, \hsppp {\sf opt} } & = & e^{j \nu_2} \cdot {\bf u}_1
\big( {\bf \Sigma}_1^{-1} {\bf \Sigma}_2 \big).
\end{eqnarray}

\item
Given that the transmitter has only statistical information of the two links,
a natural candidate for beamforming in the high-${\sf SNR}$ extreme
is the solution to the maximization of an appropriately-defined average
${\sf SINR}$ metric for each user. Motivated by the fact
(see~(\ref{mission5})) that the instantaneous sum-rate for the $i$-th
user ($R_i$) is an increasing function of $| {\bf h}_i^H {\bf w}_i|^2$
whereas $R_j$ (for $j \neq i$) is a decreasing function of
$|{\bf h}_j^H {\bf w}_i|^2$, we define an ``average'' ${\sf SINR}$
metric as follows:
\begin{eqnarray}
{\sf SINR}_i \triangleq
\frac{ E \left[ |{\bf h}_i^H {\bf w}_i|^2 \right]}
{ E \left[ |{\bf h}_j^H{\bf w}_i |^2 \right]} =
\frac{ {\bf w}_i^H {\bf\Sigma}_i {\bf w}_i} { {\bf w}_i^H {\bf \Sigma}_j
{\bf w}_i}.
\end{eqnarray}
The optimization problem of interest is to maximize ${\sf SINR}_i$ which
has the generalized eigenvector structure as solution~\cite{cr_rao}:
\begin{eqnarray}
\argmax \limits_{ {\bf w}_i \hsppp : \hsppp
{\bf w} _i^H {\bf w}_i = 1 } {\sf SINR}_i =
e^{j \nu_i} \hsppp
{\bf u}_1 \big( {\bf \Sigma}_j^{-1} {\bf \Sigma}_i \big), \hspp
j \neq i, \hspp \{ i, j \} = 1,2.
\end{eqnarray}
It follows that if user $i$ selfishly maximizes (its own) ${\sf SINR}_i$
metric, then the set of such beamforming vectors maximize the ergodic
sum-rate in the high-${\sf SNR}$ regime. In this sense, the solution to
the broadcast problem mirrors and generalizes the single-user setting,
where the optimality of signalling along the statistical eigen-modes of
the channel is
well-understood~\cite{foschini_blast,visotsky,goldsmith_review,bolsckei,venu_capacity}.
Further, while optimal beamformer solutions in terms of the generalized
eigenvectors are obtained in the perfect CSI case of the broadcast setting
for the beamforming design problem~\cite{wiesel2,coord_bf} and the
interference channel problem~\cite{vannapu2}, to the best of our knowledge,
this solution in the statistical case is a first. A similar result is
obtained in a related work of ours~\cite{vasanth_isit10_intf} on statistical
beamforming vector design for the interference channel case. Since the
generalized eigenvector solution has an intuitive explanation, it is of
interest to obtain useful insights on the optimality of this solution in
more general multi-user settings.

\item
The ergodic sum-rate in~(\ref{rhs3}) is increasing in $\kappa_1$ and thus in
$\frac{\eta_1}{\eta_2}$. We now observe that ill-conditioning of
${\bf \Sigma}_1$ is necessary and sufficient to ensure that
$\frac{\eta_1}{\eta_2}$ is large. For this, we use standard eigenvalue
inequalities for product of Hermitian matrices~\cite{cr_rao} to see
that
\begin{eqnarray}
\frac{\chi_1}{\chi_2} =
\frac{ \lambda_1( {\bf \Sigma}_1 ) \cdot \lambda_2 ( {\bf \Sigma}_2^{-1} ) }
{ \lambda_2( {\bf \Sigma}_1 ) \cdot  \lambda_1( {\bf \Sigma}_2^{-1} ) }
\leq
\frac{\eta_1}{\eta_2} \leq \frac{ \lambda_1( {\bf \Sigma}_1 ) \cdot
\lambda_1( {\bf \Sigma}_2^{-1} ) }
{ \lambda_2( {\bf \Sigma}_1 ) \cdot  \lambda_2( {\bf \Sigma}_2^{-1} ) }
= \chi_1 \cdot \chi_2
\end{eqnarray}
where $\chi_i = \frac{ \lambda_1({\bf \Sigma}_i)}
{ \lambda_2( {\bf \Sigma}_i)}, \hsppp i = 1,2$. In other words, the more
ill-conditioned ${\bf \Sigma}_1$ is, the larger the high-$\snr$ statistical
beamforming sum-rate asymptote is (and {\em vice versa}).

On the other hand, the ergodic sum-rate in~(\ref{rhs3}) is not monotonic
in $\frac{\tau_1}{\tau_2}$. Nevertheless, it can be seen that as a
function of $\frac{\tau_1}{\tau_2}$, it has local maxima as
$\frac{\tau_1}{\tau_2} \rightarrow 0$ and $\frac{\tau_1}
{\tau_2} \rightarrow \infty$, and a minimum at $\frac{\tau_1}{\tau_2} = 1$.
The more well-conditioned ${\bf \Sigma}_2$ is, the more closer $\frac{\tau_1}
{\tau_2}$ is to $1$ and hence, the high-${\sf SNR}$ statistical beamforming
sum-rate asymptote is minimized. If ${\bf \Sigma}_2$ is ill-conditioned,
the value taken by $\frac{\tau_1}{\tau_2}$ depends on the angle between the
dominant eigenvectors of ${\bf \Sigma}_2$ and ${\bf \Sigma}$. If the two
eigenvectors are nearly parallel, $\frac{\tau_1}{\tau_2}$ is close to zero
and if they are nearly perpendicular, $\frac{\tau_1}{\tau_2}$ is very large.
In either case, the high-${\sf SNR}$ statistical beamforming sum-rate
asymptote is locally maximized.

The conclusion from the above analysis is that among all possible channels,
the ergodic sum-rate is maximized (or minimized) when ${\bf \Sigma}_1$ and
${\bf \Sigma}_2$ are both ill- (or well-)conditioned.
In other words, if both the users encounter poor scattering (that leads to
an ill-conditioning of their respective covariance matrices), their fading is
spatially localized. The transmitter can simultaneously excite these spatial
localizations without causing a proportional increase in the interference
level of the other user thus resulting in a higher ergodic sum-rate. On the
other hand, rich scattering implies that fading is spatially isotropic for
both the users. Any spatially localized excitation for one user will cause an
isotropic interference level at the other user thus resulting in a smaller
ergodic sum-rate.

\item
A special case that is of considerable interest is when ${\bf \Sigma}_1$
and ${\bf \Sigma}_2$ have the same set of orthonormal eigenvectors. This
would be a suitable model for certain indoor scenarios where the antenna
separation for the two users is the same~\cite{canonical_jayesh}.
Denoting (for simplicity) the set of common eigenvectors by ${\bf u}_1$
and ${\bf u}_2$, we can decompose ${\bf \Sigma}_1$ and ${\bf \Sigma}_2$ as
\begin{eqnarray}
{\bf \Sigma}_1 & = & \big[ {\bf u}_1, {\bf u}_2 \big] \hspp
{\sf diag}\big([\lambda_1, \hspp \lambda_2 ]\big) \hspp
\big[ {\bf u}_1, {\bf u}_2 \big]^H, \\
{\bf \Sigma}_2 & = & \big[ {\bf u}_1, {\bf u}_2 \big] \hspp
{\sf diag} \big([\mu_1, \hspp \mu_2 ]\big) \hspp
\big[ {\bf u}_1, {\bf u}_2 \big]^H.
\end{eqnarray}
We re-use the notations $\chi_1$ and $\chi_2$ to denote
\begin{eqnarray}
\chi_1 \triangleq \frac{\lambda_1}{\lambda_2}
&{\rm and}&
\chi_2 \triangleq \frac{\mu_1}{\mu_2}.
\end{eqnarray}
Without loss in generality, we can assume that $\chi_1 \geq 1$. Two
scenarios\footnote{These possibilities arise because
even though the set of eigenvectors of ${\bf \Sigma}_1$ and ${\bf \Sigma}_2$
are the same, there is no specific reason to expect the dominant
eigenvector of ${\bf \Sigma}_1$ to also be a dominant eigenvector of
${\bf\Sigma}_2$. Observe that the first case subsumes the setting where
$\mu_1 = \mu_2 = \mu$ and ${\bf \Sigma}_2 = \mu {\bf I}$.} arise depending
on the relationship between $\chi_1$ and $\chi_2$: i)
$\chi _1  \geq \chi_2$, and ii) $\chi _1 < \chi _2$.
\begin{thm}
\label{thm1}
In the high-$\snr$ extreme, the ergodic sum-rate is maximized by the
following choice of beamforming vectors:
\begin{eqnarray}
\begin{array}{cc}
{\bf w}_{1, \hsppp {\sf opt}} = e^{j \nu_1} \hsppp {\bf u}_1,
\hspp \hspp {\bf w}_{2, \hsppp {\sf opt}} =
e^{j \nu_2} \hsppp {\bf u}_2 & {\rm if}
\hspp \hspp
\chi_1  \geq \chi _2 
, \\
{\bf w}_{1, \hsppp {\sf opt}} = e^{j \nu_2} \hsppp {\bf u}_2,
\hspp \hspp {\bf w}_{2, \hsppp {\sf opt}} =
e^{j \nu_1} \hsppp {\bf u}_1 & {\rm if} \hspp \hspp 
\chi _1  < \chi _2  
\end{array}
\end{eqnarray}
for some choice of $\nu_i \in [0, 2\pi), \hspp i = 1,2$. The optimal
ergodic sum-rate satisfies
\begin{eqnarray}
\label{rhs4}
\lim_{\rho \rightarrow \infty} {\cal R}_{\sf sum} =
\left\{
\begin{array}{cc}
\frac{\chi_1 \hsppp \cdot \hsppp
\log \left( \chi_1  \right)}
{\chi_1 -1} +
\frac{ \log \left( \chi_2  \right)}
{\chi_2 -1} & {\rm if} \hspp
\chi_1  \geq \chi_2  \\
\frac{\chi_2  \hsppp \cdot \hsppp
\log \left( \chi_2  \right)}
{\chi_2 -1} +
\frac{\log \left( \chi_1  \right)}
{\chi_1 -1} & {\rm if} \hspp
\chi_1  < \chi_2.
\end{array}
\right.
\end{eqnarray}
\end{thm}
{\vspace{0.15in}}
\begin{proof}
While Theorem~\ref{thm2} can be particularized to this special case
easily, we pursue an alternate proof technique in Appendix~\ref{app_thm1}
that exploits the comparative relationship between $\tau_1$ and $\tau_2$ (which
is possible in the special case) and the fact that $\tau_3 = 0$.
\end{proof}
A comparison of the proof techniques of Theorems~\ref{thm2} and~\ref{thm1}
is presented in Appendix~\ref{app_comp}.

\item
Some remarks on the optimization set-up of this paper are necessary. The
proofs of Theorems~\ref{thm2} and~\ref{thm1} require us to consider
a six-dimensional optimization over the parameter space of $\{ |\alpha|,
|\gamma|, \theta_{\alpha}, \theta_{\beta}, \theta_{\gamma}, \theta_{\delta}
\}$. As a result, any geometric interpretation of the optimization is
impossible. A naive approach to the six-dimensional optimization
problems in this paper
is to adopt the method of
matrix differentiation calculus. For this approach to work, we need
to show that both the set over which optimization is done as well as
the optimized function are convex, neither of which is true in our case.
Specifically, neither ${\cal G}(2,1)$ nor ${\sf St}(2,1)$ are
convex sets. 
It also turns out that the ergodic sum-rate is neither convex nor
concave\footnote{It is possible that some function of the ergodic
sum-rate may be convex. But we are not aware of any likely candidate
that could work.} over the set of beamforming vectors, even over a
locally convex domain or an extended convex domain (like the interior of
the sphere).

\ignore{
Curiously, setting the first derivative of the ergodic sum-rate
expression in Theorem~\ref{thm2} to zero produces the generalized
eigenvector(s) of the pair of covariance matrices of the two users as
the solution. However, this approach is not rigorous enough to warrant
theoretical merit. In fact, it would not be possible to claim the local
optimality\footnote{If the first derivative condition alone holds, we
can only claim a local extremum, not even a local maximum.} of the
generalized eigenvector solution unless we can show that the associated
Hessian is also negative-definite. Computing the Hessian is cumbersome
and numerical computations show that the Hessian is actually not
negative-definite in our setting.
}

\item
The approach adopted in Appendices~\ref{app_gen} and~\ref{app_thm1}
overcomes these difficulties, and it consists of two steps. In the first
step, we produce an upper bound to the ergodic sum-rate that is independent
of the optimization parameters. In the second step, we show that this upper
bound can be realized by a specific choice of beamforming vectors thereby
confirming that choice's optimality. This approach seems to be the most
natural (and first principles-based) recourse to solving the non-convex
optimization problem at hand. An alternate approach to optimize the
ergodic sum-rate is non-linear optimization theory~\cite{floudas}. But
this approach is fraught with complicated Hessian calculations and technical
difficulties such as distinguishing between local and global extrema.
\end{enumerate}

\subsection{Intermediate-$\snr$ Regime: Candidate Beamforming Vectors}
While physical intuition on the structure of the optimal ergodic sum-rate
maximizing beamforming vectors has been obtained in the low- and the
high-$\snr$ extremes, the essentially intractable nature of the
Exponential integral in the ergodic rate expressions of
Theorem~\ref{prop_basic_rate} means that such a possibility at an
arbitrary $\snr$ is difficult. Nevertheless, the single-user
set-up~\cite{vasanth_isit07_heath,vasanth_semiunitary} suggests that the optimal
beamforming vectors (that determine the modes that are excited) and the power
allocation across these modes can be continuously parameterized by a function
of the $\snr$. Motivated by the single-user case, a desirable
quality for a ``good'' beamforming vector structure
($\left\{ {\bf w}_{i, \hsppp {\sf cand}}(\rho), \hsppp i = 1,2 \right\}$)
at an arbitrary $\snr$ of $\rho$ is that the limiting behavior of such a
structure in the low- and the high-$\snr$ extremes should be the solutions of
Prop.~\ref{prop_snr_low} and Theorem~\ref{thm2}. That is,
\begin{eqnarray}
\lim_{ \rho \rightarrow 0}
{\bf w}_{i, \hsppp {\sf cand}}(\rho) & = & e^{j \nu_i} \hsppp
{\bf u}_1( {\bf \Sigma}_i),
\label{lowsnr_limit} \\
\lim_{ \rho \rightarrow \infty}
{\bf w}_{i, \hsppp {\sf cand}}(\rho) & = & e^{j \nu_i} \hsppp
\frac{ {\bf \Sigma}_2^{ - \frac{1}{2} } {\bf v}_i }
{ \| {\bf \Sigma}_2^{ - \frac{1}{2} } {\bf v}_i \|}, \hsp i = 1,2,
\label{highsnr_limit}
\end{eqnarray}
where the above limits are seen as manifold operations~\cite{arias_smith}
on ${\cal G}(2,1)$.

A natural candidate that meets~(\ref{lowsnr_limit}) and~(\ref{highsnr_limit})
is the following choice parameterized by $\alpha(\rho)$ and $\beta(\rho)$
satisfying $\left\{ \alpha(\rho), \hsppp \beta(\rho) \right\} \in [0, \infty)$
and $\nu_i \in [0, 2\pi), \hsppp i = 1,2$:
\begin{eqnarray}
\label{proposed1}
{\bf w}_{1, \hsppp {\sf cand}}(\rho) & = &
e^{j \nu_1} \cdot
{\sf Dom. \hsppp eig.} \left(  \big( \alpha(\rho) {\bf \Sigma}_2 + {\bf I}
\big)^{-1} \hsppp {\bf \Sigma}_1 \right) \\
{\bf w}_{2, \hsppp {\sf cand}}(\rho) & = &
e^{j \nu_2} \cdot
{\sf Dom. \hsppp eig.} \left(  \big( \beta(\rho) {\bf \Sigma}_1 + {\bf I}
\big)^{-1} \hsppp {\bf \Sigma}_2 \right)
\label{proposed2}
\end{eqnarray}
where the notation ${\sf Dom. \hsppp eig}(\bullet)$ stands for the
unit-norm dominant eigenvector operation. These vectors can be seen
to be solutions to the following optimization problems:
\begin{eqnarray}
\argmax \limits_{ {\bf w}_i \hsppp : \hsppp
{\bf w} i^H {\bf w}_i = 1 } {\sf SINR}_i =
{\bf w}_{i, \hsppp {\sf cand}}(\rho), \hspp i = 1,2
\end{eqnarray}
where
\begin{eqnarray}
{\sf SINR}_1 & = & \frac{  {\bf w}_1^H {\bf \Sigma}_1 {\bf w}_1 }
{ {\bf w}_1^H {\bf w}_1 + \alpha(\rho) \hsppp {\bf w}_1^H {\bf \Sigma}_2
{\bf w}_1 } \\
{\sf SINR}_2 & = & \frac{  {\bf w}_2^H {\bf \Sigma}_2 {\bf w}_2 }
{ {\bf w}_2^H {\bf w}_2 + \beta(\rho) \hsppp {\bf w}_2^H {\bf \Sigma}_1
{\bf w}_2 }.
\end{eqnarray}
The choice in~(\ref{proposed1})-(\ref{proposed2}) is a low-dimensional
mapping from ${\cal G}(2,1) \times {\cal G}(2,1)$ to ${\mathbb R}^+
\times {\mathbb R}^+$ thus considerably simplifying the search space for
candidate beamforming vectors. It must be noted that while the search
space is simplified, the generalized eigenvector operation is a
non-linear mapping~\cite{stewart} in $\alpha(\rho)$ and $\beta(\rho)$.

\subsection{Numerical Studies}
We now study the ergodic sum-rate performance with ${\bf w}_1$
and ${\bf w}_2$ as in~(\ref{proposed1})-(\ref{proposed2}) via two
numerical examples.
In the first study, we consider a system (note that ${\sf Tr}(
{\bf \Sigma}_1) = {\sf Tr}( {\bf \Sigma}_2) = M = 2$) with
\begin{eqnarray}
{\bf \Sigma}_1 & = & \left[
\begin{array}{cc}
1.7745   &  -0.5178 + 0.0247i \\
-0.5178 - 0.0247i &   0.2255
\end{array} \right], 
\label{sig1_better}
\\
{\bf \Sigma}_2 & = & \left[
\begin{array}{cc}
1.2522  & -0.8739 - 0.2711i \\
-0.8739 + 0.2711i &  0.7478
\end{array} \right].
\label{sig2_better}
\end{eqnarray}
Fig.~\ref{fig3}(a) shows the ergodic sum-rate as a function of $\rho$
for four schemes. For the first scheme, for every $\rho$, an optimal
choice $\{ \alpha^{\star}(\rho), \beta^{\star}(\rho) \}$ is obtained
from the search space $\alpha(\rho) \times \beta(\rho) \in [0, \infty) \times
[0, \infty)$ as follows:
\begin{eqnarray}
\left\{ \alpha^{\star}(\rho), \hspp \beta^{\star}(\rho) \right\}
& = & \argmax \limits _{ \{ \alpha(\rho), \hsppp \beta(\rho) \}  }
E \left[ R_1 \right] + E \left[ R_2 \right] \nonumber \\
& & {\hspace{0.3in}}
{\rm s.t.} \hspp \hspp
{\bf w}_1 \hsppp = \hsppp {\bf w}_{1, \hsppp {\sf cand}}(\rho), \hspp
{\bf w}_2  \hsppp = \hsppp {\bf w}_{2, \hsppp {\sf cand}}(\rho).
\end{eqnarray}
The performance of the beamforming vectors with $\alpha^{\star}(\rho)$
and $\beta^{\star}(\rho)$ for every $\rho$ is plotted along with the
performance of the candidate obtained via a numerical (Monte Carlo)
search over ${\cal G}(2,1) \times {\cal G}(2,1)$. As
motivated in the prior discussion, while we expect the performance with
$\{ \alpha^{\star}(\rho), \beta^{\star}(\rho) \}$ to be good, it is
surprising that this choice is indeed optimal. Further, the performance
of a set of beamforming vectors with $\alpha(\rho) = \beta(\rho) = 0$ and
$\alpha(\rho) = 100, \beta(\rho) = 15$ (fixed for all $\rho$
in~(\ref{proposed1})-(\ref{proposed2})) are also plotted. Observe that
these two sets approximate the low- and the high-$\snr$ solutions of
Prop.~\ref{prop_snr_low} and Theorem~\ref{thm2}, respectively.
\begin{figure}[htb!]
\begin{center}
\begin{tabular}{cc}
\includegraphics[height=2.5in,width=3in]{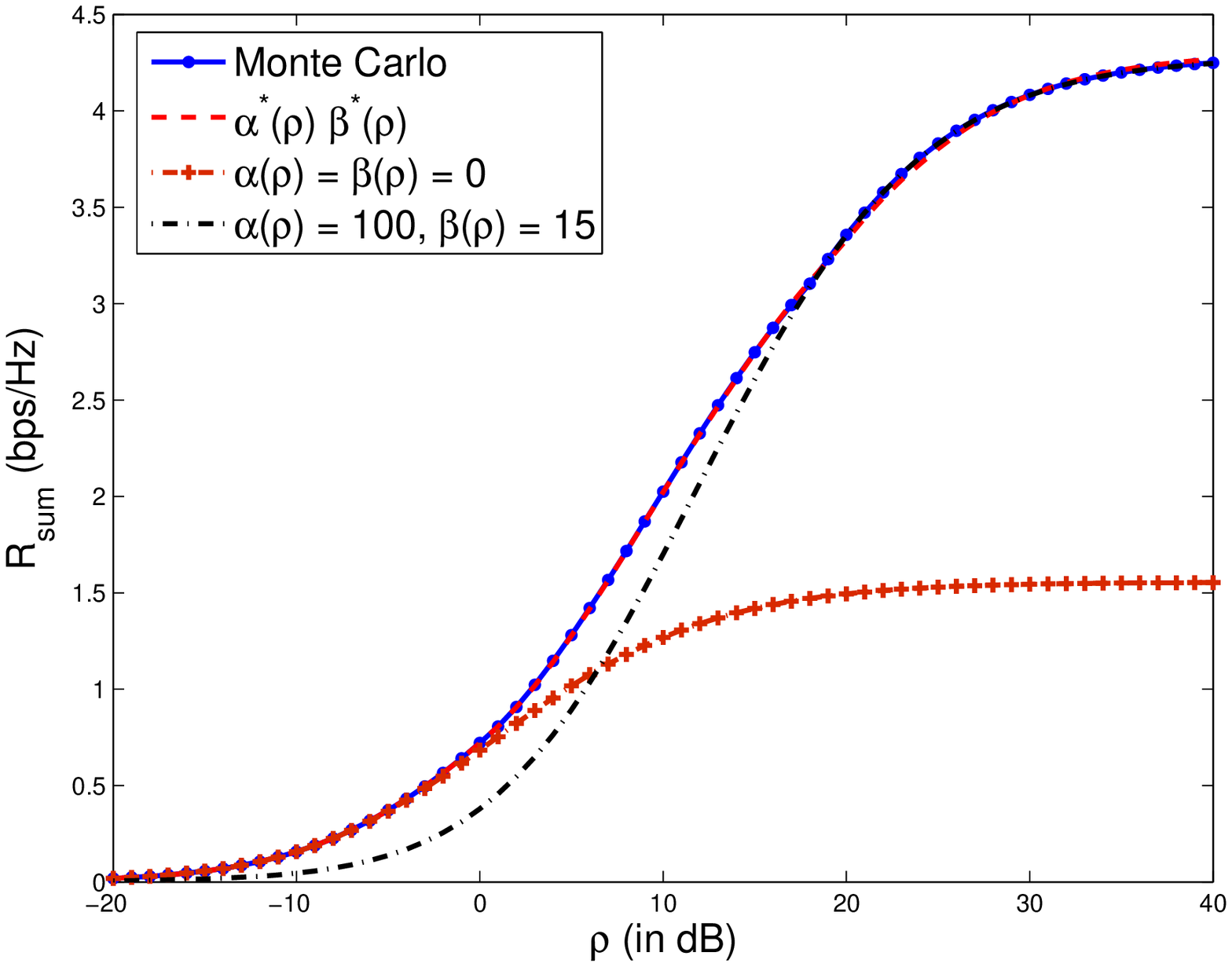} &
\includegraphics[height=2.5in,width=3in]{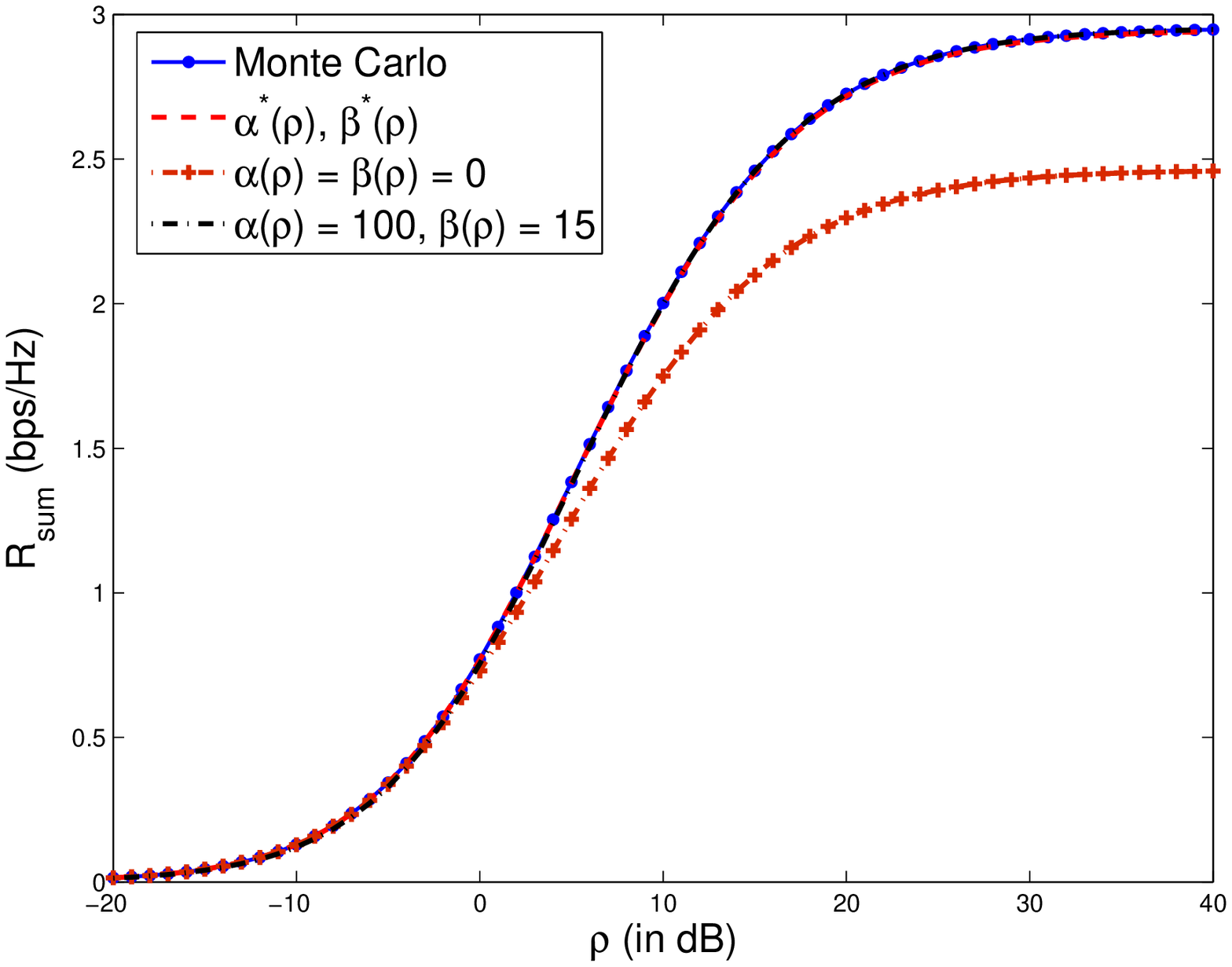}
\\ (a) & (b)
\end{tabular}
\caption{\label{fig3} Performance of proposed scheme with ${\bf \Sigma}_1$
and ${\bf\Sigma}_2$: (a) as in~(\ref{sig1_better})-(\ref{sig2_better}),
(b) as in~(\ref{sig1_worse})-(\ref{sig2_worse}).}
\end{center}
\end{figure}

In the second study, we consider a system (again, note that
${\sf Tr}( {\bf \Sigma}_1) = {\sf Tr}( {\bf \Sigma}_2) = M = 2$) with
\begin{eqnarray}
{\bf \Sigma}_1 & = & \left[
\begin{array}{cc}
1.3042 &   0.0543 - 0.2540i \\
0.0543 + 0.2540i  & 0.6958
\end{array} \right], 
\label{sig1_worse} \\
{\bf \Sigma}_2 & = & \left[
\begin{array}{cc}
1.1161 & -0.2195 + 0.4340i \\
-0.2195 - 0.4340i  & 0.8839
\end{array} \right].
\label{sig2_worse}
\end{eqnarray}
Fig.~\ref{fig3}(b) plots the performance of the proposed scheme, the low-
and the high-$\snr$ solutions in addition to the candidate obtained via
a numerical search over ${\cal G}(2,1) \times {\cal G}(2,1)$. As before,
the low- and the
high-$\snr$ solutions are optimal in their respective extremes while the
candidate $\{ \alpha^{\star}(\rho), \beta^{\star}(\rho) \}$ is optimal
across all $\rho$.

Note that in Fig.~\ref{fig3}(a) there exists an $\snr$-regime where both
the low- and the high-$\snr$ solutions are sub-optimal. In
contrast, in Fig.~\ref{fig3}(b), the high-$\snr$ solution essentially
coincides with the numerical search for all $\rho$ whereas at the
low-$\snr$ extreme, the performance of the low-$\snr$ solution is as
expected. We now explain why the high-$\snr$ solution performs as well as
$\{ \alpha^{\star}(\rho), \beta^{\star}(\rho) \}$ for all $\rho$. For this,
we need to understand the behavior of the angle between the proposed set of
beamforming vectors in~(\ref{proposed1})-(\ref{proposed2}) and the
low-$\snr$ solution as a function of $\alpha$ and $\beta$. In
Fig.~\ref{fig4}(b), we plot $\cos \left( {\sf Angle}_1(\alpha(\rho)) \right)$
as a function of $\alpha(\rho)$ and $\cos \left( {\sf Angle}_2(\beta(\rho))
\right)$ as a function of $\beta(\rho)$ where
\begin{eqnarray}
\label{angla1}
\cos \left( {\sf Angle}_1(\alpha(\rho)) \right) & = & \left|
\left( {\sf Dom. \hsppp eig.} \left(  \big( \alpha(\rho) {\bf \Sigma}_2 + {\bf I}
\big)^{-1} \hsppp {\bf \Sigma}_1 \right)  \right)^H
{\sf Dom. \hsppp eig.} \left( {\bf \Sigma}_1 \right) \right| \\
\cos \left( {\sf Angle}_2(\beta(\rho)) \right) & = & \left|
\left( {\sf Dom. \hsppp eig.} \left(  \big( \beta(\rho) {\bf \Sigma}_1 + {\bf I}
\big)^{-1} \hsppp {\bf \Sigma}_2 \right)  \right)^H
{\sf Dom. \hsppp eig.} \left( {\bf \Sigma}_2 \right) \right|.
\label{angla2}
\end{eqnarray}
From Fig.~\ref{fig4}(b), we note that the chordal distance\footnote{In
short, the chordal distance is the square-root of the difference of $1$
and the square of the quantity computed in~(\ref{angla1})
(or~(\ref{angla2})). See~(\ref{chordal}) for more details.} between the
low- and the
high-$\snr$ solutions is small (on the order of $0.05$). Also, observe
that there is a quick convergence of~(\ref{proposed1})-(\ref{proposed2})
as $\alpha$ (or $\beta$) increases to the high-$\snr$ solution and hence
the high-$\snr$ solution is a good approximation to the choice
$\{ \alpha^{\star}(\rho), \beta^{\star}(\rho) \}$ over a large $\snr$
range. On the other hand, from Fig.~\ref{fig4}(a), we see that the chordal
distance between the low- and the high-$\snr$ solutions is large (on the
order of $0.90$), which translates to the sub-optimality gap in
Fig.~\ref{fig3}(a).

\begin{figure}[htb!]
\begin{center}
\begin{tabular}{cc}
\includegraphics[height=2.5in,width=3in]{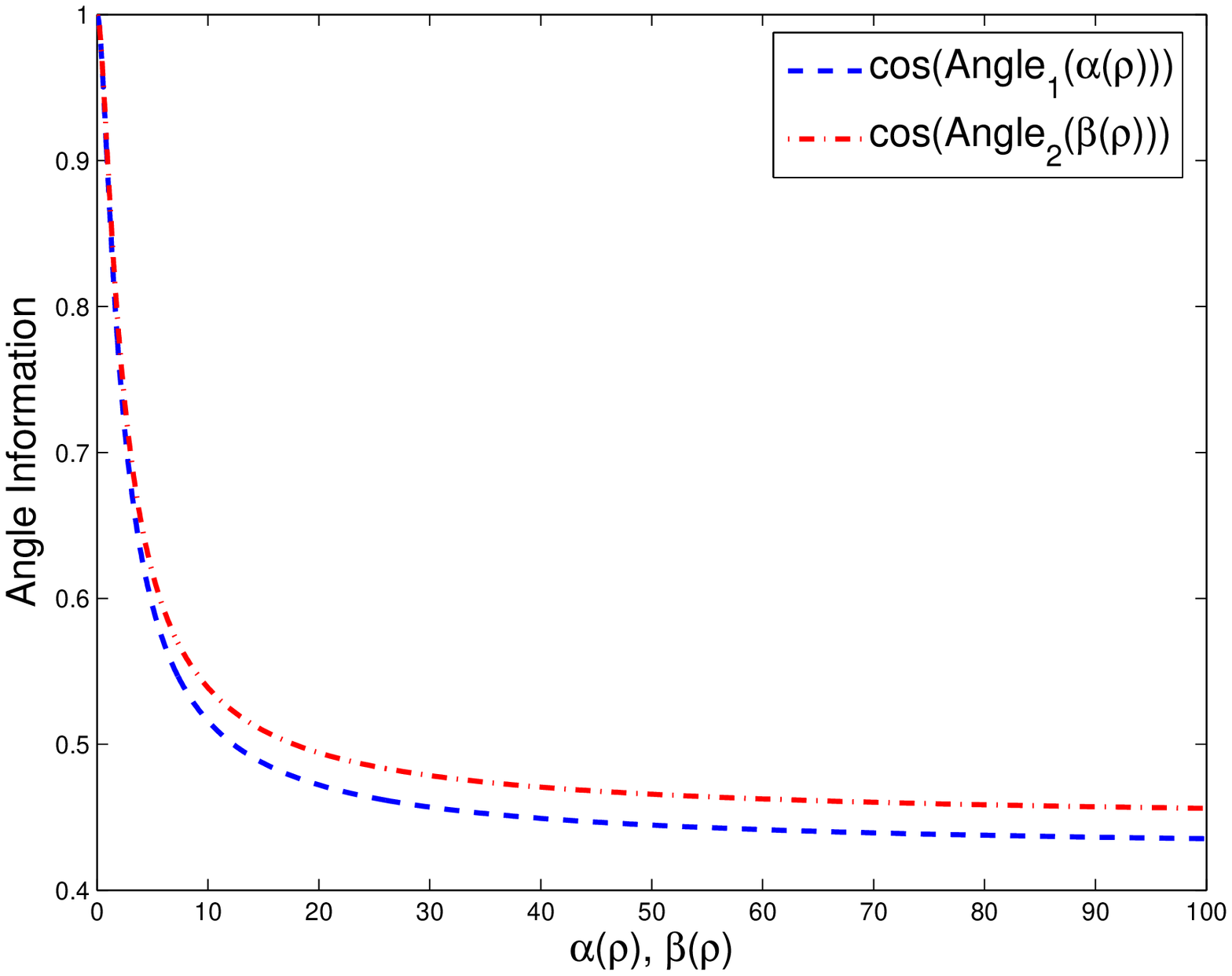} &
\includegraphics[height=2.5in,width=3in]{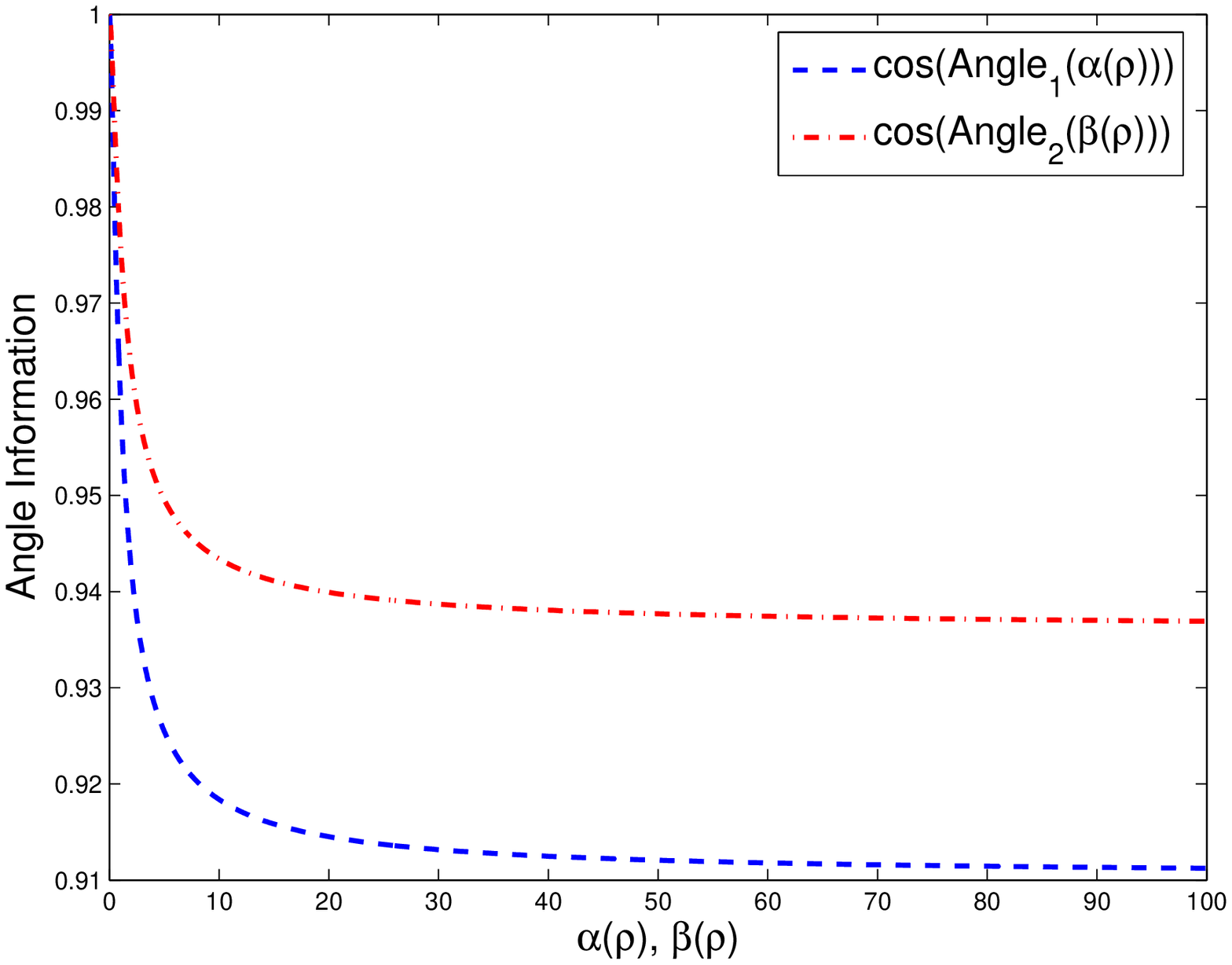}
\\ (a) & (b)
\end{tabular}
\caption{\label{fig4} Angle$_1$ and Angle$_2$, defined
in~(\ref{angla1})-(\ref{angla2}), as a function of $\alpha$ and $\beta$
for the setting in: (a)~(\ref{sig1_better})-(\ref{sig2_better}),
(b)~(\ref{sig1_worse})-(\ref{sig2_worse}).}
\end{center}
\end{figure}

While we are unable to prove the optimality structure of the proposed
scheme in the intermediate-$\snr$ regime, motivated by our numerical
studies, we make the following conjecture.
\begin{conj}
\label{thm_general}
In the intermediate-$\snr$ regime, the ergodic sum-rate is maximized
by the following choice of beamforming vectors:
\begin{eqnarray}
{\bf w}_{1, \hsppp {\sf opt} } & = & e^{j \nu_1} \cdot
{\sf Dom. \hsppp eig.} \left(  \big( \alpha^{\star}(\rho)
{\bf \Sigma}_2 + {\bf I} \big)^{-1} \hsppp {\bf \Sigma}_1 \right)
\label{cand1}
\\
{\bf w}_{2, \hsppp {\sf opt} } & = & e^{j \nu_2} \cdot
{\sf Dom. \hsppp eig.} \left(  \big( \beta^{\star}(\rho)
{\bf \Sigma}_1 + {\bf I} \big)^{-1} \hsppp {\bf \Sigma}_2 \right)
\label{cand2}
\end{eqnarray}
for some choice of $\nu_i \in [0, 2 \pi), \hsppp i = 1,2$.
The notation ${\sf Dom. \hsppp eig}(\bullet)$ stands for the
unit-norm dominant eigenvector operation and
\begin{eqnarray}
\left\{ \alpha^{\star}(\rho), \hspp \beta^{\star}(\rho) \right\}
& = & \argmax \limits_{ \{ \alpha(\rho), \hsppp \beta(\rho) \}  }
E \left[ R_1 \right] + E \left[ R_2 \right] \nonumber \\
& & {\hspace{0.3in}}
{\rm s.t.} \hspp \hspp
{\bf w}_1 \hsppp = \hsppp {\bf w}_{1, \hsppp {\sf cand}}(\rho), \hspp
{\bf w}_2  \hsppp = \hsppp {\bf w}_{2, \hsppp {\sf cand}}(\rho).
\label{alphabetastar}
\end{eqnarray}
\endproof
\end{conj}

\ignore{
Once a closed-form expression is
obtained, in Sec.~\ref{sec4} we will characterize the structure of
the optimal beamforming vectors as a function of the channel statistics
and ${\sf SNR}$ ($\rho$).

Define the condition numbers $\kappa_1$ and $\kappa_2$ as
\begin{eqnarray}
\kappa_1 \triangleq \frac{\lambda_1}{\lambda_2} \hspp \hspp {\rm and }
\hspp \hspp \kappa_2 \triangleq \frac{\mu_1}{\mu_2}.
\end{eqnarray}
}

\section{Ergodic Sum-Rate: Generalizations}
\label{sec5}
We studied the structure of ergodic sum-rate maximizing beamforming vectors in
Sec.~\ref{sec4}. In this section, we consider more general problems of
this nature.
\subsection{Maximizing $E \left[R_i \right]$}
Consider a system where the Quality-of-Service metric of one user
significantly dominates that of the other user. For example, one user
is considerably more important to the network operator than the other.
The relevant metric to optimize in this scenario is not the ergodic
sum-rate, but the rate achievable by the more important user. In this
setting, we have the following result.
\begin{prp}
\label{prop1}
The optimal choice of the pair $\left( {\bf w}_{1, \hsppp {\sf opt}},
{\bf w}_{2, \hsppp {\sf opt}} \right)$ that maximizes $E \left[  R_i \right]$
is:
\newline
\noindent {\bf \em i) Low-$\snr$ Extreme:}
\begin{eqnarray}
{\bf w}_{i, \hsppp {\sf opt}} = e^{j \nu_1} \hsppp {\bf u}_1
\left( {\bf \Sigma}_i  \right) \hsp {\rm and} \hsp
{\bf w}_{j, \hsppp {\sf opt}} = {\sf any} \hspp {\sf vector \hspp on} \hspp
{\cal G}(2,1), \hspp j \neq i 
\label{opt_prop34}
\end{eqnarray}
{\vspace{-0.2in}}
\newline
\noindent {\bf \em ii) High-$\snr$ Extreme:}
\begin{eqnarray}
{\bf w}_{i, \hsppp {\sf opt}} = e^{j \nu_1} \hsppp {\bf u}_1
\left( {\bf \Sigma}_i  \right)
\hsp {\rm and} \hsp {\bf w}_{j, \hsppp {\sf opt}} = e^{j \nu_2}
\hsppp {\bf u}_2 \left( {\bf \Sigma}_i  \right), \hspp j \neq i
\label{opt_prop3}
\end{eqnarray}
for some choice of $\nu_i \in [0, 2\pi), \hsppp i = 1,2$.
\end{prp}
\begin{proof}
See Appendix~\ref{app_prop1}.
\end{proof}
With the above choice of beamforming vectors, $E \left[ R_i \right]$ can
be written as
\begin{eqnarray}
E \left[ R_i \right]
& \stackrel{\rho \rightarrow \infty}{ \rightarrow } &
\frac{\chi_i \log(\chi_i)}{\chi_i - 1}
\label{eqn5}
\end{eqnarray}
where $\chi_i = \frac{ \lambda_1( {\bf \Sigma}_i )}
{ \lambda_2( {\bf \Sigma}_i )}$.
From~(\ref{eqn5}), it is to be noted that $E\left[ R_i \right]$
increases 
as $\chi_i$ increases. That is, the more ill-conditioned ${\bf \Sigma}_i$
is, the larger the high-$\snr$ statistical beamforming rate asymptote is
(and {\em vice versa}). This should be intuitive as our goal is only to
maximize $E \left[ R_i \right]$ and the beamforming vectors
in~(\ref{opt_prop3}) achieve that goal.

\subsection{Weighted Ergodic Sum-Rate Maximization}
In a system where the Quality-of-Service metrics of the two users are
comparable (but not the same), the relevant metric to optimize is the
weighted-sum of ergodic rates achievable by the two users~\cite{kobayashi}.
Specifically, the objective function here is
\begin{eqnarray}
{\cal R}_{\sf weighted} = \zeta_1 E \left[R_1 \right] +
\zeta_2 E \left[ R_2 \right]
\end{eqnarray}
for some choice of weights $\zeta_1$ and $\zeta_2$ satisfying
(without loss in generality) $\left\{
\zeta_1, \hsppp \zeta_2 \right\} \in [0,1].$ Note that $E \left[ R_i \right]$
is a special case of this objective function with $\zeta_1 = 1, \zeta_2 = 0$
or $\zeta_1 = 0, \zeta_2 = 1$.

Maximizing ${\cal R}_{\sf weighted}$ to obtain a closed-form
characterization 
of the optimal beamforming
vectors seems hard in general. Motivated by the study for the sum-rate in
the intermediate-$\snr$ regime in Sec.~\ref{sec4}, we now consider a set
of candidate beamforming vectors that produce known optimal structures in
special cases. For this, it is important to note that no choice of
$\alpha(\rho)$ and $\beta(\rho)$ in~(\ref{proposed1})-(\ref{proposed2})
can produce the beamforming vectors in~(\ref{opt_prop3}). A candidate set
of beamforming vectors that not only produces the special (extreme) cases
in the ergodic sum-rate setting, but also~(\ref{opt_prop3}) is the
following choice parameterized by four quantities, $\{ \alpha(\rho),
\beta(\rho), \gamma(\rho), \delta(\rho) \} \in [0, \infty)$:
\begin{eqnarray}
\label{proposed3}
{\bf w}_{1, \hsppp {\sf weighted, \hsppp cand}}(\rho) & = &
e^{j \nu_1} \cdot
{\sf Dom. \hsppp eig.} \left(  \big( \alpha(\rho) {\bf \Sigma}_2 + {\bf I}
\big)^{-1} \hsppp  \big( \gamma(\rho) {\bf \Sigma}_1  + {\bf I} \big)
\right) \\
{\bf w}_{2, \hsppp {\sf weighted, \hsppp cand}}(\rho) & = &
e^{j \nu_2} \cdot
{\sf Dom. \hsppp eig.} \left(  \big( \beta(\rho) {\bf \Sigma}_1 + {\bf I}
\big)^{-1} \hsppp \big( \delta(\rho) {\bf \Sigma}_2 + {\bf I} \big)
\right)
\label{proposed4}
\end{eqnarray}
where the notation ${\sf Dom. \hsppp eig}(\bullet)$ stands for the usual
unit-norm dominant eigenvector operation. As
before,~(\ref{proposed3})-(\ref{proposed4}) corresponds to a low-dimensional
map from $\left\{ {\cal G}(2,1) \right\}^4$ to $\left\{ [0, \infty)
\right\}^4$ and thus a simplification in the search for a good beamformer
structure.

\begin{figure}[htb!]
\centering
\begin{tabular}{c}
\includegraphics[height=3in,width=3.8in]{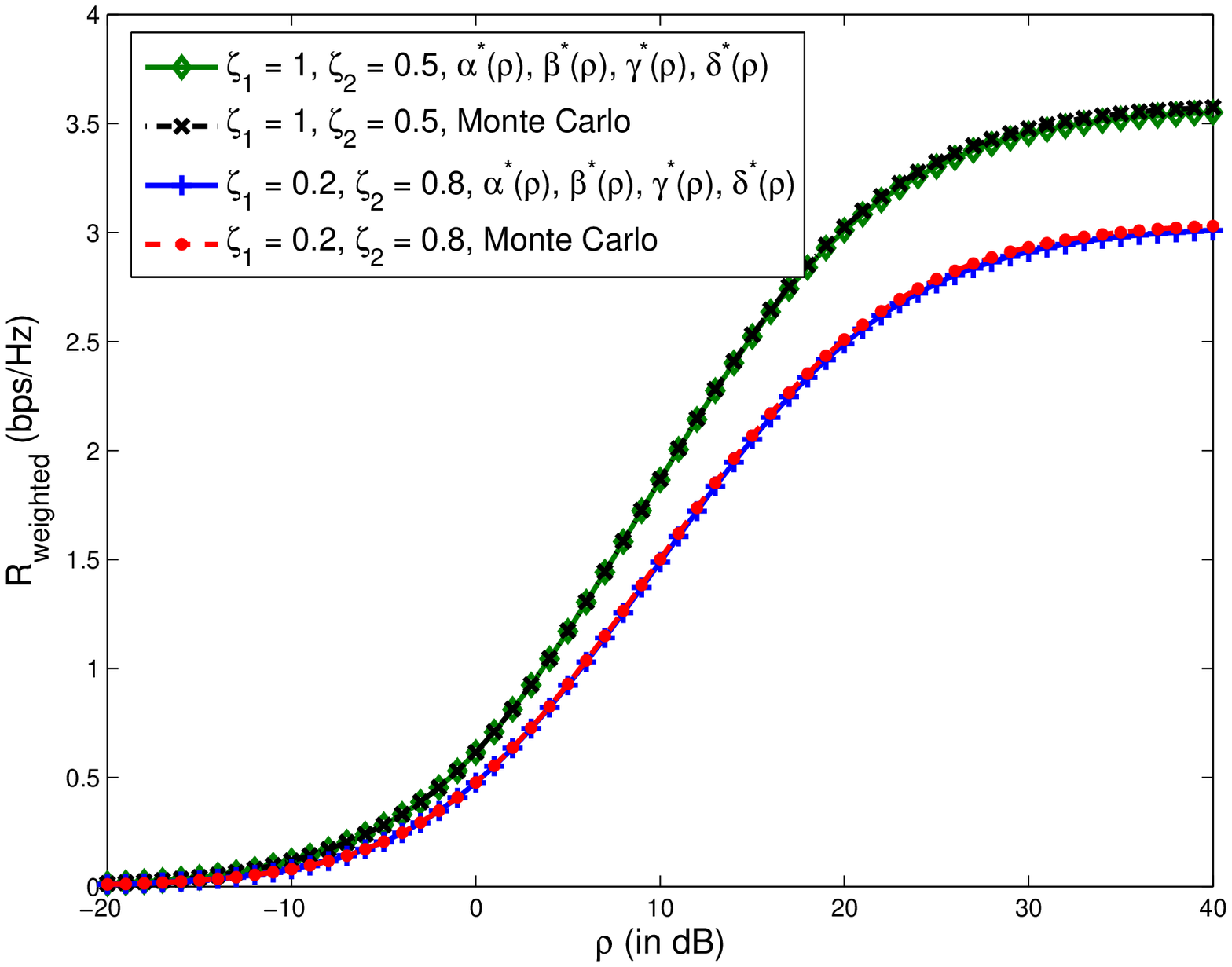}
\end{tabular}
\caption{Weighted ergodic sum-rate of proposed scheme with ${\bf \Sigma}_1$
and ${\bf \Sigma}_2$ as in~(\ref{sig1_better})-(\ref{sig2_better}).
\label{fig5}}
\end{figure}
We now study the performance of the proposed beamforming vectors
in~(\ref{proposed3})-(\ref{proposed4}) for the system with ${\bf \Sigma}_1$
and ${\bf \Sigma}_2$ as in~(\ref{sig1_better})-(\ref{sig2_better}).
Two sets of weights are considered: i) $\zeta_1 = 1, \zeta_2 = 0.5$
and ii) $\zeta_1 = 0.2, \zeta_2 = 0.8$. Fig.~\ref{fig5} plots the
performance of two schemes. The first scheme corresponds to a Monte Carlo
search over ${\cal G}(2,1)$, whereas the second scheme corresponds to the
use of an optimal choice $\{ \alpha^{\star}(\rho), \beta^{\star}(\rho),
\gamma^{\star}(\rho), \delta^{\star}(\rho) \}$ (for every $\rho$) with
\begin{align}
& \left\{ \alpha^{\star}(\rho), \hspp \beta^{\star}(\rho), \hspp
\gamma^{\star}(\rho), \hspp \delta^{\star}(\rho) \right\}
\nonumber \\
& {\hspace{0.5in}} =
\argmax \limits _{ \{ \alpha(\rho), \hsppp \beta(\rho),
\hsppp \gamma(\rho), \hsppp \delta(\rho) \}  }
E \left[ R_1 \right] + E \left[ R_2 \right] \nonumber \\
& {\hspace{0.85in}}
{\rm s.t.} \hspp \hspp
{\bf w}_1 \hsppp = \hsppp
{\bf w}_{1, \hsppp {\sf weighted, \hsppp cand}}(\rho), \hspp
{\bf w}_2  \hsppp = \hsppp
{\bf w}_{2, \hsppp {\sf weighted, \hsppp cand}}(\rho).
\label{eq3}
\end{align}
As can be seen from Fig.~\ref{fig5}, the proposed scheme
in~(\ref{proposed3})-(\ref{proposed4}) performs as well as the Monte
Carlo search for both sets of weights. Numerical studies suggest that
similar performance is seen across all possible ${\bf \Sigma}_1$
and ${\bf\Sigma}_2$, and all possible weights $\zeta_1$ and $\zeta_2$.
Motivated by these studies, we pose the following conjecture.
\begin{conj}
\label{conj2}
In the intermediate-$\snr$ regime, the weighted ergodic sum-rate,
${\cal R}_{\sf weighted}$, is maximized by the following choice of
beamforming vectors:
\begin{eqnarray}
{\bf w}_{1} & = & e^{j \nu_1} \cdot
{\sf Dom. \hsppp eig.} \left(  \big( \alpha^{\star}(\rho) {\bf \Sigma}_2 +
{\bf I} \big)^{-1} \hsppp  \big( \gamma^{\star}(\rho) {\bf \Sigma}_1  +
{\bf I} \big) \right) \\
{\bf w}_{2} & = & e^{j \nu_2} \cdot
{\sf Dom. \hsppp eig.} \left(  \big( \beta^{\star}(\rho) {\bf \Sigma}_1 +
{\bf I} \big)^{-1} \hsppp  \big( \delta^{\star}(\rho) {\bf \Sigma}_2  +
{\bf I} \big) \right)
\end{eqnarray}
for some choice of $\nu_i \in [0, 2 \pi), \hsppp i = 1,2$ and where
$\{ \alpha^{\star}(\rho), \hspp \beta^{\star}(\rho), \hspp
\gamma^{\star}(\rho), \hspp \delta^{\star}(\rho) \}$ are as in~(\ref{eq3}).
\endproof
\end{conj}

\subsection{Rank-Deficient Case}
Following up on Remark 3 in Sec.~\ref{sec4b}, we now consider the extreme
case where both\footnote{The case when only one of the ${\bf \Sigma}_i$ is
rank-deficient can be studied along analogous lines and no details are
provided.} ${\bf \Sigma}_1$ and ${\bf \Sigma}_2$ are rank-deficient in
more detail.
\begin{prp}
The ergodic information-theoretic rate achievable at user $i$ is
\begin{eqnarray}
E \left[ R_i \right] & = &
h \left( \frac{\rho}{2} \hsppp \lambda_1 \big( {\bf \Sigma}_i \big)
\hsppp \left( | {\bf u}_1 \big( {\bf \Sigma}_i \big)^H {\bf w}_i |^2
+ | {\bf u}_1 \big( {\bf \Sigma}_i \big)^H {\bf w}_j |^2 \right) \right)
- h \left( \frac{\rho}{2} \hsppp \lambda_1 \big( {\bf \Sigma}_i \big)
\hsppp | {\bf u}_1 \big( {\bf \Sigma}_i \big)^H {\bf w}_j |^2 \right),
\nonumber \\
& & {\hspace{4in}}
\hspp j \neq i, \hspp i = 1,2
\end{eqnarray}
where $h(\bullet)$ is as in~(\ref{hx}) and the eigen-decomposition of
${\bf \Sigma}_i$ is
\begin{eqnarray}
{\bf \Sigma}_i = \lambda_1 \big( {\bf \Sigma}_i \big) \cdot
{\bf u}_1 \big( {\bf \Sigma}_i \big) {\bf u}_1 \big(
{\bf \Sigma}_i \big)^H, \hspp i = 1,2.
\end{eqnarray}
\end{prp}
{\vspace{0.05in}}
\begin{proof}
While Theorem~\ref{prop_basic_rate} (as stated) is explicitly dependent
on both 
${\bf \Sigma}_1$ and ${\bf \Sigma}_2$ being of full rank and is
hence not directly applicable in this extreme setting, much of the
analysis follows through. The key to the proof is that all the results in
Appendix~\ref{app_cdf} (Lemmas~\ref{lemma_density2} and~\ref{lemma_density34})
also hold when some of the diagonal entries of ${\bf \Lambda}_i$ are zero.
In fact, this fact is implicitly used to compute $E \left[ I_{i,2} \right]$
in Theorem~\ref{prop_basic_rate}.
\end{proof}

\subsection{Three-User Case: $M = 3$}
We now consider the task of generalizing Theorem~\ref{prop_basic_rate} to the
three-user ($M = 3$) case.
\begin{prp}
The ergodic information-theoretic rate achievable at user $i$ (where
$i = 1,2,3$) with linear beamforming in the three-user case is
\begin{align}
& E \left[R_{i} \right] =
E \left[  I_{i, \hsppp 1} \right] - E \left[ I_{i, \hsppp 2}  \right]
\nonumber
\\ & {\hspace{0.1in}}
= \frac{ {\bf \Lambda}_{i, \hsppp 1}^2 \cdot h \left(
\frac{\rho  {\bf \Lambda}_{i, \hsppp 1} } {3} \right) }
{ \left( {\bf \Lambda}_{i, \hsppp 1} - {\bf \Lambda}_{i, \hsppp 2}
\right) \left( {\bf \Lambda}_{i, \hsppp 1} -
{\bf \Lambda}_{i, \hsppp 3} \right) }
-
\frac{ {\bf \Lambda}_{i, \hsppp 2}^2
\cdot h \left( \frac{\rho  {\bf \Lambda}_{i, \hsppp 2} }{3} \right)
}
{ \left( {\bf \Lambda}_{i, \hsppp 1} - {\bf \Lambda}_{i, \hsppp 2}
\right) \left( {\bf \Lambda}_{i, \hsppp 2} -
{\bf \Lambda}_{i, \hsppp 3} \right) }
+ \frac{ {\bf \Lambda}_{i, \hsppp 3}^2 \cdot h \left(
\frac{\rho  {\bf \Lambda}_{i, \hsppp 3} }{3} \right) }
{ \left( {\bf \Lambda}_{i, \hsppp 1} - {\bf \Lambda}_{i, \hsppp 3}
\right) \left( {\bf \Lambda}_{i, \hsppp 2} -
{\bf \Lambda}_{i, \hsppp 3} \right) }
\nonumber \\
& {\hspace{0.1in}}
+
\frac{ \widetilde{ {\bf \Lambda}}_{i, \hsppp 1}^2
\cdot h \left( \frac {\rho  \widetilde{ {\bf \Lambda}}_{i, \hsppp 1} } {3}
\right) }
{ \left( \widetilde{{\bf \Lambda}}_{i, \hsppp 1} -
\widetilde{{\bf \Lambda}}_{i, \hsppp 2}
\right) \left( \widetilde{{\bf \Lambda}}_{i, \hsppp 1} -
\widetilde{{\bf \Lambda}}_{i, \hsppp 3} \right) }
-
\frac{ \widetilde{ {\bf \Lambda}}_{i, \hsppp 2}^2
\cdot h \left( \frac{\rho  \widetilde{{\bf \Lambda}} _{i, \hsppp 2} }
{3} \right) } { \left( \widetilde{ {\bf \Lambda} }_{i, \hsppp 1} -
\widetilde{ {\bf \Lambda}}_{i, \hsppp 2} \right)
\left( \widetilde{ {\bf \Lambda}}_{i, \hsppp 2} -
\widetilde{ {\bf \Lambda}}_{i, \hsppp 3} \right) }
+
\frac{ \widetilde{ {\bf \Lambda}}_{i, \hsppp 3}^2 \cdot
h \left( \frac {\rho  \widetilde{ {\bf \Lambda}}_{i, \hsppp 3} }
{3} \right)  }
{ \left( \widetilde{ {\bf \Lambda}}_{i, \hsppp 1} -
\widetilde{ {\bf \Lambda}}_{i, \hsppp 3}
\right) \left( \widetilde{ {\bf \Lambda}}_{i, \hsppp 2} -
\widetilde{ {\bf \Lambda}} _{i, \hsppp 3} \right) }
\label{three_user_formula}
\end{align}
where $h(\bullet)$ is as in~(\ref{hx}).
The eigenvalue matrices ${\bf \Lambda}_{i} = {\sf diag}\big( [
{\bf \Lambda}_{i, \hsppp 1}, \hsppp
{\bf \Lambda}_{i, \hsppp 2}, \hsppp
{\bf \Lambda}_{i, \hsppp 3}] \big)$ and
${\widetilde{ \bf \Lambda}}_{i} = {\sf diag}\big([
\widetilde{ {\bf \Lambda}} _{i, \hsppp 1}, \hsppp
\widetilde{ {\bf \Lambda}} _{i, \hsppp 2}, \hsppp
\widetilde{ {\bf \Lambda}}_{i, \hsppp 3}] \big)$ are defined as
in~(\ref{eqna7}) and~(\ref{eqna9}), and can be obtained in terms of the
beamforming vectors and the covariance matrices by solving the
associated cubic equations.
\end{prp}
\begin{proof}
The proof is tedious, but follows along the lines of
Theorem~\ref{prop_basic_rate}. The first step is in characterizing
${p}_i(y)$, which is done in Lemma~\ref{lemma_density34} of
Appendix~\ref{app_cdf}. We can then generalize~(\ref{basic_eqn1})
using~(\ref{eqn_density3}) as
\begin{eqnarray}
E \left[ I_{i, \hsppp 1} \right] & = & \frac{ {\sf I}_1 }
{ \left( {\bf \Lambda}_{i, \hsppp 1} - {\bf \Lambda}_{i, \hsppp 2} \right)
\left( {\bf \Lambda}_{i, \hsppp 1} - {\bf \Lambda}_{i, \hsppp 3} \right)}
+ \frac{ {\sf I}_2 }
{ \left( {\bf \Lambda}_{i, \hsppp 1} - {\bf \Lambda}_{i, \hsppp 3} \right)
\left( {\bf \Lambda}_{i, \hsppp 2} - {\bf \Lambda}_{i, \hsppp 3} \right)} \\
{\sf I}_1 & = & \int_{x = 0}^{\infty} x^2 e^{-x} \int_{y = 0}
^{  {\bf \Lambda}_{i, \hsppp 1} - {\bf \Lambda}_{i, \hsppp 2}  }
y \log \left( 1 + \frac{\rho}{3} {\bf \Lambda}_{i, \hsppp 1} x
- \frac{\rho}{3} x y  \right) dy \hsppp dx \\
{\sf I}_2 & = & \int_{x = 0}^{\infty} x^2 e^{-x} \int_{y = 0}
^{  {\bf \Lambda}_{i, \hsppp 2} - {\bf \Lambda}_{i, \hsppp 3}  }
y \log \left( 1 + \frac{\rho}{3} {\bf \Lambda}_{i, \hsppp 3} x
+ \frac{\rho}{3} x y  \right) dy \hsppp dx.
\end{eqnarray}
These integrals are cumbersome, but straightforward to compute
using~\cite[$4.337(2), 4.337(5)$, p.\ 572]{gradshteyn}. The result
is the expression in the statement of the proposition.
\end{proof}

\subsection{General $M$-User Case}
As can be seen from Appendix~\ref{app_cdf} (Lemma~\ref{lemma_density34}),
the expression for ${p}_{i}(y)$ becomes more complicated as $M$
increases. Without a recourse to ${p}_i(y)$, closed-form expressions
for the ergodic sum-rate of the linear beamforming scheme
can be obtained using a recent advance in~\cite{hammarwall1,hammarwall2} that
allows the computation of the density function of weighted-sum of standard
central chi-squared terms ({\em generalized} chi-squared random variables).
For example, if $ {\bf \Lambda}_i = {\sf diag} \big( [ {\bf \Lambda}_{i, \hsppp 1},
\hspp \cdots, \hspp {\bf \Lambda}_{i, \hsppp M}] \big)$ and
${\bf \Lambda}_{i,j}, \hsppp j = 1, \cdots, M$ are distinct\footnote{More
complicated expressions can be obtained in case $\{ {\bf \Lambda}_{i,j} \}$
are not distinct. These expressions can be derived in a straightforward
manner using the results in~\cite{hammarwall1}.}, we have
\begin{eqnarray}
E \left[ I_{i, \hsppp 1} \right] = \sum_{k=1}^M \prod_{j = 1, \hsppp
j \neq k}^M \frac{ {\bf \Lambda}_{i, k} }
{  {\bf \Lambda}_{i, k} - {\bf\Lambda}_{i, j} } \cdot h \left(
\frac{ \rho {\bf \Lambda}_{i, k} }{M} \right).
\label{gen_form}
\end{eqnarray}
For $E \left[  I_{i, \hsppp 2} \right]$, we replace $\big\{
{\bf \Lambda}_{i, k} \big\}$ with $\big\{ {\widetilde{ \bf \Lambda}}_{i, k}
\big\}$. It can be checked that these expressions
match with the expressions in this paper for the $M = 2$ and $M = 3$
settings. Nevertheless, it is important to note that the formulas
in~(\ref{three_user_formula}) and~(\ref{gen_form}) are in terms of the
eigenvalue matrices $\{ {\bf \Lambda}_i, {\widetilde{ \bf \Lambda}}_i,
\hsppp i = 1, \cdots, M \}$, which become harder (and impossible for
$M \geq 5$) to compute in closed-form as a function of the beamforming
vectors and the transmit covariance matrices as $M$ increases. Tractable
approximations
to the ergodic sum-rate and beamforming vector optimization based on such
approximations are necessary, which is the subject of ongoing work.

\ignore{
\section{stuff to do}
1. prove lemma on density function.
6. show that the special case subcases are not exactly subcases but are two
different ramifications of the general case depending on the eigenvalues
7. provide insight on vectors for sum rate optimization.
9. antenna asymptotics.
10. an algorithm for arbitrary snr + arbitrary weights vector optimization,
Manton's algorithm
14. numerical plots
15. intro/conclusion + references to existing work.
}

\section{Concluding Remarks}
\label{sec6}
This paper considered the design of statistical beamforming
vectors in a MISO broadcast setting to maximize the ergodic sum-rate.
The approach pursued here for the simplest non-trivial problem with
two-users is as follows: first, the beamforming vectors are fixed
and ergodic rate expressions are computed in closed-form in terms of
the covariance matrices of the links and the beamforming vectors. The
optimization of this non-convex function results in a generalized
eigenvector structure for the optimal beamforming vectors, the solution to
maximizing an appropriately-defined ${\sf SINR}$ metric for each
user. This structure generalizes the single-user setup where the
dominant eigen-modes of the transmit covariance matrix of the links are
excited. The main results of this paper are presented in 
Table~I for different ${\sf SNR}$ ($\rho$) assumptions where we use
${\bf u}_1(\bullet)$ and ${\bf u}_2(\bullet)$ to denote the dominant and
sub-dominant eigenvectors of the matrix under consideration.

Possible extensions of this work include unifying the special case of
Theorem~\ref{thm1} with the general case of Theorem~\ref{thm2}, and
proving Conjectures~\ref{thm_general} and~\ref{conj2}. Developing
intuition in the three-user case as well as tractable approximations
in the general $M$-user ($M > 2$) case critically depend on exploiting the
functional structure of the ergodic sum-rate expression. The generalized
eigenvector solution has been seen in other multi-user scenarios as well,
for example, the interference channel problem with two
antennas~\cite{vasanth_isit10_intf}. Generalizing the theme developed in the
broadcast setting to the interference channel setting, the Rayleigh case
to the Ricean case, and the perfect CSI case to the case where only
statistical information is available are all important tasks. 
\begin{center}
\label{table1}
\begin{tabular}{|c|c|c|}
\hline
\hline
\multicolumn{3}{|c|}{Table I: Structure of Optimal Beamforming Vectors
} \\
\hline
\hline
{\rm Objective Function:} & $\argmax \limits_{  {\bf w}_1, \hsppp
{\bf w}_2 } E \left[R_i \right], \hspp i = 1,2$
& $ \argmax \limits_{  {\bf w}_1, \hsppp
{\bf w}_2 }  E \left[ R_1 \right] + E \left[ R_2 \right]$ \\
\hline
\multirow{3}{*}{$\rho \rightarrow 0$}
& ${\bf w}_{i, \hsppp {\sf opt}} = {\bf u}_1 \left( {\bf \Sigma}_i \right)$
& ${\bf w}_{1, \hsppp {\sf opt}} = {\bf u}_1 \left( {\bf \Sigma}_1 \right)$
\\
& ${\bf w}_{j, \hsppp {\sf opt}} = {\sf any} \hspp {\sf vector} \hspp
{\sf on} \hspp {\cal G}(2,1),$
& ${\bf w}_{2, \hsppp {\sf opt}} = {\bf u}_1 \left( {\bf \Sigma}_2 \right)$
\\
& $j \neq i$
(See Prop.~\ref{prop1}) & (See Prop.~\ref{prop_snr_low}) \\
\hline
\multirow{4}{*}{$\rho$ \hsppp {\rm intermediate}}
&  & ${\bf w}_{1,\hsppp {\sf opt}} = {\bf u}_1 \left(
\left( \alpha^{\star}(\rho) {\bf \Sigma}_2 + {\bf I} \right)^{-1} {\bf \Sigma}_1
\right)$ \\
& -- & ${\bf w}_{2,\hsppp {\sf opt}} = {\bf u}_1 \left(
\left( \beta^{\star}(\rho) {\bf \Sigma}_1 + {\bf I} \right)^{-1} {\bf \Sigma}_2
\right)$ \\
& & $\{ \alpha^{\star}(\rho), \beta^{\star}(\rho) \} \geq 0$, chosen \\
& & appropriately (See Conjecture~\ref{thm_general}) \\
\hline
\multirow{3}{*}{$\rho \rightarrow \infty$}
& ${\bf w}_{i, \hsppp {\sf opt}} = {\bf u}_1 \left( {\bf \Sigma}_i \right)$
& ${\bf w}_{1, \hsppp {\sf opt}} = {\bf u}_1 \left( {\bf \Sigma}_2^{-1}
{\bf \Sigma}_1 \right)$
\\
& ${\bf w}_{j, \hsppp {\sf opt}} =
{\bf u}_2 \left( {\bf \Sigma}_i \right),$
& ${\bf w}_{2, \hsppp {\sf opt}} = {\bf u}_1 \left( {\bf \Sigma}_1^{-1}
{\bf \Sigma}_2 \right)$
\\
& $j \neq i$ (See
Prop.~\ref{prop1}) &  (See Theorems~\ref{thm2} and~\ref{thm1})
\\
\hline \hline
\end{tabular}
\end{center}

\appendix
\subsection{Density Function of Weighted-Norm of Isotropically Distributed
Unit-Norm Vectors}
\label{app_cdf}
Towards computing $E[ I_{i,\hsppp 1}]$, we generalize the technique expounded
in~\cite{mukkavilli} where the surface area (that is required) to be computed
is treated as a
differential element of a corresponding solid volume (at a specific radius
value), and the volume of the necessary solid object is calculated using tools
from higher-dimensional integration (geometry). In this direction, we have
\begin{eqnarray}
{p}_i(x) \hsppp dx & \triangleq &
{\sf P} \left( \widehat{ {\bf h} }_{\iid, \hsppp i}^H
\hsppp {\bf \Lambda}_i \hsppp \widehat{ {\bf h} }_{\iid, \hsppp i} \in
[x, x + dx ] \right) \\
{p}_i(x) & = &  \frac{\partial}{  \partial x} \hsppp
{\sf P} \left( \widehat{ {\bf h} }_{\iid, \hsppp i}^H
\hsppp {\bf \Lambda}_i \hsppp \widehat{ {\bf h} }_{\iid, \hsppp i}
\leq x \right)
\end{eqnarray}
with
\begin{eqnarray}
{\sf P} \left( \widehat{ {\bf h} }_{\iid, \hsppp i}^H
\hsppp {\bf \Lambda}_i \hsppp \widehat{ {\bf h} }_{\iid, \hsppp i}
\leq x \right) & = & 1 -
\frac{ {\sf Area} \left(x, \hsppp 1 \right) }
{ {\sf Area} \left( 1 \right)  }
\end{eqnarray}
where
\begin{align}
& {\sf Area} \left(x, \hsppp y  \right) \triangleq
{\sf Area} \left( \widehat{ {\bf h} }_{\iid, \hsppp i}^H \hsppp
{\bf \Lambda}_i \hsppp \widehat{ {\bf h} }_{\iid, \hsppp i}
\geq x, \hsppp \| \widehat{ {\bf h} }_{\iid, \hsppp i} \|^2 = y \right)
\hsppp {\rm and}
\\ & {\hspace{0.8in}} {\sf Area} \left( y \right)  \triangleq
{\sf Area} \left( \| \widehat{ {\bf h} }_{\iid, \hsppp i} \|^2 = y \right)
\end{align}
denote the area of a (unit radius) spherical cap carved out by the
ellipsoid
\begin{eqnarray}
\left\{ \widehat{ {\bf h} }_{\iid, \hsppp i} :
\widehat{ {\bf h} }_{\iid, \hsppp i}^H \hsppp {\bf \Lambda}_i
\hsppp \widehat{ {\bf h} }_{\iid, \hsppp i} = x \right\}
\end{eqnarray}
and the area of a (unit radius) complex sphere, respectively. The
volume of the objects desired in the computation of ${p}_i(x)$ are
\begin{eqnarray}
{\sf Vol} \left( x, r^2 \right) & \triangleq & {\sf Vol} \left(
\widehat{ {\bf h} }_{\iid, \hsppp i}^H
\hsppp {\bf \Lambda}_i \hsppp \widehat{ {\bf h} }_{\iid, \hsppp i}
\geq x, \hspp
\| \widehat{ {\bf h} }_{\iid, \hsppp i} \|^2 \leq r^2
\right) \\
& = & \int_{y = 0}^{r^2} {\sf Area} \left(x, y \right) dy,
\\ {\sf Vol} (r^2) & \triangleq &
{\sf Vol} \left( \| \widehat{ {\bf h} }_{\iid, \hsppp i} \|^2 \leq r^2
\right)
= \int_{x = 0}^{r^2} {\sf Area}(x) dx. 
\end{eqnarray}
Thus, we have
\begin{eqnarray}
{\sf Area} \left( x,1 \right) & = &  \frac{ \partial} { \partial r^2}
{\sf Vol} \left( x, r^2  \right) \Big| _{r = 1} ,
\\ {\sf Area} \left( x \right) & = & \frac{ \partial}{\partial r^2}
{\sf Vol} (r^2) \Big|_{r = 1} \hspp {\rm and} \hspp {\rm hence,}
\\ {p}_i(x) & = & - \frac{ \frac{\partial^2} {\partial x r^2}
{\sf Vol} \left( x, r^2\right) \Big| _{r = 1} }
{ \frac{\partial}{\partial r^2} {\sf Vol} \left( r^2 \right)
\Big|_{r = 1} }. 
\label{eqn4}
\end{eqnarray}

It is important to realize that computing ${\sf Vol} \left( x, r^2\right)$
is {\em non-trivial} even in the simplest case of $M = 2$. This is because every
additional dimension to the complex ellipsoid corresponds to addition of two
real dimensions, thus rendering a geometric visualization impossible. For
example, with $M = 2$, we have the intersection of
two four-dimensional real objects. Nevertheless, the following lemma
captures the complete structure of ${p}_i(x)$ when $M = 2$.
\begin{lem}
\label{lemma_density2}
If $M = 2$, the random variable
$\widehat{ {\bf h} }_{\iid, \hsppp i}^H \hsppp {\bf \Lambda}_i \hsppp
\widehat{ {\bf h} }_{\iid, \hsppp i}$ is uniformly distributed in the
interval $[{\bf \Lambda}_{i, \hsppp 2}, \hsppp {\bf \Lambda}_{i,\hsppp 1}]$.
\end{lem}
\begin{proof}
First, note that it follows
from~\cite[Lemma 2]{mukkavilli} that
\begin{eqnarray}
{\sf Vol}(r^2) = \frac{ \pi^M r^{2M}  }{ M! }.
\end{eqnarray}
For computing ${\sf Vol} \left(x, r^2 \right)$, we follow the same variable
transformation as in~\cite{mukkavilli}. We set $\widehat{ {\bf h} }
_{ \iid, \hsppp i }(k) = r_k \exp(j \phi_k)$ for $k = 1, 2$. The ellipsoid
is contained completely in the sphere of radius $r$ if $r$ is such that
$r \geq \sqrt{ \frac{x}{ {\bf \Lambda}_{i,\hsppp 2} } }$, whereas the
sphere is contained completely in the ellipsoid if
$r \leq \sqrt{ \frac{x}{ {\bf \Lambda}_{i,\hsppp 1} } }$. In the
intermediate regime for $r$, a non-trivial intersection between the two objects
is observed and one can compute the volume by performing a two-dimensional
integration as follows:
\begin{eqnarray}
{\sf Vol} \left( x, r^2 \right) & = &
\iint_{{\cal A} } r_1 r_2 \phi_1 \phi_2 dr_1 dr_2 d\phi_1 d\phi_2
\\ & = & \left( 2 \pi \right)^2 \cdot
\iint_{ {\cal B} } r_1 dr_1 r_2 dr_2 
\\ & = & \left( 2 \pi \right)^2 \cdot
\int_0^{r^{\star} } r_2 dr_2 \int_{L}^U r_1 dr_1 
\end{eqnarray}
where
\begin{eqnarray}
{\cal A} & = & \Big\{ r_1, r_2 \hspp : \hspp r_1^2 {\bf \Lambda}_{i,\hsppp 1}
+ r_2^2 {\bf \Lambda}_{i, \hsppp 2} \geq x, \hspp r_1^2 + r_2^2 \leq r^2
\Big \} \nonumber \\
& & \hspp {\rm and} \hspp \Big \{ \phi_1, \phi_2 \hspp : \hspp
[0, 2 \pi)  \Big \},
\\ {\cal B} & = &
\Big \{ r_1, r_2 \hspp : \hspp r_1^2 {\bf \Lambda}_{i,\hsppp 1}
+ r_2^2 {\bf \Lambda}_{i, \hsppp 2} \geq x, \hspp r_1^2 + r_2^2 \leq r^2
\Big \}, 
\\ L & = & \sqrt{ \frac{ x - r_2^2 \hsppp {\bf \Lambda}_{i, \hsppp 2} }
{ {\bf \Lambda}_{i, \hsppp 1} } }, {\hspace{0.2in}}
U = \sqrt{ r^2 - r_2^2  },
{\hspace{0.2in}} r^{\star} =
\frac{ r^2\hsppp {\bf \Lambda}_{i, \hsppp 1} - x }
{  {\bf \Lambda}_{i,\hsppp 1} - {\bf \Lambda}_{i, \hsppp 2} }.
\end{eqnarray}
Trivial computation establishes the following:
\begin{eqnarray}
{\sf Vol} \left(x, r^2  \right) = \left \{
\begin{array}{cl}
0, & r \leq \sqrt{ \frac{x  }{ {\bf \Lambda}_{i,\hsppp 1} } }
\\
\frac{\pi^2}{2} \cdot \frac{ \left( r^2 \hsppp
{\bf \Lambda}_{i,\hsppp 1} - x \right)^2 }
{  {\bf \Lambda}_{i, \hsppp 1} \cdot \left( {\bf \Lambda}_{i, \hsppp 1} -
{\bf \Lambda}_{i, \hsppp 2}  \right) }, &
\sqrt{ \frac{ x}{ {\bf \Lambda}_{i, \hsppp 1} } } \leq r \leq
\sqrt{ \frac{ x}{ {\bf \Lambda}_{i, \hsppp 2} } } \\
\frac{ \pi^2}{2} \cdot \left( r^4 - \frac{x^2}{  {\bf \Lambda}_{i, \hsppp 1}
\hsppp {\bf \Lambda}_{i, \hsppp 2} }  \right), &
r \geq \sqrt{ \frac{x  }{ {\bf \Lambda}_{i,\hsppp 2} } }.
\end{array}
\right. 
\end{eqnarray}
Another trivial computation using~(\ref{eqn4}) results in
\begin{eqnarray}
{p}_i(x) = \frac{ 1}{ {\bf \Lambda}_{i, \hsppp 1} -
{\bf \Lambda}_{i, \hsppp 2} }. 
\end{eqnarray}
That is, $\widehat{ {\bf h} }_{\iid, \hsppp i}^H \hsppp
{\bf \Lambda}_i \hsppp \widehat{ {\bf h} }_{\iid, \hsppp i}$ is
uniformly distributed in its range.
\end{proof}
The structure of ${p}_i(x)$ gets more complicated as $M$ increases.
We now provide its structure in the $M = 3$ and $M = 4$ cases without proof.
\begin{lem}
\label{lemma_density34}
With $M= 3$, the density function ${p}_i(x)$ is of the form:
\begin{eqnarray}
{p}_i(x) =
\left \{
\begin{array}{cl}
0, & x \leq {\bf \Lambda}_{i, \hsppp 3} \\
\frac{2 \hsppp \left(x - {\bf \Lambda}_{i,\hsppp 3} \right) }
{ \left( {\bf \Lambda}_{i,\hsppp 1} - {\bf \Lambda}_{i, \hsppp 3} \right)
\hsppp \left(  {\bf \Lambda}_{i,\hsppp 2} - {\bf \Lambda}_{i, \hsppp 3}
\right) },
& {\bf \Lambda}_{i, \hsppp 3} \leq x \leq {\bf \Lambda}_{i, \hsppp 2}
\\ \frac{2 \hsppp \left( {\bf \Lambda}_{i,\hsppp 1} - x \right) }
{ \left( {\bf \Lambda}_{i,\hsppp 1} - {\bf \Lambda}_{i, \hsppp 2} \right)
\hsppp \left(  {\bf \Lambda}_{i,\hsppp 1} - {\bf \Lambda}_{i, \hsppp 3}
\right) },
& {\bf \Lambda}_{i, \hsppp 2} \leq x \leq {\bf \Lambda}_{i, \hsppp 1}
\\
0, & x \geq {\bf \Lambda}_{i, \hsppp 1} .
\end{array}
\right. 
\label{eqn_density3}
\end{eqnarray}
With $M= 4$, the density function ${p}_i(x)$ takes the form:
\begin{eqnarray}
{p}_i(x) = \left \{
\begin{array}{cl}
0, & x \leq {\bf \Lambda}_{i, \hsppp 4} \\
\frac{3 \hsppp \left(x - {\bf \Lambda}_{i,\hsppp 4} \right)^2 }
{ \left( {\bf \Lambda}_{i,\hsppp 1} - {\bf \Lambda}_{i, \hsppp 4} \right)
\hsppp
\left(  {\bf \Lambda}_{i,\hsppp 2} - {\bf \Lambda}_{i, \hsppp 4} \right)
\hsppp
\left(  {\bf \Lambda}_{i,\hsppp 3} - {\bf \Lambda}_{i, \hsppp 4} \right)
},
& {\bf \Lambda}_{i, \hsppp 4} \leq x \leq {\bf \Lambda}_{i, \hsppp 3}
\\ 
\frac{ 3 }
{ \left( {\bf \Lambda}_{i,\hsppp 1} - {\bf \Lambda}_{i, \hsppp 3} \right)
\hsppp
\left(  {\bf \Lambda}_{i,\hsppp 2} - {\bf \Lambda}_{i, \hsppp 4}
\right) } \cdot {\cal L}_0, 
& {\bf \Lambda}_{i, \hsppp 3} \leq x \leq {\bf \Lambda}_{i, \hsppp 2}
\\ 
\frac{3 \hsppp \left( {\bf \Lambda}_{i,\hsppp 1} - x \right)^2 }
{ \left( {\bf \Lambda}_{i,\hsppp 1} - {\bf \Lambda}_{i, \hsppp 2} \right)
\hsppp
\left(  {\bf \Lambda}_{i,\hsppp 1} - {\bf \Lambda}_{i, \hsppp 3} \right)
\hsppp
\left(  {\bf \Lambda}_{i,\hsppp 1} - {\bf \Lambda}_{i, \hsppp 4} \right)},
& {\bf \Lambda}_{i, \hsppp 2} \leq x \leq {\bf \Lambda}_{i, \hsppp 1}
\\
0, & x \geq {\bf \Lambda}_{i, \hsppp 1}
\end{array}
\right. 
\label{eqn_density4}
\end{eqnarray}
where
\begin{eqnarray}
{\cal L}_0 & = &
\frac{ \left( x - {\bf \Lambda}_{i, \hsppp 3} \right)
\left( {\bf \Lambda}_{i, \hsppp 2} - x \right) }
{ {\bf \Lambda}_{i, \hsppp 2} - {\bf \Lambda}_{i, \hsppp 3} }
+
\frac{ \left( x - {\bf \Lambda}_{i, \hsppp 4} \right)
\left( {\bf \Lambda}_{i, \hsppp 1} - x \right) }
{ {\bf \Lambda}_{i, \hsppp 1} - {\bf \Lambda}_{i, \hsppp 4} }.
\end{eqnarray}
\endproof
\end{lem}

\begin{figure}[htb!]
\centering
\begin{tabular}{c}
\includegraphics[height=3in,width=3.8in]{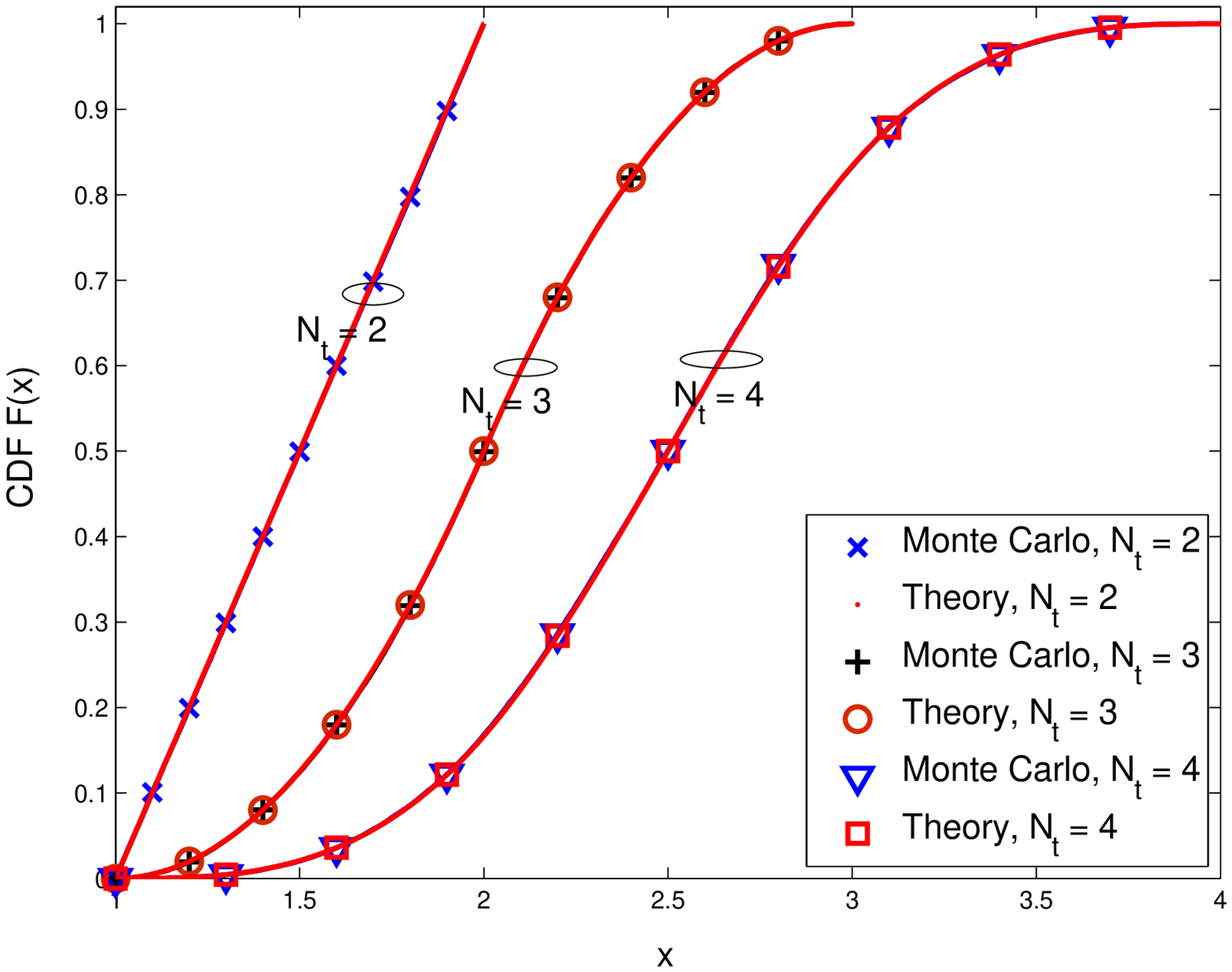}
\end{tabular}
\caption{CDF of weighted-norm of isotropically distributed beamforming
vectors.
\label{fig1}}
\end{figure}

Fig.~\ref{fig1} illustrates the trends of the cumulative distribution
function (CDF) by plotting the fit between the theoretical expressions
in Lemmas~\ref{lemma_density2} and~\ref{lemma_density34}, and the CDF
estimated by Monte Carlo methods. The cases considered are: a)
${\bf \Lambda}_i = {\sf diag}([2\hspp \hsppp 1])$ for $M = 2$,
b) ${\bf \Lambda}_i = {\sf diag}([3 \hsppp \hspp 2\hsppp \hspp 1])$
for $M = 3$, and c) ${\bf \Lambda}_i = {\sf diag}([4 \hspp\hsppp 3
\hsppp\hspp 2\hspp\hsppp 1])$ for $M = 4$. The figure shows the
excellent match between theory and Monte Carlo estimates.

\subsection{Rewriting the Rate Expression in the High-$\snr$ Extreme}
\label{app_rewrite}
A straightforward exercise shows that~(\ref{highsnr_simp}) can be
rewritten as in~(\ref{ri_highsnr}) below:
\begin{eqnarray}
E \left[ R_i \right] & \stackrel{ \rho \rightarrow \infty }
{ \rightarrow} &
\frac{1}{2} \hsppp g\left(d_{ {\bf \Sigma}_i } ( {\bf w}_1, {\bf w}_2  )
\right) + \log \left(1 + \frac{A_i}{B_i} \right) - \log(2)
\label{ri_highsnr}
\end{eqnarray}
where $g(\bullet)$ is a function defined as
\begin{eqnarray}
\label{gstruct}
g(z) & \triangleq &  f(z) + 2\log(z),  \\
f(z) & = & \frac{1}{\sqrt{1 - z^2}} \log \left(
\frac{1 + \sqrt{1 - z^2} }{1 - \sqrt{1 - z^2}} \right), \hsp
0 < z < 1.  \label{fstruct}
\end{eqnarray}
In~(\ref{ri_highsnr}), 
$d_{ {\bf \Sigma}_i } \left( {\bf w}_1, {\bf w}_2 \right)$
is defined as
\begin{eqnarray}
d_{ {\bf \Sigma}_i } \left( {\bf w}_1, {\bf w}_2  \right)
& \triangleq &
\sqrt{ \frac{ 4 \left(A_i B_i - C_i^2  \right) }
{ \left(A_i + B_i  \right)^2 } }
\end{eqnarray}
with $A_i, B_i$ and $C_i$ as in Theorem~\ref{prop_basic_rate}.
As illustrated in Fig.~\ref{fig2b}, $f (\bullet)$ is monotonically
decreasing as a function of its argument and $g(\bullet)$ is increasing
with
\begin{eqnarray}
2\log(2) = \lim_{z \rightarrow 0} g(z)
\leq & g(z) &  \leq \lim_{z \rightarrow 1} g(z) = 2 \\
\infty = \lim_{z \rightarrow 0} f(z) \geq & f(z) & \geq
\lim_{z \rightarrow 1} f(z) = 2.
\end{eqnarray}
\begin{figure}[htb!]
\centering
\begin{tabular}{c}
\includegraphics[height=3in,width=3.8in]{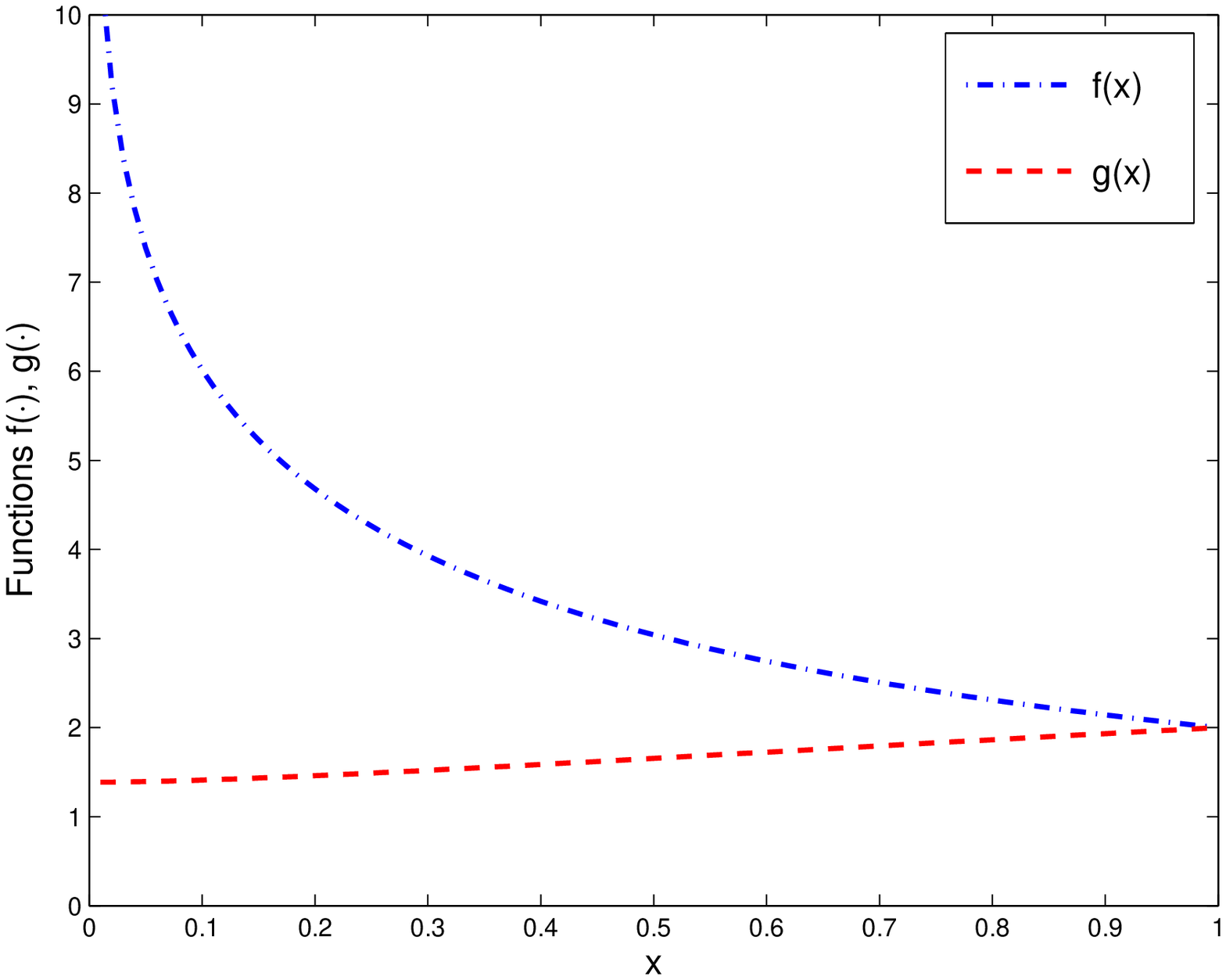}
\end{tabular}
\caption{The behavior of $f(x)$ and $g(x)$.\label{fig2b}}
\end{figure}
Formal proofs of these facts are provided in Appendix~\ref{app_prop2}
next. Some properties of $d_{ {\bf \Sigma}_i } ( {\bf w}_1, {\bf w}_2  )$
are now established.
\begin{itemize}
\item
The quantity $d_{ {\bf \Sigma}_i } \left( {\bf w}_1, {\bf w}_2  \right)$
is a generalized ``distance'' semi-metric\footnote{A semi-metric satisfies
all the properties necessary for a distance metric except the triangle
inequality.} between ${\bf w}_1$ and ${\bf w}_2$ satisfying
\begin{eqnarray}
0 \leq d_{ {\bf \Sigma}_i } \left( {\bf w}_1, {\bf w}_2  \right) \leq 1.
\label{ineq3}
\end{eqnarray}
To establish the lower bound in~(\ref{ineq3}), note that an application of the
Cauchy-Schwarz inequality implies that $C_i^2 \leq A_i B_i$.
Equality in the lower bound in~(\ref{ineq3}) is achieved if and only
if ${\bf \Sigma}_i^{1/2} {\bf w}_1 = \zeta {\bf \Sigma}_i^{1/2} {\bf w}_2$
for some $\zeta \in {\mathbb{C}}$. Since ${\bf \Sigma}_i$ is
positive-definite, this is possible only when ${\bf w}_1 = \zeta
{\bf w}_2$. Since both ${\bf w}_1$ and ${\bf w}_2$ are unit-norm, this
is possible only with $|\zeta| = 1$. In other words, equality in the
lower bound only occurs for ${\bf w}_1 = {\bf w}_2$ on ${\cal G}(2,1)$.
The fact that $d_{ {\bf \Sigma}_i } \left( {\bf w}_1, {\bf w}_2 \right)
\leq 1$ is obvious. Symmetry of the distance metric in ${\bf w}_1$ and
${\bf w}_2$ is obvious.

\item
The triangle inequality does not hold in general. One counter-example
is as follows:
\begin{eqnarray}
\label{excep1}
{\bf \Sigma}_i = {\sf diag} \big( [ 20, \hsppp 1] \big),
\hsp \hsp
{\bf w}_1 =  \left[ \frac{1}{\sqrt{3} }, \hsppp \sqrt{\frac{2}{3}} \right],
\\
{\bf w}_2 = \left[ \frac{1}{\sqrt{2}}, \hsppp \frac{-1} {\sqrt{2}} \right],
\hsp \hsp
{\bf w}_3 = \left[ -\sqrt{\frac{3.3}{7}}, \hsppp \sqrt{ \frac{3.7}{7} } \right].
\label{excep2}
\end{eqnarray}
This choice results in
\begin{eqnarray}
d_{ {\bf \Sigma}_i } \left( {\bf w}_1, {\bf w}_3 \right)
\approx 0.2536, \hspp \hspp
d_{ {\bf \Sigma}_i } \left( {\bf w}_1, {\bf w}_2 \right) +
d_{ {\bf \Sigma}_i } \left( {\bf w}_2, {\bf w}_3  \right)
\approx 0.2534.
\end{eqnarray}
Many such counter-examples can be listed out via a routine numerical
search. Numerical studies also suggest that the triangle inequality
holds for almost all choices of $\{ {\bf w}_i \}$ provided that
$\chi_i = \frac{ \lambda_1( {\bf \Sigma}_i)}
{ \lambda_2( {\bf \Sigma}_i)}$ is not too large (unlike the example
in~(\ref{excep1})-(\ref{excep2})).

\item
The 
upper bound in~(\ref{ineq3}) is achieved only if $A_i = B_i$ and $C_i = 0$.
By decomposing ${\bf w}_1$ and ${\bf w}_2$ along the orthonormal set of
basis vectors $\{ {\bf u}_1({\bf \Sigma}_i), {\bf u}_2({\bf \Sigma}_i) \}$,
it can be checked that $A_i = B_i$ and $C_i = 0$ is possible only if
\begin{eqnarray}
{\bf w}_1 & = & e^{j \nu_1} \cdot \left[ \hsppp {\bf u}_1({\bf \Sigma}_i)
\cdot \sqrt{ \frac{1} {\chi_i + 1} } + e^{j \nu_2} \cdot
{\bf u}_2({\bf \Sigma}_i) \cdot \sqrt{ \frac{\chi_i} {\chi_i + 1} }
\hsppp \right] \label{eq_w1} \\
{\bf w}_2 & = & e^{j (\nu_1 + \nu_3)} \cdot \left[ \hsppp
{\bf u}_1({\bf \Sigma}_i) \cdot \sqrt{ \frac{1} {\chi_i + 1} } -
e^{j \nu_2} \cdot {\bf u}_2({\bf \Sigma}_i) \cdot \sqrt{
\frac{\chi_i} {\chi_i + 1} } \hsppp \right]
\label{eq_w2}
\end{eqnarray}
for some choice of $\nu_j \in [0, 2\pi), \hsppp j = 1,2,3$.

\item
The semi-metric reduces to the standard {\em chordal} distance
metric~\cite{arias_smith} on ${\cal G}(2,1)$
\begin{eqnarray}
d_{ {\bf \Sigma}_i } \left( {\bf w}_1, {\bf w}_2  \right) =
\sqrt{1  - |{\bf w}_1^H {\bf w}_2|^2}
\label{chordal}
\end{eqnarray}
if ${\bf \Sigma}_i = \lambda \hsppp {\bf I}$ for some $\lambda > 0$.
\end{itemize}

\subsection{Monotonicity of $f(\bullet)$ and $g(\bullet)$}
\label{app_prop2}
We first claim that
\begin{eqnarray}
2 \leq \frac{1}{ \sqrt{1 - z^2}} \cdot \log \left(
\frac{1 + \sqrt{1 - z^2} } { 1 - \sqrt{1 - z^2}}  \right) \leq
\frac{2}{z^2}, \hspp \hspp 0  < z < 1,
\label{mission4}
\end{eqnarray}
which is equivalent to:
\begin{eqnarray}
\exp \left( 2 \cdot \sqrt{1 - z^2} \right) \leq
\frac{ 1 + \sqrt{1 - z^2} } { 1 - \sqrt{1 - z^2}} \leq
\exp \left( \frac{2 \sqrt{1 - z^2} } {z^2} \right),
\hspp \hspp 0 < z < 1.
\end{eqnarray}
For this, we start with the exponential series expansion of
$\exp \left( 2 \cdot \sqrt{1 - z^2} \right)$ that results in:
\begin{eqnarray}
\exp \left( 2 \cdot \sqrt{1 - z^2} \right) & = & 1 +
\sum_{k=1}^{\infty} \frac{ 2^k \cdot \left( 1 - z^2 \right)^{\frac{k}{2} } }
{ \Gamma(k + 1)} \\
& \leq & 1 + 2 \sum_{k=1}^{\infty} \left( 1 - z^2 \right)^{
\frac{k}{2} } = 1 + \frac{ 2\sqrt{1 - z^2}} { 1  - \sqrt{1 - z^2}}
\end{eqnarray}
where $\Gamma(\cdot)$ is the Gamma function,
the second inequality follows from the fact that $\frac{2 ^{k-1}}
{ \Gamma(k+1)} \leq 1$ for all $k \geq 1$, and the last equality from the
sum of an infinite geometric series. For the other side
of~(\ref{mission4}), note that
\begin{eqnarray}
\exp \left( \frac{2 \sqrt{1 - z^2} } {z^2} \right) & \geq & 1
+ \frac{ 2 \sqrt{1 - z^2} }{z^2} + \frac{ 2 (1 - z^2)}{z^4} \\
& \geq & 1 + \frac{2 \sqrt{1 - z^2} } {z^2} \left( 1 +
\sqrt{1 - z^2} \right) = 1 + \frac{ 2\sqrt{1 - z^2}} { 1  - \sqrt{1 - z^2}}
\end{eqnarray}
where the first inequality follows by truncating the terms of the
asymptotic expansion and the second follows by using the fact that
$z^2 < 1$. The proof is complete by noting that
\begin{eqnarray}
\frac{\partial f(z)}{ \partial z} & = & \frac{-z}{1 - z^2} \cdot
\left[ \frac{2}{z^2} - \frac{1}{\sqrt{1-z^2}} \log \left(
\frac{1 + \sqrt{1 - z^2}}{1 - \sqrt {1- z^2}} \right) \right] < 0 \\
\frac{\partial g(z)}{ \partial z} & = & \frac{z}{1 - z^2} \cdot
\left[ \frac{1}{\sqrt{1-z^2}} \log \left( \frac{1 + \sqrt{1 - z^2}}
{1 - \sqrt {1- z^2}} \right) - 2 \right] > 0.
\end{eqnarray}
\endproof

\subsection{Completing the Proof of Theorem~\ref{thm2}}
\label{app_gen}
With the description of ${\bf w}_1$ and ${\bf w}_2$ as
in~(\ref{new_eqn1})-(\ref{new_eqn2}), elementary algebra shows that
\begin{eqnarray}
A_1 & = & {\bf w}_1^H {\bf \Sigma}_1 {\bf w}_1 =
\frac{ |\alpha|^2 \eta_1 + |\beta|^2 \eta_2}{X^2}, \hspp
X = \| \alpha \hsppp {\bf \Sigma}_2^{-\frac{1}{2} } \hsppp {\bf v}_1 +
\beta \hsppp {\bf \Sigma}_2^{-\frac{1}{2} } \hsppp {\bf v}_2 \| \\
B_1 & = & {\bf w}_2^H {\bf \Sigma}_1 {\bf w}_2 =
\frac{ |\gamma|^2 \eta_1 + |\delta|^2 \eta_2}{Y^2}, \hspp
Y = \| \gamma \hsppp {\bf \Sigma}_2^{-\frac{1}{2} } \hsppp {\bf v}_1 +
\delta \hsppp {\bf \Sigma}_2^{-\frac{1}{2} } \hsppp {\bf v}_2 \|
\\
C_1 & = & | {\bf w}_1^H {\bf \Sigma}_1 {\bf w}_2| =
\frac{ \big| \alpha^{\star} \gamma \eta_1 + \beta^{\star} \delta \eta_2
\big| } {XY}
\\
A_2 & = & {\bf w}_2^H {\bf \Sigma}_2 {\bf w}_2 = \frac{1}{Y^2}
\\
B_2 & = & {\bf w}_1^H {\bf \Sigma}_2 {\bf w}_1 = \frac{1}{X^2}
\\
C_2 & = & | {\bf w}_1^H {\bf \Sigma}_2 {\bf w}_2 | = \frac{
\big| \alpha^{\star} \gamma + \beta^{\star} \delta \big|}{XY}.
\end{eqnarray}
As in the statement of the theorem, let $\tau_i, \hsppp i =1,2,3$ denote
\begin{eqnarray}
\tau_1 = {\bf v}_1^H \hsppp {\bf \Sigma}_2^{-1} \hsppp {\bf v}_1,
\hsp \hsp
\tau_2 = {\bf v}_2^H \hsppp {\bf \Sigma}_2^{-1} \hsppp {\bf v}_2,
\hsp \hsp
\tau_3 = {\bf v}_1^H \hsppp {\bf \Sigma}_2^{-1} \hsppp {\bf v}_2.
\end{eqnarray}
We can rewrite $X^2$ and $Y^2$ in terms of $\{ \tau_i \}$ as
\begin{eqnarray}
X^ 2 & = & |\alpha|^2 \tau_1 + |\beta|^2 \tau_2 + 2 |\alpha| |\beta|
|\tau_3| \cos( \theta_1)
\label{Xsquared}
\\
Y^2 & = & |\gamma|^2 \tau_1 + |\delta|^2 \tau_2 + 2 |\gamma| | \delta|
|\tau_3 | \cos(\theta_2)
\label{Ysquared}
\end{eqnarray}
where $\theta_1 = {\sf arg}( \tau_3 ) 
+ \theta_{\beta} - \theta_{\alpha}$ and $\theta_2 =
{\sf arg}( \tau_3) 
+ \theta_{\delta} - \theta_{\gamma}$. Now note that if $\left\{ {\bf v}_1,
\hsppp {\bf v}_2 \right\}$ is a pair of eigenvectors for ${\bf \Sigma}$, then so
is the pair $\left\{ e^{j \nu_1} \hsppp {\bf v}_1, \hsppp e^{j \nu_2} \hsppp
{\bf v}_2 \right\}$ for any choice of $\nu_1$ and $\nu_2$. In other words,
the choice of $\left\{ {\bf v}_1, \hsppp {\bf v}_2 \right\}$ is unique
{\em only} on ${\cal G}(2,1)$, and not on ${\sf St}(2,1)$.
Hence, ${\sf arg}(\tau_3)$
can be chosen arbitrarily and {\em independently} in
determining the values of $X^2$ and $Y^2$. With the specific choice that
${\sf arg}( {\bf v}_1^H {\bf \Sigma}_2^{-1} {\bf v}_2 ) = \frac{\pi}{2} +
\theta_{\alpha} - \theta_{\beta}$ in~(\ref{Xsquared}) and
${\sf arg}( {\bf v}_1^H {\bf \Sigma}_2^{-1} {\bf v}_2 ) =
\frac{\pi}{2} + \theta_{\gamma} - \theta_{\delta}$ in~(\ref{Ysquared}),
we have
\begin{eqnarray}
X^2 & = & |\alpha|^2 \tau_1 + |\beta|^2 \tau_2 \\
Y^2 & = & |\gamma|^2 \tau_1 + |\delta|^2 \tau_2.
\end{eqnarray}
Thus, the high-$\snr$ expression for the ergodic sum-rate can be simplified as
\begin{align}
& 2 E \left[ R_1 \right] + 2 E \left[ R_2 \right] + 4 \log(2)
\nonumber \\
& \hsp \hsp =
g \left( \frac{ 2 \sqrt{\eta_1 \eta_2} XY \cdot |\beta \gamma - \alpha \delta| }
{ \Big(  |\alpha|^2 \eta_1 + |\beta|^2 \eta_2
\Big) \cdot Y^2 + \Big( |\gamma|^2 \eta_1 + |\delta|^2 \eta_2 \Big)
\cdot X^2 } \right) + 2 \log \left( 1 + \frac{X^2}{Y^2} \right) \nonumber \\
& \hsp \hsp \hsp \hsp +
g \left( \frac{2 X Y \cdot | \beta \gamma - \alpha \delta| }
{ X^2 + Y^2 } \right)
+ 2 \log \left( 1 + \frac{Y^2}{X^2} \cdot
\frac{ |\alpha|^2 \eta_1 + |\beta|^2 \eta_2}
{ |\gamma|^2 \eta_1 + |\delta|^2 \eta_2}  \right) .
\label{thm2_eqnimp}
\end{align}
Using~(\ref{fstruct}) to rewrite the above equation in terms of
$f(\bullet)$, we have after simplification:
\begin{align}
& {\hspace{-0in}}
2 E \left[ R_1 \right] + 2 E \left[ R_2 \right] - \log(\eta_1 \eta_2)
= f \left( \frac{ 2 \sqrt{\eta_1 \eta_2} XY \cdot | \beta \gamma - \alpha \delta|  }
{ \big(  |\alpha|^2 \eta_1 + |\beta|^2 \eta_2 \big) Y^2 +
\big( |\gamma|^2 \eta_1 + |\delta|^2 \eta_2 \big) X^2 }
\right)
\nonumber \\ & {\hspace{2.2in}}
+ f \left( \frac{2 X Y \cdot | \beta \gamma - \alpha \delta|  }
{ X^2 + Y^2 } \right)
+ 2 \log \left( \frac{ | \beta \gamma - \alpha \delta|^2 }
{ |\gamma|^2 \eta_1 + |\delta|^2 \eta_2 } \right).
\end{align}
We now claim that the following two inequalities hold:
\begin{eqnarray}
\frac{ X Y }{ X^2+Y^2} & \geq & \frac{\sqrt{\tau_1\tau_2}}
{ \tau_1 + \tau_2}
\label{ineq1} \\
\frac{XY \cdot \big( |\gamma|^2 \eta_1 + |\delta|^2 \eta_2 \big) }
{ \big(  |\alpha|^2 \eta_1 + |\beta|^2 \eta_2 \big) Y^2 +
\big( |\gamma|^2 \eta_1 + |\delta|^2 \eta_2 \big) X^2} & \geq &
\frac{ \sqrt{ \tau_1 \tau_2} \cdot \eta_2 } { \tau_1 \eta_2 + \tau_2 \eta_1 }.
\label{ineq2}
\end{eqnarray}
The proof of~(\ref{ineq1}) and~(\ref{ineq2}) will be tackled later.

Using~(\ref{ineq1}) and~(\ref{ineq2}) in conjunction with the decreasing nature
of $f(\bullet)$, we have
\begin{align}
& {\hspace{-0in}}
2 E \left[ R_1 \right] + 2 E \left[ R_2 \right] - \log(\eta_1 \eta_2)
\leq f \left( \frac{ 2 \sqrt{\eta_1 \eta_2 \tau_1 \tau_2} \cdot \eta_2 }
{ \tau_1 \eta_2 + \tau_2 \eta_1} \cdot
\frac{ | \beta \gamma - \alpha \delta|  }
{\big( |\gamma|^2 \eta_1 + |\delta|^2 \eta_2 \big) } \right)
\nonumber \\
& {\hspace{1.6in}}
+ f \left( \frac{2 \sqrt{\tau_1\tau_2}} {\tau_1 + \tau_2} \cdot
| \beta \gamma - \alpha \delta| \right)
+ 2 \log \left( \frac{ | \beta \gamma - \alpha \delta|^2 }
{ |\gamma|^2 \eta_1 + |\delta|^2 \eta_2 } \right)
\\ & {\hspace{2in}} =
g \left( \frac{ 2 \sqrt{\eta_1 \eta_2 \tau_1 \tau_2} \cdot \eta_2 }
{ \tau_1 \eta_2 + \tau_2 \eta_1} \cdot
\frac{ | \beta \gamma - \alpha \delta|  }
{\big( |\gamma|^2 \eta_1 + |\delta|^2 \eta_2 \big) } \right)
\nonumber \\
& {\hspace{1.6in}}
+ g \left( \frac{2 \sqrt{\tau_1\tau_2}} {\tau_1 + \tau_2} \cdot
| \beta \gamma - \alpha \delta| \right)
+ 2 \log \left(
\frac{ \big( \tau_1 \eta_2 + \tau_2 \eta_1 \big)
\big(\tau_1 + \tau_2 \big) } { 4 \tau_1 \tau_2 \eta_2 \sqrt{\eta_1 \eta_2}
} \right).
\end{align}
Note that $\{ \theta_{\bullet} \}$ enter the above optimization only via
the term $| \beta \gamma - \alpha \delta |$ and
\begin{eqnarray}
|\beta \gamma - \alpha \delta | \leq |\alpha| \sqrt{1 - |\gamma|^2} +
|\gamma | \sqrt{1 - |\alpha|^2}
\label{eqn_ii3}
\end{eqnarray}
with equality achieved if and only if $\theta_{\alpha} + \theta_{\delta}
- \theta_{\beta} - \theta_{\gamma} = \pi$ (modulo $2\pi$). Parameterizing
$|\alpha|$ and $|\gamma|$ as $|\alpha| = \sin(\theta)$ and $|\gamma| =
\sin(\phi)$ for some $\{ \theta, \phi \} \in [0, \pi/2]$, we have
\begin{eqnarray}
|\beta \gamma - \alpha \delta | \leq \sin(\theta) \cos(\phi) + \cos(\theta)
\sin(\phi) = \sin(\theta + \phi) \leq 1
\end{eqnarray}
since $0 \leq \theta + \phi \leq \pi$. Further, $\eta_1 \geq \eta_2$
implies that
$|\gamma|^2 \eta_1 + |\delta|^2 \eta_2 \geq \eta_2$ and hence, we have
\begin{eqnarray}
\frac{ | \beta \gamma - \alpha \delta| }
{  |\gamma|^2 \eta_1 + |\delta|^2 \eta_2} \leq \frac{1}{\eta_2}.
\end{eqnarray}
Using the fact that $g(\bullet)$ is an increasing function, we get an upper
bound for the sum-rate as
\begin{eqnarray}
2 E \left[ R_1 \right] + 2 E \left[ R_2 \right]
\leq f \left( \frac{2 \sqrt{\eta_1 \eta_2 \tau_1\tau_2}  }
{\eta_1 \tau_2 + \eta_2 \tau_1} \right) +
f \left( \frac{ 2 \sqrt{\tau_1 \tau_2} } { \tau_1 + \tau_2} \right)
+ \log \left( \frac{\eta_1}{\eta_2} \right).
\end{eqnarray}
This upper bound is achievable with the choice of $|\alpha| = 1$ and $|\gamma| =
0$ in~(\ref{new_eqn1}) and~(\ref{new_eqn2}), which is equivalent to~(\ref{bformer}).
Substituting this choice in the sum-rate expression yields
\begin{eqnarray}
E \left[ R_1 \right] + E \left[ R_2 \right]
& \stackrel{\rho \rightarrow \infty} {\rightarrow} &
\frac{1}{2} \hsppp f \left( \frac{2 \sqrt{\eta_1 \eta_2 \tau_1\tau_2}  }
{\eta_1 \tau_2 + \eta_2 \tau_1} \right) +
\frac{1}{2} \hsppp
f \left( \frac{ 2 \sqrt{\tau_1 \tau_2} } { \tau_1 + \tau_2} \right)
+ \frac{1}{2} \log \left( \frac{\eta_1}{\eta_2} \right) \\
& = & \frac{1}{2} \cdot \frac{ \eta_1 \tau_2 + \eta_2 \tau_1}
{ | \eta_1 \tau_2 - \eta_2 \tau_1| } \cdot \log \left(
\frac{ \eta_1 \tau_2 + \eta_2 \tau_1 + | \eta_1 \tau_2 - \eta_2 \tau_1| }
{ \eta_1 \tau_2 + \eta_2 \tau_1 - | \eta_1 \tau_2 - \eta_2 \tau_1| }
\right) \nonumber \\
&&+ \frac{1}{2} \cdot \frac{\tau_1 + \tau_2}{|\tau_1 - \tau_2|}
\cdot \log \left( \frac{ \tau_1 + \tau_2 + | \tau_1 - \tau_2 |}
{ \tau_1 + \tau_2 - |\tau_1 - \tau_2|} \right) +
\frac{1}{2} \hsppp \log \left( \frac{\eta_1}{\eta_2} \right).
\label{sumrate}
\end{eqnarray}
To simplify~(\ref{sumrate}), we define $\kappa_{1}$
and $\kappa_{2}$ as 
\begin{eqnarray}
\kappa_{1} =
\frac{ \eta_1 \tau_2}{ \eta_2 \tau_1} & {\rm and} &
\kappa_{2} = \frac{\tau_2}{\tau_1}.
\end{eqnarray}
The fact that $\eta_1 \geq \eta_2$ implies that $\kappa_{1} \geq \kappa_{2}$.
Thus, there are three possibilities: i) $\kappa_{1} \geq \kappa_{2} \geq 1$,
ii) $\kappa_{1} \geq 1 \geq \kappa_{2}$, and iii) $1 \geq \kappa_{1} \geq
\kappa_{2}$. It is straightforward but tedious to check that in all of the
three cases,~(\ref{sumrate}) reduces to~(\ref{rhs3}) as in the statement of
the theorem. The proof will be complete if the inequalities~(\ref{ineq1})
and~(\ref{ineq2}) can be established.

\noindent {\bf \em Proof of~(\ref{ineq1}):} For the first inequality,
note that
\begin{eqnarray}
\frac{X^2 + Y^2}{XY} = \frac{X}{Y} + \frac{Y}{X} = t + \frac{1}{t} \triangleq
q(t)
\end{eqnarray}
can be written as a symmetric function $q(t)$ in $t$ where $t = \frac{X}{Y}$.
Further, noting that $q(t)$ is decreasing in $t$ for $t \leq 1$ and is
increasing in $t$ for $t \geq 1$, the maximum of $\frac{X^2 + Y^2}{XY}$ is
achieved either when $\frac{X}{Y}$ achieves its largest or smallest value.
The inequality in~(\ref{ineq1}) follows since
\begin{eqnarray}
\sqrt{ \frac{ \min(\tau_1, \tau_2) }  { \max(\tau_1, \tau_2) } }
\leq \frac{X}{Y} \leq
\sqrt{ \frac{ \max(\tau_1, \tau_2) } { \min(\tau_1, \tau_2) } }.
\end{eqnarray}

\noindent {\bf \em Proof of~(\ref{ineq2}):} The proof of~(\ref{ineq2})
is more involved. For this, note that
\begin{align}
& {\hspace{0.1in}}
{\cal L}_1 \triangleq
\frac{ \big(  |\alpha|^2 \eta_1 + |\beta|^2 \eta_2 \big) Y^2 +
\big( |\gamma|^2 \eta_1 + |\delta|^2 \eta_2 \big) X^2 }
{ XY \cdot  \big( |\gamma|^2 \eta_1 + |\delta|^2 \eta_2 \big) }
\\
& {\hspace{0.3in}} =
\sqrt{ \frac{ |\alpha|^2 (\tau_1 - \tau_2) + \tau_2 } {
|\gamma|^2 (\tau_1 - \tau_2) + \tau_2 } } \cdot
\left(1 + \frac{ |\alpha|^2 (\eta_1 - \eta_2) + \eta_2 } {
|\gamma|^2 (\eta_1 - \eta_2) + \eta_2 } \cdot
\frac { |\gamma|^2 (\tau_1 - \tau_2) + \tau_2  }
{|\alpha|^2 (\tau_1 - \tau_2) + \tau_2  } \right).
\label{eqn10}
\end{align}
By taking derivative with respect to $|\alpha|^2$, note that
the first term in~(\ref{eqn10}) is increasing in $|\alpha|^2$ for any fixed
choice of $|\gamma|$ if and only if $\frac{\tau_1}{\tau_2} \geq 1$. Similarly,
for any fixed choice of $|\gamma|$, the second term in~(\ref{eqn10}) is
increasing in $|\alpha|^2$ if and only if $\frac{\eta_1}{\eta_2} \geq
\frac{\tau_1}{\tau_2}$. Thus, the condition
\begin{eqnarray}
1 \leq \frac{\tau_1}{\tau_2} \leq \frac{\eta_1} {\eta_2}
\label{eqn11}
\end{eqnarray}
is necessary and sufficient to ensure that for any choice of $|\gamma|$,
${\cal L}_1$ is maximized by the choice $|\alpha| = 1$. On analogous lines,
taking the derivative with respect to $|\gamma|^2$, it can be seen that for
any fixed choice of $|\alpha|$, both the terms in~(\ref{eqn10}) are
decreasing in $|\gamma|^2$ if and only if the same condition in~(\ref{eqn11})
holds. In other words, under~(\ref{eqn11}), ${\cal L}_1$ is maximized by
$|\alpha| = 1$ and $|\gamma| = 0$. At this stage, two other possibilities
need to be considered: i)
$\frac{\tau_1}{\tau_2} \leq 1 \leq \frac{\eta_1} {\eta_2}$, and ii)
$1 \leq \frac{\eta_1}{\eta_2} \leq \frac{\tau_1} {\tau_2}$. In either
case, we will show that
\begin{eqnarray}
\left( \frac{ |\alpha|^2 (\eta_1 - \eta_2) + \eta_2}
{ |\gamma|^2 ( \eta_1 - \eta_2) + \eta_2} - \frac{\eta_1}{\eta_2} \right)
\cdot \frac{Y}{X} + \left( \frac{X}{Y} - \sqrt{\frac{\tau_1}{\tau_2}} \right)
\cdot \left( 1  - \frac{\eta_1}{\eta_2} \sqrt{\frac{ \tau_2}{\tau_1}} \frac{Y }{X}
\right) \leq 0,
\label{eqn15}
\end{eqnarray}
which is equivalent to~(\ref{ineq2}),
or the statement that ${\cal L}_1$ is maximized by $|\alpha| = 1$ and $|\gamma| = 0$. For
this, note that in either case, we have
\begin{eqnarray}
\frac{ |\alpha|^2 (\eta_1 - \eta_2) + \eta_2}
{ |\gamma|^2 ( \eta_1 - \eta_2) + \eta_2} \leq \frac{\eta_1}{\eta_2}.
\label{eqn13}
\end{eqnarray}
In the first case, we also have
\begin{eqnarray}
\sqrt{ \frac{\tau_1}{ \tau_2} } \leq \frac{X}{Y}
\leq \sqrt{ \frac{\tau_2}{ \tau_1} } \leq
\frac{\eta_1}{\eta_2} \sqrt{ \frac{\tau_2}{ \tau_1} },
\label{eqn14}
\end{eqnarray}
where the last step in~(\ref{eqn14}) follows from $\frac{\eta_1}{\eta_2}
\geq 1$. Combining~(\ref{eqn13}) and~(\ref{eqn14}), we note that~(\ref{eqn15})
is immediate when $\frac{\tau_1}{\tau_2} \leq 1 \leq \frac{\eta_1}{\eta_2}$.
In the second case, however,~(\ref{eqn14}) is replaced with
\begin{eqnarray}
\sqrt{ \frac{\tau_2}{ \tau_1} } \leq \frac{X}{Y}
\leq \sqrt{ \frac{\tau_1}{ \tau_2} }. \label{eqn14p}
\end{eqnarray}
It can be seen that if $|\alpha|$ and $|\gamma|$ are such that
\begin{eqnarray}
\frac{X}{Y} = D \cdot \frac{\eta_1}{\eta_2} \sqrt{\frac{\tau_2}{\tau_1} },
\end{eqnarray}
for some choice of $D$ satisfying
$1 \leq D \leq \frac{ \tau_1}{\tau_2} \cdot \frac{ \eta_2} {\eta_1}$,~(\ref{eqn15}) holds
immediately. Thus, we only need to show that~(\ref{eqn15}) holds when
$|\alpha|$ and $|\gamma|$ are such that
\begin{eqnarray}
\frac{X}{Y} = D \cdot \frac{\eta_1}{\eta_2} \sqrt{\frac{\tau_2}{\tau_1} },
\end{eqnarray}
for some choice of $D$ satisfying $\frac{\eta_2}{\eta_1} \leq D \leq 1$.
After some elementary algebra, our task is to show that
\begin{eqnarray}
\frac{ |\alpha|^2 (\eta_1 - \eta_2)  + \eta_2}
{ |\gamma|^2 (\eta_1 - \eta_2) + \eta_2} & \leq &
\frac{\eta_1}{\eta_2} \cdot \left( 1 - (1 - D) \cdot
\left( 1  - D \cdot \frac{\eta_1}{\eta_2} \cdot \frac{\tau_2}{\tau_1}
\right) \right) \label{eq20}
\\
\frac{X}{Y} & = &
D \cdot \frac{\eta_1}{\eta_2} \sqrt{\frac{\tau_2}{\tau_1} }
\end{eqnarray}
By bounding the denominator of~(\ref{eq20}) as
$ \eta_2 \leq |\gamma|^2 (\eta_1 - \eta_2 ) + \eta_2 \leq \eta_1$, it
can be seen that~(\ref{eq20}) holds if the following quadratic inequality
in $D$ is true:
\begin{eqnarray}
D^2 \cdot \frac{ \eta_1^2 \tau_2} {\eta_2^2 \tau_1} \cdot
\left( 2 - \frac{ \tau_1 \eta_2 - \tau_2 \eta_1} { \eta_1( \tau_1 - \tau_2)}
\right) - D \cdot \frac{\eta_1}{\eta_2} \cdot \left( 1 + \frac{ \eta_1 \tau_2}
{\eta_2 \tau_1} \right) +
\frac{ \tau_1 \eta_2 - \tau_2 \eta_1} { \eta_2( \tau_1 - \tau_2)} \leq 0.
\label{eq21}
\end{eqnarray}
For this, we note that the left-hand side represents a convex
parabola in $D$ with maximum achieved at either $D = \frac{\eta_2}{\eta_1}$
or $D = 1$. Substituting $D = \frac{\eta_2}{\eta_1}$ and $D = 1$ and
simplifying, we see that
\begin{eqnarray}
{\sf LHS \hspp of \hspp (\ref{eq21})} \Big|_{D = \frac{\eta_2}{\eta_1}}
& = & - \frac{ \tau_2}{ \eta_1 \eta_2 \tau_1} \cdot
\frac{ \eta_1 - \eta_2}{ \tau_1 - \tau_2} \cdot
\left( \eta_1 (\tau_1 - \tau_2) + \tau_1 (\eta_1 - \eta_2) \right) \leq 0
\\
{\sf LHS \hspp of \hspp (\ref{eq21})} \Big|_{D = 1}
& = & - \frac{ (\eta_1 - \eta_2) \cdot (\eta_2 \tau_1 - \eta_1 \tau_2)}
{ \eta_2^2 (\tau_1 - \tau_2)} \leq 0.
\end{eqnarray}
Since the maximum of the parabola in the domain $\frac{\eta_2}{\eta_1}
\leq D \leq 1$ is below $0$,~(\ref{eqn15}) holds. Thus, we are done
with the aspect of showing that $|\alpha| = 1, |\gamma| = 0$ is sum-rate
optimal.
\endproof

\subsection{Proof of Theorem~\ref{thm1}}
\label{app_thm1}
Following the logic of Appendix~\ref{app_gen}, we decompose ${\bf w}_1$
and ${\bf w}_2$ along $\{ {\bf u}_1, {\bf u}_2 \}$ since they form an
orthonormal basis:
\begin{eqnarray}
{\bf w}_1 & = & \alpha {\bf u}_1 + \beta {\bf u}_2
\label{ref1}
\\
{\bf w}_2 & = & \gamma {\bf u}_1 + \delta {\bf u}_2
\label{ref2}
\end{eqnarray}
for some choice of $\{ \alpha, \beta, \gamma, \delta \}$ with $\alpha =
|\alpha| e^{j \theta_{\alpha}}$ (similarly, for other quantities)
satisfying $| \alpha |^2 + |\beta |^2 = |\gamma|^2 + | \delta |^2 = 1$. We
now study the ergodic sum-rate optimization over the six-dimensional parameter
space. 
With the description of ${\bf w}_1$ and ${\bf w}_2$ as
in~(\ref{ref1})-(\ref{ref2}), elementary algebra shows that
\begin{eqnarray}
A_1 & = & {\bf w}_1^H {\bf \Sigma}_1 {\bf w}_1 =
|\alpha|^2 \lambda_1 + |\beta|^2 \lambda_2 \\
B_1 & = & {\bf w}_2^H {\bf \Sigma}_1 {\bf w}_2 =
|\gamma|^2 \lambda_1 + |\delta|^2 \lambda_2 \\
C_1 & = & | {\bf w}_1^H {\bf \Sigma}_1 {\bf w}_2| =
| \alpha^{\star} \gamma \lambda_1 + \beta^{\star} \delta
\lambda_2 | \\
A_2 & = & {\bf w}_2^H {\bf \Sigma}_2 {\bf w}_2 =
|\gamma|^2 \mu_1 + |\delta|^2 \mu_2 \\
B_2 & = & {\bf w}_1^H {\bf \Sigma}_2 {\bf w}_1 =
|\alpha|^2 \mu_1 + |\beta|^2 \mu_2 \\
C_2 & = & | {\bf w}_1^H {\bf \Sigma}_2 {\bf w}_2 | =
| \alpha^{\star} \gamma \mu_1 + \beta^{\star} \delta
\mu_2 |
\end{eqnarray}
and hence,
\begin{eqnarray}
d_{ {\bf \Sigma}_1 } ( {\bf w}_1, {\bf w}_2 )^2 & = &
\frac{ 4(A_1 B_1 - C_1^2)} { (A_1 + B_1)^2 } =
\frac{ 4 \lambda_1 \lambda_2 \cdot |\beta \gamma - \alpha \delta|^2 }
{ \Big[ (|\alpha|^2 + |\gamma|^2) \lambda_1 +
( |\beta|^2 + |\delta|^2 ) \lambda_2 \Big]^2 } \\
d_{ {\bf \Sigma}_2 } ( {\bf w}_1, {\bf w}_2 )^2 & = &
\frac{ 4(A_2 B_2 - C_2^2)} { (A_2 + B_2)^2 } =
\frac{ 4 \mu_1 \mu_2 \cdot |\beta \gamma - \alpha \delta|^2 }
{ \Big[ (|\alpha|^2 + |\gamma|^2) \mu_1 +
( |\beta|^2 + |\delta|^2 ) \mu_2 \Big]^2 }.
\end{eqnarray}
The high-$\snr$ expression of the ergodic sum-rate can be written as
\begin{align}
& 2 E \left[ R_1 \right] + 2 E \left[ R_2 \right] + 4 \log(2)
\nonumber \\
& \hsp \hsp = g \left( \frac{
\sqrt{ 4 \lambda_1 \lambda_2 } \cdot |\beta \gamma - \alpha \delta| }
{ (|\alpha|^2 + |\gamma|^2) \lambda_1 + ( |\beta|^2 + |\delta|^2 ) \lambda_2 }
\right) +
g \left( \frac{ \sqrt{4 \mu_1 \mu_2} \cdot |\beta \gamma - \alpha \delta| }
{ (|\alpha|^2 + |\gamma|^2) \mu_1 + ( |\beta|^2 + |\delta|^2 ) \mu_2 }
\right) \nonumber \\
& \hsp \hsp \hsp \hsp {\hspace{0.5in}} +
2 \log \left( 1 + \frac{ |\alpha|^2 \lambda_1 + |\beta|^2 \lambda_2 }
{ |\gamma|^2 \lambda_1 + |\delta|^2 \lambda_2 } \right)
+ 2 \log \left(1 + \frac{ |\gamma|^2 \mu_1 + |\delta|^2 \mu_2 }
{|\alpha|^2 \mu_1 + |\beta|^2 \mu_2 } \right).
\label{thm1_eqnimp}
\end{align}
Note that $\{ \theta_{\bullet} \}$ enter the above optimization only via
the term $| \beta \gamma - \alpha \delta|$ and as in~(\ref{eqn_ii3}), we have
\begin{eqnarray}
| \beta \gamma - \alpha \delta| \leq |\beta||\gamma| + |\alpha| |\delta|
\leq 1.
\end{eqnarray}
Given that $\chi_1 = \frac{\lambda_1}{\lambda_2} \geq 1$, three possibilities arise
depending on the relationship between $1$, $\chi_1$ and $\chi_2 =
\frac{\mu_1}{\mu_2}$: i) $\chi_1 >
1 \geq \chi_2$, ii) $\chi_1 > \chi_2 > 1$, and iii) $\chi_2 \geq
\chi_1 > 1$.

\noindent {\bf \em Case i):} In the first case where $\chi_2 \leq 1$,
we use the fact that $g(\bullet)$ is an increasing function to bound the
sum-rate as
\begin{align}
& 2 E \left[ R_1 \right] + 2 E \left[ R_2 \right] + 4 \log(2)
\nonumber \\
& \hsp \hsp \leq f \left( \frac{
\sqrt{ 4 \lambda_1 \lambda_2 } \cdot \sin(\theta + \phi) }
{  (\sin^2(\theta) + \sin^2(\phi)) (\lambda_1 - \lambda_2) +
2 \lambda_2 } \right) +
f \left( \frac{ \sqrt{4 \mu_1 \mu_2} \cdot \sin(\theta + \phi) }
{ 2 \mu_2 - (\sin^2(\theta) + \sin^2(\phi)) (\mu_2 - \mu_1) }
\right) \nonumber \\
& \hsp \hsp \hsp \hsp {\hspace{0.2in}}
+ 2 \log \left( \frac{ \sqrt{4 \lambda_1 \lambda_2}
\cdot \sin(\theta + \phi) }{ \sin^2(\phi) (\lambda_1 - \lambda_2)
+ \lambda_2 } \right) + 2 \log \left( \frac{ \sqrt{4 \mu_1 \mu_2} \cdot
\sin(\theta + \phi) } { \mu_2 - \sin^2(\theta)(\mu_2 - \mu_1) }
\right).
\end{align}
Observing that
\begin{eqnarray}
(\sin^2(\theta) + \sin^2(\phi)) (\lambda_1 - \lambda_2) +
2 \lambda_2 & \leq & \sin^2(\phi) (\lambda_1 - \lambda_2) +
\lambda_1 + \lambda_2 \\
2 \mu_2 - (\sin^2(\theta) + \sin^2(\phi)) (\mu_2 - \mu_1)
& \leq &
2 \mu_2 - \sin^2(\theta)(\mu_2 - \mu_1)
\end{eqnarray}
and $f(\bullet)$ is a decreasing function, we have
\begin{align}
& 2 E \left[ R_1 \right] + 2 E \left[ R_2 \right] + 4 \log(2)
\nonumber \\
& \hsp \hsp \leq f \left( \frac{
\sqrt{ 4 \lambda_1 \lambda_2 } \cdot \sin(\theta + \phi) }
{ \sin^2(\phi) (\lambda_1 - \lambda_2) + \lambda_1 + \lambda_2 }
\right) +
f \left( \frac{ \sqrt{4 \mu_1 \mu_2} \cdot \sin(\theta + \phi) }
{ 2 \mu_2 - \sin^2(\theta)(\mu_2 - \mu_1) } \right)
\nonumber \\
& \hsp \hsp \hsp \hsp + 2 \log \left( \frac{ \sqrt{4 \lambda_1 \lambda_2}
\cdot \sin(\theta + \phi) }{ \sin^2(\phi) (\lambda_1 - \lambda_2)
+ \lambda_2 } \right) + 2 \log \left( \frac{ \sqrt{4 \mu_1 \mu_2} \cdot
\sin(\theta + \phi) } {  \mu_2 - \sin^2(\theta)(\mu_2 - \mu_1) }
\right) \\
& \hsp \hsp = g \left( \frac{
\sqrt{ 4 \lambda_1 \lambda_2 } \cdot \sin(\theta + \phi) }
{ \sin^2(\phi) (\lambda_1 - \lambda_2) + \lambda_1 + \lambda_2 }
\right) +
g \left( \frac{ \sqrt{4 \mu_1 \mu_2} \cdot \sin(\theta + \phi) }
{ 2 \mu_2 - \sin^2(\theta)(\mu_2 - \mu_1) } \right)
\nonumber \\
& \hsp \hsp \hsp \hsp + 2 \log \left( 1 + \frac{\lambda_1}
{ \sin^2(\phi)(\lambda_1 - \lambda_2) + \lambda_2 } \right) +
2 \log \left( 1 + \frac{ \mu_2} {\mu_2 - \sin^2(\theta)
(\mu_2 - \mu_1) } \right).
\label{rhs2}
\end{align}
It is straightforward to note that the right-hand side of~(\ref{rhs2})
is decreasing in $\sin^2(\phi)$ and increasing in $\sin^2(\theta)$. If in
addition, the condition that $\theta + \phi = \pi/2$ is satisfied, then
an upper bound on $E \left[R_1 \right] + E \left[ R_2 \right]$ can be
maximized. This results in the choice $\theta = \pi/2$ and $\phi = 0$ (that is, $|\alpha|
= 1$ and $|\gamma| = 0$) and with this choice, we have
\begin{eqnarray}
2 E \left[ R_1 \right] + 2 E \left[ R_2 \right] \leq
f \left( \frac{ \sqrt{4 \lambda_1 \lambda_2}} {\lambda_1 + \lambda_2}
\right) +
f \left( \frac{ \sqrt{4 \mu_1 \mu_2} } {\mu_1 + \mu_2} \right) +
\log \left( \frac{\lambda_1} {\lambda_2} \right) +
\log \left( \frac{\mu_2}{\mu_1} \right) .
\end{eqnarray}
It is also straightforward to check that the choice of ${\bf w}_1$ and
${\bf w}_2$ as in the statement of the theorem meets this upper
bound and is thus optimal.

\noindent {\bf \em Case ii):} In the second case where $\chi_1 >
\chi_2 > 1$, we start as in {\bf \em Case i)} and after optimization
over $\{ \theta_{\bullet} \}$, we can bound the sum-rate as
\begin{align}
& 2 E \left[ R_1 \right] + 2 E \left[ R_2 \right] + 4 \log(2)
\nonumber \\
& \hsp \hsp \leq f \left( \frac{2 \sqrt{\chi_1} \cdot \sin(\theta + \phi) }
{ (\chi_1 - 1) (\sin^2(\theta) + \sin^2(\phi) ) + 2 } \right) +
f \left( \frac{2 \sqrt{\chi_2} \cdot \sin(\theta + \phi) }
{ (\chi_2 - 1) (\sin^2(\theta) + \sin^2(\phi) ) + 2 } \right) \nonumber \\
& \hsp \hsp \hsp \hsp + 2 \log \left( \frac{ 2 \sqrt{\chi_1} \cdot
\sin( \theta + \phi) } { (\chi_1 - 1) \sin^2(\phi) + 1} \right) +
2 \log \left( \frac{ 2 \sqrt{\chi_2} \cdot
\sin( \theta + \phi) } { (\chi_2 - 1) \sin^2(\theta) + 1} \right)
\\
& \hsp \hsp = g \left( \frac{2 \sqrt{\chi_1} \cdot \sin(\theta + \phi) }
{ (\chi_1 - 1) (\sin^2(\theta) + \sin^2(\phi) ) + 2 } \right) +
g \left( \frac{2 \sqrt{\chi_2} \cdot \sin(\theta + \phi) }
{ (\chi_2 - 1) (\sin^2(\theta) + \sin^2(\phi) ) + 2 } \right) \nonumber \\
& \hsp \hsp \hsp \hsp + 2 \log \left(
\frac{ (\chi_1 - 1) (\sin^2(\theta) + \sin^2(\phi) ) + 2}
{ (\chi_1 - 1) \sin^2(\phi) + 1} \right) +
2\log \left(  \frac{ (\chi_2 - 1) (\sin^2(\theta) + \sin^2(\phi) ) + 2 }
{  (\chi_2 - 1) \sin^2(\theta) + 1  }  \right).
\label{rhs}
\end{align}
The joint dependence between $\theta$ and $\phi$ in the right-hand side
of~(\ref{rhs}) precludes the possibility of breaking down the double
variable optimization of~(\ref{rhs}) into a pair of single variable
optimizations.
That is, the technique from {\bf \em Case i)} fails here and this case
needs to be studied differently.

The proof in this case follows in three steps. In the first step, when
$\chi_1 > \chi_2$, we show that ${\cal L}_2$ (defined as below) is
maximized by $\theta = \pi/2$ and $\phi = 0$:
\begin{align}
{\cal L}_2 & \triangleq
\frac{ (\chi_1 - 1) (\sin^2(\theta) + \sin^2(\phi) ) + 2}
{ (\chi_1 - 1) \sin^2(\phi) + 1} \cdot
\frac{ (\chi_2 - 1) (\sin^2(\theta) + \sin^2(\phi) ) + 2 }
{  (\chi_2 - 1) \sin^2(\theta) + 1  } \\
& = \frac{ (\chi_1-1)(\chi_2-1) \left[ \sin^2(\theta) +
\sin^2(\phi) \right]^2 + 2 \left(\chi_1 + \chi_2 -2 \right)
\left[ \sin^2(\theta) + \sin^2(\phi) \right] + 4}
{ \left[  (\chi_1-1) \sin^2(\phi) + 1 \right]
\left[ (\chi_2-1) \sin^2(\theta) + 1 \right] }.
\label{rhs1}
\end{align}
\ignore{
=  1 + \frac{ (\chi_1 -1) \sin^2(\theta) + 1 }
{ (\chi_1-1) \sin^2(\phi) + 1} +
\frac{ (\chi_2-1) \sin^2(\phi) + 1} { (\chi_2-1) \sin^2(\theta) + 1}
+ \frac{ (\chi_1 -1) \sin^2(\theta) + 1 }
{ (\chi_1-1) \sin^2(\phi) + 1} \cdot
\frac{ (\chi_2-1) \sin^2(\phi) + 1} { (\chi_2-1) \sin^2(\theta) + 1}.
\end{align}
Since
\begin{align}
\frac{ (\chi_2 - 1) \sin^2(\phi) + 1 } { (\chi_2 - 1) \sin^2(\theta) + 1}
= \frac{ (\chi_1 - 1) \sin^2(\phi) + 1 }
{ (\chi_1 - 1) \sin^2(\theta) + 1 } +
\frac{ (\chi_1 - \chi_2) \left( \sin^2(\theta) - \sin^2(\phi) \right)  }
{ \left[ (\chi_2 - 1) \sin^2(\theta) + 1 \right] \cdot
\left[ (\chi_1 - 1) \sin^2(\theta) + 1 \right] },
\end{align}
we can write ${\cal L}$ as
\begin{align}
{\cal L} & = & \frac{ (\chi_1-1)(\chi_2-1) \left[ \sin^2(\theta) +
\sin^2(\phi) \right]^2 + 2 \left(\chi_1 + \chi_2 -2 \right)
\left[ \sin^2(\theta) + \sin^2(\phi) \right] + 4}
{  \left[ (\chi_2-1) \sin^2(\theta) + 1 \right]
\left[  (\chi_1-1) \sin^2(\phi) + 1 \right]  }.
\end{align}
}
For this, we set $\sin^2(\theta) + \sin^2(\phi)$ to take a specific value
of $\alpha$. Since the numerator is only a function of $\alpha$, maximizing
${\cal L}_2$ is equivalent to minimizing the denominator of~(\ref{rhs1}).
There are two possible cases depending on whether $\alpha \leq 1$ or
$1 < \alpha \leq 2$. In the former case, it can be seen that the
denominator is minimized by $\phi = 0$ and $\theta = \sin^{-1}
(\sqrt{\alpha})$, whereas in the latter case, it is minimized by
$\phi = \sin^{-1}( \sqrt{\alpha-1} )$ and $\theta = \pi/2$. Substituting
these values, it can be seen that
\begin{eqnarray}
{\cal L}_2 \leq
\left\{
\begin{array}{cc}
\frac{ (\chi_1-1)(\chi_2-1) \hsppp \alpha^2 +
2 \alpha( \chi_1 + \chi_2 - 2) + 4}
{ 1 + (\chi_2-1) \hsppp \alpha} & {\rm if} \hspp \alpha \leq 1 \\
\frac{ (\chi_1-1)(\chi_2-1) \hsppp \alpha^2 + 2
\alpha( \chi_1 + \chi_2 - 2) + 4}
{ \chi_2 \cdot \left[ 1 + (\chi_1-1) \hsppp (\alpha-1)  \right]} &
{\rm if} \hspp 1 < \alpha \leq 2.
\end{array}
\right.
\label{eqn_nn}
\end{eqnarray}
A straightforward derivative calculation (using the critical fact that
$\chi_1 > \chi_2$) shows that while the right-hand
side of~(\ref{eqn_nn}) is increasing for $\alpha \leq 1$, it is decreasing
for $1 < \alpha \leq 2$. 
In other words,
\begin{eqnarray}
{\cal L}_2 \leq \frac{ (1 + \chi_1) \cdot ( 1 + \chi_2)}{\chi_2},
\end{eqnarray}
and this upper bound is achieved with $\theta = \pi/2, \phi= 0$.

In the second step, if $\theta + \phi > \pi/2$, we have
\begin{eqnarray}
\frac{ \sin(\theta + \phi) }{ (\chi_i-1) \left( \sin^2(\theta) +
\sin^2(\phi) \right) + 2 } \leq
\frac{ \sin(\theta+ \phi)}{ \chi_i + 1} \leq \frac{1}{\chi_i + 1},
\hsp i = 1,2
\label{holds}
\end{eqnarray}
since $\sin(\theta) > \sin( \pi/2 - \phi) = \cos(\phi)$.
Therefore,
\begin{align}
& 2 E \left[ R_1 \right] + 2 E \left[ R_2 \right] + 4 \log(2)
\nonumber \\
& \hsp \hsp \leq g \left( \frac{2 \sqrt{\chi_1} \cdot \sin(\theta + \phi) }
{ (\chi_1 - 1) (\sin^2(\theta) + \sin^2(\phi) ) + 2 } \right) +
g \left( \frac{2 \sqrt{\chi_2} \cdot \sin(\theta + \phi) }
{ (\chi_2 - 1) (\sin^2(\theta) + \sin^2(\phi) ) + 2 } \right)
+ 2 \log( {\cal L}_2) \\
& \hsp \hsp \leq g \left( \frac{2 \sqrt{\chi_1}}
{\chi_1+1} \right) + g \left( \frac{2 \sqrt{\chi_2}}
{\chi_2+1} \right) + 2 \log( {\cal L}_2), \\
& 2 E \left[ R_1 \right] + 2 E \left[ R_2 \right] \leq
f \left( \frac{2 \sqrt{\chi_1}} {\chi_1+1} \right)
+ f \left( \frac{2 \sqrt{\chi_2}} {\chi_2+1} \right)
+ \log \left(  \frac{\chi_1}{\chi_2} \right),
\end{align}
where the second inequality follows from the monotonicity of $g( \bullet )$
and the third from Step 1.

Note that~(\ref{holds}) fails if $\theta + \phi = \nu \leq \pi/2$.
Thus, in the third step, we consider this possibility. Here, a
straightforward manipulation shows that
\begin{eqnarray}
\sin^2(\theta) + \sin^2(\nu - \theta)
& = & \sin^2(\nu) - 2 \sin(\theta) \sin(\nu -\theta) \cos(\nu)
\leq \sin^2(\nu),
\end{eqnarray}
where the last inequality follows because $\nu \leq \pi/2$. Hence, we have
\begin{align}
& {\hspace{-0.2in}}
2 E \left[ R_1 \right] + 2 E \left[ R_2 \right] + 4 \log(2) \leq
g \left( \frac{ 2 \sqrt{\chi_1} \sin(\nu)}
{ (\chi_1-1) \sin^2(\nu)  + 2 }  \right) +
g \left( \frac{ 2 \sqrt{\chi_2} \sin(\nu)}
{ (\chi_2-1) \sin^2(\nu)  + 2 } \right)
\nonumber \\ & \hsp \hsp \hsp {\hspace{0.8in}}
+ 2\log \left( \frac{ \left[ (\chi_1-1) \sin^2(\nu)  + 2 \right]
\cdot
\left[ (\chi_2-1) \sin^2(\nu)  + 2 \right] }
{ \left[ (\chi_1-1) \sin^2(\nu-\theta) + 1 \right] \cdot
\left[ (\chi_2-1) \sin^2(\theta) +1 \right]} \right) \triangleq {\cal L}_3.
\end{align}
It can be easily seen that ${\cal L}_3$ is maximized by $\theta = \pi/2$ and
$\phi = 0$. The upper bound is the same in both the cases
$\theta + \phi \leq \pi/2$
and $\theta + \phi > \pi/2$. And this upper bound is met by the choice
$\theta = \pi/2$ and $\phi = 0$ (that is, $|\alpha| = 1$ and
$|\gamma| = 0$) and is hence optimal.

\noindent {\bf \em Case iii):} Since the expression for the sum-rate
is symmetric in $\chi_1$ and $\chi_2$, an argument analogous to
{\bf \em Case ii)} completes the theorem in the case $\chi_1 \leq \chi_2$.

The optimal sum-rate in all the three cases is given by the unified
expression
\begin{eqnarray}
2 E \left[ R_1 \right] + 2 E \left[ R_2 \right]
\stackrel{\rho \rightarrow \infty}{\rightarrow}
f \left( \frac{ 2 \sqrt{\chi_1} } {\chi_1 + 1} \right) +
f \left( \frac{ 2 \sqrt{\chi_2} } {\chi_2 + 1} \right) +
\Big| \log \left( \chi_1 \right) - \log \left( \chi_2
\right) \Big|
\end{eqnarray}
where $f(\bullet)$ is as defined in~(\ref{fstruct}). This expression
can be simplified as in the statement of the theorem.
\endproof

\subsection{Comparison of Proof Techniques of Theorems~\ref{thm2}
and~\ref{thm1}}
\label{app_comp}
We first show that Theorem~\ref{thm2} reduces to Theorem~\ref{thm1}
under the assumption that the eigenvectors of ${\bf \Sigma}_1$ and
${\bf \Sigma}_2$ coincide. For this, we set
\begin{eqnarray}
\alpha' = \frac{ \alpha \sqrt{\tau_1} } {X}, \hsp \hsp
\beta' = \frac{ \beta \sqrt{\tau_2} } {X}, &&
\gamma' = \frac{ \gamma \sqrt{\tau_1} } {Y}, \hsp \hsp
\delta' = \frac{ \delta \sqrt{\tau_2} } {Y}
\label{transf}
\end{eqnarray}
where $X$ and $Y$ are as in~(\ref{Xsquared}) and~(\ref{Ysquared}),
respectively. Note that the above transformation is a bijection from the
space $\Big\{ \alpha, \beta, \gamma, \delta \hsppp : \hsppp |\alpha|^2 +
|\beta|^2 = 1 = |\gamma|^2 + |\delta|^2 \Big\}$ to the space $\Big\{
\alpha', \beta', \gamma', \delta' \hsppp : \hsppp |\alpha'|^2 +
|\beta'|^2 = 1 = |\gamma'|^2 + |\delta'|^2 \Big\}$. With this transformation,
it can be checked that~(\ref{thm2_eqnimp}) reduces to~(\ref{thm1_eqnimp}).
It can also be seen that the sum-rate expression in~(\ref{rhs3}) reduces
to that in~(\ref{rhs4}) in both cases.

The technique pursued in the general case diverges from that in
Appendix~\ref{app_thm1} in two ways.

\noindent {\bf \em Difference 1:} It can be easily checked that
$\tau_3 = 0$ if and only if the set of eigenvectors of ${\bf \Sigma}_1$
and ${\bf \Sigma}_2$ coincide. In general, $\tau_3 \neq 0$ and
${\sf arg}(\tau_3)$ could affect the sum-rate optimization. The first
step in Appendix~\ref{app_gen} is to show that this is not the case and
${\sf arg}(\tau_3)$ plays no role in the optimization.
This is done by exploiting the fact that the sum-rate optimization
(see~(\ref{prob_defn})) is a problem over ${\cal G}(2,1)$ and not over
${\sf St}(2,1).$

\noindent {\bf \em Difference 2:} The second complication
is that there is a definitive (and easily classifiable) comparative
relationship between
$\tau_1 = {\bf v}_1^H \hsppp {\bf \Sigma}_2^{-1} \hsppp {\bf v}_1$ and
$\tau_2 = {\bf v}_2^H \hsppp {\bf \Sigma}_2^{-1} \hsppp {\bf v}_2$ in the
special case. This comparative relationship does not generalize to the
setting where the eigenvectors of ${\bf \Sigma}_1$ and ${\bf \Sigma}_2$
are different.

Specifically, under Case i) of the discussion in Theorem~\ref{thm1},
$\tau_1 > \tau_2$ if and only if $\chi_2 < 1$ whereas under Case ii),
$\tau_1 > \tau_2$ if and only if $\chi_2 > 1$. On the other hand, in
the general case, all the three possibilities: i) $\tau_1 > \tau_2$,
ii) $\tau_1 < \tau_2$, iii) $\tau_1 = \tau_2$ can occur for appropriate
choices of ${\bf \Sigma}_1$ and ${\bf \Sigma}_2$. We now illustrate this
with a numerical example. Let ${\bf \Sigma}_2$ be fixed such that
\begin{eqnarray}
{\bf \Sigma}_2 = \left[
\begin{array}{cc}
1 & -0.6897 \\
-0.6897 & 1
\end{array}
\right].
\end{eqnarray}
With the choice
\begin{eqnarray}
{\bf \Sigma}_1 = \left[
\begin{array}{cc}
1 & 0.8 \\
0.8 & 1
\end{array} \right],
\end{eqnarray}
it can be seen that $\eta_1 = 5.8, \eta_2 = 0.1184, \tau_1 = 3.2222$ and
$\tau_2 = 0.5918$, whereas if
\begin{eqnarray}
{\bf \Sigma}_1 = \left[
\begin{array}{cc}
1 & -0.8 \\
-0.8 & 1
\end{array} \right],
\end{eqnarray}
it can be seen that $\eta_1 = 1.0653, \eta_2 = 0.6444, \tau_1 =0.5918$ and
$\tau_2 = 3.2222$. It can be seen that $\eta_1 = 1.4603, \eta_2 =
0.5397$ and $\tau_1 = \tau_2 = 1.90725$ if ${\bf \Sigma}_1$ satisfies
\begin{eqnarray}
{\bf \Sigma}_1 = \left[
\begin{array}{cc}
\frac{2}{3} & -0.34485 \\
-0.34485 & \frac{1}{3}
\end{array} \right].
\end{eqnarray}

\ignore{
In the special case where the eigenvectors of ${\bf \Sigma}_1$ and
${\bf \Sigma}_2$ coincide with $\chi_1 \geq \chi_2$ (Case i) of the
discussion in Theorem~\ref{thm1}), we have from~(\ref{defn1}) that
\begin{eqnarray}
{\bf v}_1 = {\bf u}_1 \hsp {\rm and} \hsp {\bf v}_2 = {\bf u}_2.
\end{eqnarray}
Further, $\kappa_1 = \chi_1$, $\kappa_2 = \chi_2$ and the quantities
defined in~(\ref{defn2}) and~(\ref{k12gen}) take the following values:
\begin{eqnarray}
\eta_1 = \frac{\lambda_1}{\mu_1} , \hsp \hsp
\eta_2 = \frac{ \lambda_2}{\mu_2}, \hsp
\hsp \tau_1 = \frac{1}{\mu_1} , \hsp \hsp \tau_2 = \frac{1}{\mu_2},
\hsp \hsp \tau_3 = 0.
\label{defn4}
\end{eqnarray}
In Case ii) ($\chi_1 < \chi_2$), we have
\begin{eqnarray}
{\bf v}_1 = {\bf u}_2 \hsp {\rm and} \hsp {\bf v}_2 = {\bf u}_1.
\end{eqnarray}
Further, $\kappa_1 = \frac{1}{\chi_1}$, $\kappa_2 = \frac{1}{\chi_2}$
with
\begin{eqnarray}
\eta_1 = \frac{\lambda_2}{\mu_2} , \hsp \hsp
\eta_2 = \frac{ \lambda_1}{\mu_1}, \hsp \hsp
\tau_1 = \frac{1}{\mu_2} , \hsp \hsp \tau_2 = \frac{1}{\mu_1},
\hsp \hsp \tau_3 = 0.
\label{defn3}
\end{eqnarray}

The

More

To illustrate,
}
These differences imply that, in general, there exists no bijective
transformation (as in~(\ref{transf})) to transform the objective function
from the form in~(\ref{thm2_eqnimp}) to that in~(\ref{thm1_eqnimp}).
Despite these issues, it would be of interest to pursue a theme that
could unify the general case with the special case.
\endproof

\subsection{Proof of Prop.~\ref{prop1}}
\label{app_prop1}
The proof in the low-$\snr$ extreme is obvious. In the high-$\snr$ extreme,
we first note that the optimization problem over the choice of a pair
$\left( {\bf w}_1, {\bf w}_2 \right)$ that results in a corresponding
choice of $\left( A_i, B_i, C_i \right)$ can be recast in the form of a
two parameter optimization problem over $(M_i, N_i)$ with
$M_i = \frac{A_i}{B_i}$ and $N_i = \frac{C_i}{B_i}$ under the constraint
that
\begin{eqnarray}
0 \leq N_i^2 \leq M_i \leq \chi_i = \frac{\lambda_1( {\bf \Sigma}_i )}
{ \lambda_2( {\bf \Sigma}_i)}.
\end{eqnarray}
For this, observe that for any given choice of $\left( {\bf w}_1,
{\bf w}_2 \right)$, the resultant $\left( A_i, B_i, C_i \right)$ has to
satisfy $C_i^2 \leq A_i B_i$ (Cauchy-Schwarz inequality) and
$\frac{A_i}{B_i} \leq \chi_i$ 
(Ritz-Raleigh ratio)~\cite{cr_rao}. 
Thus, we have
\begin{eqnarray}
\max \limits_{ {\bf w}_1, \hsppp {\bf w}_2 } \lim_{\rho \rightarrow
\infty} E \left[ R_i \right] \leq
\max \limits_{ 0 \hsppp \leq \hsppp N_i^2 \hsppp \leq  \hsppp
M_i \hsppp \leq \hsppp \chi_i 
} \lim_{\rho \rightarrow \infty} E \left[R_i \right].
\end{eqnarray}

Since the high-$\snr$ expression for $E \left[ R_i \right]$ satisfies
\begin{eqnarray}
2 E \left[ R_i \right] + 2 \log(2)
& 
= & g\left( \frac{2 \sqrt{ M_i - N_i^2} } { M_i+1} \right) +
2 \log \left(1 + M_i \right),
\end{eqnarray}
and $g(\bullet)$ is an increasing function, optimization over $N_i$
which affects only the first term on the right-hand side implies that
the optimal choice of $N_i$ is zero. Plugging this choice and using
the structure of $g(\bullet)$ and $f(\bullet)$ from~(\ref{gstruct})
and~(\ref{fstruct}) respectively, we have
\begin{eqnarray}
E \left[ R_i \right] \leq \frac{M_i}{ |M_i - 1|} \cdot |\log(M_i)|.
\label{ref3}
\end{eqnarray}
The right-hand side of~(\ref{ref3}) 
is increasing in $M_i$ since
the derivative function satisfies
\begin{eqnarray}
\frac{d \frac{ M_i \cdot |\log(M_i)|}{ |M_i - 1| }}{ dM_i}
= \left\{
\begin{array}{cc}
\frac{ M_i - 1 - \log(M_i)}{ (M_i - 1)^2 } & {\rm if} \hsppp M_i > 1
\\
\frac{1}{2} & {\rm if} \hsppp M_i = 1 \\
\frac{ \log \left( \frac{1}{M_i} \right) - (1 - M_i)}
{ (1 - M_i)^2 } & {\rm if } \hsppp M_i < 1,
\end{array}
\right.
\end{eqnarray}
where all the three pieces are positive and hence, the derivative function is
smooth. The goal of maximizing $M_i$ under the constraints of
$N_i = 0$ and unit-normedness of ${\bf w}_1$ and ${\bf w}_2$ is met
by the choice as in the statement of the proposition. This upper bound
to $E \left[  R_i \right]$ is also met by the same choice of beamforming
vectors and this choice is thus optimal.

Recall from~(\ref{ri_highsnr}) (the high-$\snr$ expression) that
$E \left[R_i \right]$ is the sum of two terms. The increasing nature of
$g(\bullet)$ means that the first term of~(\ref{ri_highsnr}) is maximized
when $d_{ {\bf \Sigma}_i } ( {\bf w}_1, {\bf w}_2 )$ is maximized.
That is, by the
choice  $\{ {\bf w}_1, {\bf w}_2 \}$ as in~(\ref{eq_w1})-(\ref{eq_w2}). On
the other hand, the second term as well as
$E \left[R_i \right]$ (which is the sum of the two terms)
are maximized by the choice in~(\ref{opt_prop3}). With this choice of
beamforming vectors, $d_{ {\bf \Sigma}_i } (\cdot, \cdot)$ can be written
as
\begin{eqnarray}
d_{ {\bf \Sigma}_i } \left( {\bf w}_{i, \hsppp {\sf opt}},
{\bf w}_{j, \hsppp {\sf opt}} \right)
& = & \frac{2 \sqrt{ \chi_i}}{ \chi_i+1 }.
\end{eqnarray}
Note that $d_{ {\bf \Sigma}_i } (\cdot, \cdot)$ decreases as
$\chi_i$ increases.
\endproof

\bibliographystyle{IEEEbib}
\bibliography{newrefsx}

\ignore{

}

\end{document}